\documentclass[times, 3p, authoryear]{elsarticle}
\usepackage[UKenglish]{babel}
\usepackage[utf8]{inputenc}
\usepackage{amsmath, amssymb, amsthm}
\usepackage{mathtools}
\usepackage{mathrsfs}
\usepackage{nicefrac}
\usepackage[inline]{enumitem}
\usepackage{color}
\usepackage{tikz}
\usepackage{hyperref}
\usetikzlibrary{arrows.meta, decorations, decorations.markings, decorations.pathreplacing}
\usetikzlibrary{trees, shapes.geometric, backgrounds, calc, positioning}

\usepackage{array}
\usepackage{booktabs}
\newcommand{\otoprule}{\midrule[\heavyrulewidth]}
\newcolumntype{L}{>{$}l<{$}}

\numberwithin{equation}{section}

\newtheorem{thm}{Theorem}
\newtheorem{lem}[thm]{Lemma}
\newtheorem{cor}[thm]{Corollary}
\newtheorem{prop}[thm]{Proposition}
\newlist{theonum}{enumerate}{2}
\setlist[theonum]{topsep=0pt, itemsep=0pt, partopsep=0pt, parsep=0pt, label=\roman*), ref=\thethm.\roman*}
\newdefinition{defn}{Definition}

\newcommand*{\mat}[1]{\boldsymbol{#1}}

\newcommand*{\crea}[1]{\hat{#1}^{\dagger}}
\newcommand*{\anni}[1]{\hat{#1}^{\vphantom{\dagger}}}

\newcommand*{\coord}[1]{\mathbf{#1}}
\newcommand*{\vecr}{\coord{r}}
\newcommand*{\vecx}{\coord{x}}

\newcommand*{\binteg}[3]{\int^{\mathrlap{#3}}_{\mathrlap{#2}}\ud{#1}\,}
\newcommand*{\integ}[1]{\!\int\!\ud{#1}\:}
\newcommand*{\iinteg}[2]{\integ{#1}\!\!\integ{#2}}

\DeclareMathOperator{\arctanh}{arctanh}
\DeclareMathOperator{\diag}{diag}
\DeclareMathOperator{\trace}{tr}
\DeclareMathOperator{\Trace}{Tr}

\DeclarePairedDelimiter{\abs}{\lvert}{\rvert}
\DeclarePairedDelimiter{\av}{\langle}{\rangle}
\DeclarePairedDelimiter{\norm}{\lVert}{\rVert}
\DeclarePairedDelimiterX\braket[2]{\langle}{\rangle}{#1\delimsize\vert#2}
\DeclarePairedDelimiterX\brakket[3]{\langle}{\rangle}{#1\delimsize\vert#2\delimsize\vert#3}
\DeclarePairedDelimiterX\ket[1]{\lvert}{\rangle}{#1}
\DeclarePairedDelimiterX\bra[1]{\langle}{\rvert}{#1}
\DeclarePairedDelimiterX\set[2]{\lbrace}{\rbrace}{#1 : #2}

\newcommand*{\closedBall}{\overline{B}}			
\newcommand*{\CompactOperators}{\mathfrak{K}}	
\newcommand*{\Complex}{\mathbb{C}}			
\newcommand*{\FiniteNpot}{\mathcal{V}}			
\newcommand*{\Fock}{\mathcal{F}}				
\newcommand*{\preFock}{\mathfrak{F}}			
\newcommand*{\HermitanMat}{\mathbb{H}}		
\newcommand*{\Nats}{\mathbb{N}}				
\newcommand*{\Nbas}{\Nats_b}				
\newcommand*{\oneH}{\mathcal{H}}				
\newcommand*{\NTdensMat}{\mathcal{P}}		
\newcommand*{\NdensMat}{\overline{\NTdensMat}}	
\newcommand*{\ToneMat}{\mathscr{N}}			
\newcommand*{\NoneMat}{\overline{\ToneMat}}	
\newcommand*{\VoneMat}{\mathscr{V}}			
\newcommand*{\TraceClass}{\mathfrak{T}}		

\newcommand*{\domain}{\textrm{dom}}			
\newcommand*{\epigraph}{\textrm{epi}}			
\newcommand*{\interior}{\textrm{int}}				
\newcommand*{\interiorOf}{\interior\,}			

\newcommand*{\du}{\partial}

\newcommand*{\e}{\mathrm{e}}
\newcommand*{\half}{\frac{1}{2}}
\newcommand*{\im}{\textrm{i}}
\newcommand*{\isDefinedAs}{\coloneqq}
\newcommand*{\nfrac}[2]{\nicefrac{#1}{#2}}
\newcommand*{\nhalf}{\nfrac{1}{2}}
\newcommand*{\nmax}{\mathfrak{n}} 
\newcommand*{\perm}{\wp}
\newcommand*{\Reals}{\mathbb{R}}
\newcommand*{\sAdenifeDsI}{\eqqcolon}

\newcommand*{\suchthat}{:}
\newcommand*{\thalf}{\tfrac{1}{2}}
\newcommand*{\ud}{\mathrm{d}}
\newcommand*{\unitMat}{\mat{1}}

\newcommand*{\weakstarly}{\stackrel{*}{\rightharpoonup}}

\journal{Physics Reports}

\begin{document}

\begin{frontmatter}
\title{One-body reduced density-matrix functional theory in finite basis sets at elevated temperatures}

\author{Klaas J. H. Giesbertz}
\ead{klaas.giesbertz@gmail.com}
\address{Theoretical Chemistry, Faculty of Exact Sciences, VU University, De Boelelaan 1083, 1081 HV Amsterdam, The Netherlands}

\author{Michael Ruggenthaler}
\ead{Michael.Ruggenthaler@mpsd.mpg.de}
\address{Theory Department, Max Planck Institute for the Structure, and Dynamics of Matter and Center for Free-Electron Laser Science, Luruper Chaussee 149, 22761 Hamburg, Germany}

\date{\today}

\begin{abstract}
In this review we provide a rigorous and self-contained presentation of one-body reduced density-matrix (1RDM) functional theory. We do so for the case of a finite basis set, where density-functional theory (DFT) implicitly becomes a 1RDM functional theory. To avoid non-uniqueness issues we consider the case of fermionic and bosonic systems at elevated temperature and variable particle number, i.e, a grand-canonical ensemble. For the fermionic case the Fock space is finite-dimensional due to the Pauli principle and we can provide a rigorous 1RDM functional theory relatively straightforwardly. For the bosonic case, where arbitrarily many particles can occupy a single state, the Fock space is infinite-dimensional and mathematical subtleties (not every hermitian Hamiltonian is self-adjoint, expectation values can become infinite, and not every self-adjoint Hamiltonian has a Gibbs state) make it necessary to impose restrictions on the allowed Hamiltonians and external non-local potentials. For simple conditions on the interaction of the bosons a rigorous 1RDM functional theory can be established, where we exploit the fact that due to the finite one-particle space all 1RDMs are finite-dimensional. We also discuss the problems arising from 1RDM functional theory as well as DFT formulated for an infinite-dimensional one-particle space.
\end{abstract}

\end{frontmatter}

\newpage

\tableofcontents

\newpage

\fbox{\begin{tabular}{Ll}
$\textbf{List of symbols}$ 	& \\
					& \\
E					&energy~\eqref{def:Energy} \\
F[\gamma]			&universal functional~\eqref{def:universalFunction} \\
\nmax			&maximum order of the interaction in the Hamiltonian~\eqref{eq:ncHam} \\
N_b					&number of one-body states, i.e.\ dimension of $\oneH$ \\
S					&entropy~\eqref{def:entropy} \\
T					&temperature \\
v					&(non-local) one-body potential (matrix) \\
Z					&partition function~\eqref{eq:equiRho} \\
\beta					&inverse temperature, i.e.\ $1/T$ \\
\gamma				&one-body reduced density matrix (1RDM) \\
\phi, \psi				&state in one-body space $\oneH$ \\
\Phi, \Psi				&state in Fock space $\Fock$ \\
\Omega				&grand potential~\eqref{def:grandPotential} \\
					& \\
\textbf{Operators}		& \\
\diag(a_i)				&diagonal matrix with elements $a_i$ on its diagonal \\
\hat{H}				&Hamiltonian acting in Fock space \\
\hat{N}				&number operator acting in Fock space \\
\hat{V}				&(non-local) one-body potential acting in the Fock space \\
\hat{T}				&kinetic energy operator acting in the Fock space \\
\hat{\rho}				&density-matrix operator~\eqref{eq:DensityMatrixOp} \\
\hat{\rho}_v		& equilibrium density-matrix operator for a potential $v$~\eqref{eq:equiRho} \\
\trace\{\cdot\}			&trace of $N_b$-dimensional vector space \\
\Trace\{\cdot\}			&trace in the Fock space \\
\av{\cdot}				&expectation value \\
\braket{\cdot}{\cdot}		&inner product \\
\norm{\cdot}			&norm \\
					& \\
\textbf{Sets \& spaces}	& \\
\closedBall_{\epsilon}(x)	&closed ball of radius $\epsilon$ centred at point $x$ \\
\preFock_{\pm}			&pre-Fock space, containing only vectors of finite length \\
\Fock_{\pm}			&Fock space~\eqref{def:FockSpace}, i.e.\ completion of $\preFock_{\pm}$ \\
\HermitanMat(n)		&space of $n \times n$ hermitian matrices~\eqref{def:hermitianMatrix} \\
\oneH				&one-body Hilbert space \\
\Nbas				&index set of one-body states, typically $\{1, \dotsc, N_b\}$ \\
\NoneMat_{\pm}		&set of ensemble $N$-representable 1RMDs~\eqref{def:NoneMat} \\
\ToneMat_{\pm}			&interior of $\NoneMat_{\pm}$~\eqref{defs:ToneMats} \\
\NdensMat_{\pm}		&set of all density-matrix operators~\eqref{def:NdensMat} \\
\NTdensMat_{\pm}	&set of all finite temperature density-matrix operators~\eqref{def:NTdensMat} \\
\TraceClass			&space of trace-class operators~\eqref{def:TraceClass} \\
\FiniteNpot			&set of potentials yielding a proper Gibbs state \\
\VoneMat_{\pm}		&set of $v$-representable 1RDMs~\eqref{def:VoneMat} \\
\end{tabular}}

\newpage

\section{Introduction}
\label{sec:intro}

The main challenge in most areas of quantum physics and quantum chemistry is to solve equations that describe many interacting particles. This central challenge of modern physics is called the quantum many-body problem. It arises rather inconspicuous due to the way we construct a many-body quantum theory from a single-particle description. In quantum mechanics, for instance, this is usually done by starting from a single electron in real space. We describe this single electron by a normalised wave function $\varphi(\vecr)$ that solves a linear equation on the Hilbert space of square-integrable functions \citep{teschl2014}, e.g., the Schrödinger equation for the hydrogen atom.

For the description of a two-particle problem we want to ensure that the properties of the single-particle theory are kept intact. For this we just give every particle its own ``real space''. Taking into account indistinguishability and the fundamental property of spin leads for the two-electron problem to a wave function in its simplest form\footnote{\label{fn:norm}We employ an unconventional, yet more consistent normalisation, as explained in more detail in Sec.~\ref{sec:FockSpace}. Our normalisation guarantees the probability interpretation of (the modulus of) the underlying many-body wave function \citep{StefanucciLeeuwen2013}.}
\begin{equation}
\Phi(\vecr_1, s_1; \vecr_2, s_2) = \varphi(\vecr_1)\chi_1(s_{1})\,\varphi(\vecr_2)\chi_2(s_{2}) - \varphi(\vecr_2)\,\chi_1(s_{2})\varphi(\vecr_1)\chi_2(s_{1}) \, ,
\end{equation}
where $\chi_{1/2}(s)$ are spin wave functions for spin $s$. A general two-electron problem is then described by a wave function of the form $\Psi(\vecx_1,\vecx_2)$ where we denote $\vecx \isDefinedAs (\vecr, s)$. To determine the wave function of an interacting two-particle problem, e.g., the ground state of a neutral hydrogen molecule (H$_2$), we have to represent the problem on a computer. We can do so either by discretising real space, i.e., that we represent continuous real space by a grid of discrete points, or by some appropriate single-particle basis. For instance, to find a good basis for H$_{2}$ we could choose an $s$-type electronic orbital $s_a(\vecr)$ at the position of the first nucleus and one $s_b(\vecr)$ at the position of the second and then define symmetry-adapted basis functions $\sigma_{g/u}(\vecr) = \bigl(s_a(\vecr) \pm s_b(\vecr)\bigr) / \sqrt{2(1 \pm \braket{s_a}{s_b})}$. Then by Gram--Schmidt orthogonalisation we can construct further functions that all together constitute an orthonormal basis for the single-particle Hilbert space. Either way, an accurate representation of the wave function usually forces us to use many grid points or basis functions and thus if we need to store for each particle $M$ entries, the amount of data we have to handle is roughly $M^{2}$ bytes. If we have more than two particles this grows exponentially with the number of particles $N$, i.e., we need to handle $M^N$ bytes to work with many-body wave functions. Even with nowadays supercomputers we can only treat relatively small systems without further tricks. Therefore tremendous effort has been put into developing methods that make numerical calculations for complex many-body systems feasible. Many methods try to find efficient and accurate approximations to the many-body wave functions such as tensor-network approaches \citep{Schollwoeck2011, Orus2014}, coupled-cluster theory \citep{Bartlett2007} or quantum Monte-Carlo techniques \citep{gubernatis2016}.

A different route is to change from the exponentially-scaling many-body wave function as the fundamental description of the multi-particle problem to an equivalent, yet reduced quantity. This is the basic idea behind density-functional theories (DFT) \citep{DreizlerGross1990,Eschrig1996, Eschrig2003}, density-matrix theories \citep{Cioslowski2000, MazziottiBook2007, PernalGiesbertz2015, bonitz2015} and Green's function techniques \citep{FetterWalecka1971, StefanucciLeeuwen2013}. While it is numerically demanding to calculate Green's functions, this approach has the advantage that it is in principle easy to increase the accuracy of the calculated Green's function by including higher-order Feynman diagrams \citep{FetterWalecka1971, StefanucciLeeuwen2013}. On the other hand, in DFT it is relatively simple to numerically calculate the one-body density but it is demanding to systematically increase the accuracy \citep{burke2012}. This is due to the fact that the many-body energy, which is the central object in ground state DFT, is a very implicit functional of the density or the auxiliary Kohn-Sham (KS) wave functions. In this respect reduced density-matrix (RDM) functional theories are an interesting compromise. For one-body RDM (1RDM) functional theory the current and kinetic energy of the many-body energy becomes explicit and for two-body RDM (2RDM) even the two-body interaction energy becomes an explicit functional. The drawback of RDM theories, however, is that in contrast to DFT it is very hard to guarantee that some arbitrary RDM is connected to a specific many-body Hamiltonian or even just an arbitrary many-body wave function. These representability as well as other subtle mathematical problems \citep{Coleman1963, Coleman1987, Klyachko2006, AltunbulakKlyachko2008, AggelenVerstichel2010, Mazziotti2012} have hampered the development and applicability of RDM functional theories. 

To overcome these problems and provide a sound mathematical foundation for further developments of 1RDM functional theory, we present in this review a rigorous formulation in finite basis sets at elevated temperature and arbitrary particle numbers as well as statistics, i.e., for fermions and bosons.
The mathematical reason to work at elevated temperatures is that we avoid non-uniqueness problems which are present in a zero temperature formalism. It is obvious that this is also a very physical choice, as in many experiments temperature effects play a significant role. Important examples are metal-insulator transitions in transition metal oxides \citep{YooMaddoxKlepeis2005, RueffMattilaBadro2005, PattersonAracneJackson2004, MitaIzakiKobayashi2005, MitaSakaiIzaki2001, NoguchiKusabaFukuoka1996}, (high $T_c$) superconductors \citep{NagamatsuNakagawaMuranaka2001, BednorzMuller1986} and protein folding \citep{Anfinsen1972, TakaiNakamuraToki2008, NichollsSharpHonig1991}.
More extreme examples are rapid heating of solids via strong laser fields \citep{GavnholtRubioOlsen2009}, dynamo effect in giant planets \citep{RedmerMattssonNettelmann2011}, shock waves \citep{RootMagyarCarpenter2010, Militzer2006}, warm dense matter \citep{KietzmannRedmerDesjarlais2008} and hot plasmas \citep{Dharma-wardanaPerrot1982, PerrotDharma-wardana2000, Dharma-wardanaMurillo2008}.
 
The choice to formulate the theory in a finite basis set is mathematically motivated to make a rigorous treatment of the grand-potential relatively simple and to establish differentiability of the involved functionals. Apart from this mathematical convenience, the finite basis, which is the usual situation in any practical numerical calculation, has a more immediate consequence for the many users of DFT. In this case DFT implicitly becomes a 1RDM functional theory. Thus this review also provides the necessary foundation for approximate DFT calculations.

\section{Theoretical motivations for the setting}
\label{sec:theoMotive}

\subsection{1RDM functional theory in disguise: DFT in finite basis sets}

One of the problems which arises in practical DFT is that one often needs to use finite basis sets for calculations. Unfortunately DFT is not well defined for finite basis sets \citep[in the accompanying statements of][]{Nooijen1992},\footnote{We point out that in a grid basis this is not the case, since there one can rely on lattice DFT~\citep{ChayesChayesRuskai1985}.} so these calculations can lead to pathological problems as is well known in the optimised-effective potential approach to the KS potential \citep{Gorling1999, KollmarFilatov2008, Jacob2011, GidopoulosLathiotakis2012, BetzingerFriedrichGorling2012}. Let us demonstrate how a finite basis is typically problematic with a simple example. We consider the exact solution for the ground state of the neutral H$_2$ problem from above in the minimal basis $\{\sigma_g(\vecr), \sigma_u(\vecr)\}$. In this case the exact ground state becomes with $\sigma_{k, l}(\vecx) = \sigma_{k}(\vecr)\chi_{l}(s)$\footnote{See footnote~\ref{fn:norm}.}
\begin{equation}
\Psi(\vecx_1,\vecx_2) = c_g\bigl(\sigma_{g,1}(\vecx_1)\sigma_{g,2}(\vecx_2) - \sigma_{g,1}(\vecx_2)\sigma_{g,2}(\vecx_1)  \bigr) + c_u\bigl(\sigma_{u,1}(\vecx_1)\sigma_{u,2}(\vecx_2) - \sigma_{u,1}(\vecx_2)\sigma_{u,2}(\vecx_1)  \bigr) \, ,
\end{equation}
where $c_g^2 + c_u^2 = 2$. The density is readily evaluated as
\begin{equation}\label{eq:interactDens}
n(\vecr_1) = \sum_{s_1,s_2} \integ{\vecr_2} \abs{\Psi(\vecx_1,\vecx_2)}^2
= c_g^2\sigma_g(\vecr_1)^2 + c_u^2\sigma_u(\vecr_1)^2 \, . 
\end{equation}
The KS approach to DFT now aims at reproducing the very same density in the same single-particle basis set but with a non-interacting auxiliary system. The corresponding single-particle KS Hamiltonian then becomes a two-by-two matrix in the single-particle states $\ket{\sigma_k}$
\begin{equation}
 \hat{h}_{\text{KS}} = \sum_{k,l} \Bigl( \ket{\sigma_k}
 \underbrace{\brakket{\sigma_k}{{-\thalf \nabla^2}}{\sigma_l}}_{{}= t_{kl}}\bra{\sigma_l} +
 \ket{\sigma_k}\underbrace{\brakket{\sigma_k}{v_{\text{KS}}}{\sigma_l}}_{{}=v_{kl}} \bra{\sigma_l}\Bigr) \, .
\end{equation}
Here the $\ket{\sigma_k}$ are connected to the spin-space orbitals 
\begin{equation}
\ket{\sigma_{k,m}} = \underbrace{\sum_{s} \integ{\vecr} \sigma_{k,m}(\vecr s) \hat{\psi}^{\dagger}(\vecr s)}_{{}=\crea{a}_{k,m}} \ket{0}
\end{equation}
by $\ket{\sigma_k} = \sum_{m}\ket{\sigma_{k,m}}$ and we employ for notational convenience and to connect the real-space perspective with a spin-orbital basis representation the field operators\footnote{We note that we later avoid the use of field operators which have some undesirable mathematical properties \citep{thirring2013} and use the non-problematic creation and annihilation operators directly (see Sec.~\ref{sec:Hamiltonians}).} obeying the fermionic anti-commutation relations $\{\anni{\psi}(\vecx), \crea{\psi}(\vecx')\} = \delta(\vecx-\vecx')$. The resulting creation and annihilation operators $\crea{a}_{k,m}$ and $\anni{a}_{k,m}$ for the spin-orbitals consequently also obey anti-commutation relations. Further, $\ket{0}$ is the vacuum state (see Sec.~\ref{sec:Hamiltonians}).

One of the problems is that in a finite basis, we cannot determine anymore whether the KS potential is local or non-local. To be more precise, with a non-local potential, we mean a potential which acts in the following manner on a function $\varphi(\vecr)$
\begin{equation}
\hat{v}\varphi(\vecr) = \integ{\vecr'}v(\vecr,\vecr')\varphi(\vecr')
= \sum_{kl} \psi_k(\vecr)v_{kl}\braket{\psi_l}{\varphi} \, ,
\end{equation}
where the summation runs over a complete basis $\{\psi_k\}$.
A local potential is a special (non-local) potential in the sense that it is diagonal in the spatial representation
\begin{equation}
\hat{v}^{\text{loc}}\varphi(\vecr) = \integ{\vecr'}v^{\text{loc}}(\vecr)\delta(\vecr-\vecr')\varphi(\vecr')
= v^{\text{loc}}(\vecr)\varphi(\vecr) \, .
\end{equation}
If we only have the matrix elements of the potential, let us say only $v_{kl}$ for $1 \leq k,l \leq m$, then we can easily construct a truly non-local potential as
\begin{equation}
\hat{v}^{\text{nl}} = \sum_{k,l=1}^m\ket{\psi_k}v_{kl}\bra{\psi_l} \, .
\end{equation}
One readily sees by acting on any other basis state that this potential is indeed not local as
$\hat{v}^{\text{nl}}\psi_k(\vecr) = 0$ for $k > m$.

With slightly more effort, we can also construct a local potential corresponding to these matrix elements. To that end, partition the space into $m(m+1)/2$ regions $\mathbb{A}_i$, i.e.\ the number of unique pairs in the finite basis.
Denote the overlap between the basis functions within these regions as $\braket{\psi_k}{\psi_l}_i$, where $i$ enumerates the regions. Further, set the potential to be constant within each of these regions with a value $v^{\text{loc}}_i$. Now we require this local potential to be consistent with the specified matrix elements $v_{kl}$, so the $v^{\text{loc}}_i$ need to satisfy
\begin{equation}
\sum_i\braket{\psi_k}{\psi_l}_i\,v^{\text{loc}}_i = v_{kl} \, .
\end{equation}
This is just a set of linear equations in which $\braket{\psi_k}{\psi_l}_i$ is regarded as a matrix with $kl$-pairs on its column and the region index $i$ as its row index. This set of linear equations will typically always have a solution. If not, just subdivide some of the regions. An explicit expression for the local potential can be given as
\begin{equation}
v^{\text{loc}}(\vecr) = \sum_iv^{\text{loc}}_i \mathbf{1}_{\mathbb{A}_i}(\vecr) ,
\end{equation}
where we used indicator functions $\mathbf{1}_{\mathbb{A}_i}(\vecr)$ defined as
\begin{equation}
\mathbf{1}_{\mathbb{A}_i}(\vecr) = \begin{cases*}
1	&if $x \in \mathbb{A}_i$ \\
0	&if $x \notin \mathbb{A}_i$ \, .
\end{cases*}
\end{equation}
As we cannot decide anymore in a finite basis set, whether the potential corresponding to a set of matrix elements $v_{kl}$ is local or non-local, a functional theory which does not need this distinction anymore, will be clearly in advantage over DFT. The functional theory employing exactly this set of non-local one-body potentials is 1RDM functional theory.

Putting these difficulties with the locality of the potential aside for the moment, let us see how far we can get within the KS DFT framework.
Assuming non-degeneracy, the unique ground state of this one-particle problem then reads $\ket{\varphi_0} = a \ket{\sigma_g} + b \ket{\sigma_u}$. The resulting two-body KS wave function becomes $\ket{\Psi_s} = ( a \crea{a}_{g,1} + b \crea{a}_{u,1})( a \crea{a}_{g,2} + b \crea{a}_{u,2}) \ket{0}$, which yields the density
\begin{equation}
 n_{s}(\vecr) = 2 \left(a^2 \sigma_g(\vecr)^2 + 2 a b \, \sigma_g(\vecr)\sigma_u(\vecr) + b^2 \sigma_u(\vecr)^2 \right) \, .
\end{equation}
As the interacting density~\eqref{eq:interactDens} is symmetric, we need either $a = 0$ or $b = 0$. So either $n_s(\vecr) = 2\sigma_g(\vecr)^2$ or $n_s(\vecr) = 2\sigma_u(\vecr)^2$. Therefore, we have $n_s(\vecr) \neq n(\vecr)$ if both $c_g \neq 0$ and $c_u \neq 0$, which is the typical case.

Our assumption in the KS construction was that we could find a non-degenerate state and it is actually this assumption that prevented us from reproducing the exact density. If we choose the KS potential such that the KS orbitals become degenerate, the both determinants $\ket{\Phi_g} = \crea{a}_{g,2}\crea{a}_{g,1}\ket{0}$ and $\ket{\Phi_u} = \crea{a}_{u,2}\crea{a}_{u,1}\ket{0}$ are degenerate and any linear combination of them is also a ground state. In particular, we can make the linear combination
\begin{equation}
\ket{\Psi_s} = c_g\ket{\Phi_g} + c_u\ket{\Phi_u} = \ket{\Psi} \, ,
\end{equation}
which would be the exact wave function and hence, yield the exact density. This would be the type of solution one expects from the Levy constrained-search approach to DFT \citep{Levy1979} as one limits oneself to pure states.

As proposed by Valone in both the 1RDM and density-functional setting \citep{Valone1980a, Valone1980b} and also independently by Lieb in the DFT setting \citep{Lieb1983}, extending the search to density-matrix operators leads to improved mathematical properties. The extension implies that the KS wave function does not necessarily need to be equal to the interaction one. For instance, if we assume degeneracy as before we could use the density-matrix operator (introduced in more detail in Sec.~\ref{sec:grandCanonIntro})
\begin{equation}
 \hat{\rho}_s = c_g^2 \ket{\Phi_g}\bra{\Phi_g} + c_u^2 \ket{\Phi_u}\bra{\Phi_u} \, .
\end{equation}
Which ever way we choose, both approaches imply that the KS system reproduces the 1RDM of the interacting system. Indeed, using the later convention (see Sec.~\ref{sec:FockSpace}) that we employ combined spin-orbital indices $i \equiv (k,m)$ the 1RDM operator reads $\hat{\gamma}_{ij} = \crea{a}_j \anni{a}_i$ and leads in our case to the 1RDM $\gamma_{ij} = \brakket{\Psi}{\hat{\gamma}_{i,j}}{\Psi} = \Trace\{\hat{\rho}_s \hat{\gamma}_{ij} \} = c_g^2\brakket{\Phi_g}{\hat{\gamma}_{i,j}}{\Phi_g} + c_u^2\brakket{\Phi_u}{\hat{\gamma}_{i,j}}{\Phi_u} = (\gamma_s)_{ij}$, where we used the definition of the trace in~\eqref{eq:Trace}. The explicit 1RDM is now given as
\begin{equation}
  \gamma = \begin{pmatrix}
   c_g^2 & 0 & 0 & 0 \\
   0 & c_g^2 & 0 & 0 \\
   0 & 0 & c_u^2 & 0 \\
   0 & 0 & 0 & c_u^2
  \end{pmatrix} \, .
\end{equation}
This effectively means that due to a lack of flexibility in the basis set, the KS system is actually forced to reproduce at least the exact 1RDM. In a finite basis set KS-DFT therefore typically degenerates to 1RDM functional theory if one insists on having exactly
\begin{equation}
\norm{n - n_s}_1
\isDefinedAs \integ{\vecr}\abs{n_s(\vecr) - n(\vecr)} = 0 \, .
\end{equation}
This finite basis size effect is not limited to two electron systems, but is a general problem of finite basis set DFT. For example, the same effect has also been observed in attempts to reproduce the correlated density of CH$_2$ \citep{SchipperGritsenkoBaerends1998}. In the smaller aug-cc-pCVTZ basis an ensemble was needed to reproduce the density with the desired accuracy, whereas in the larger cc-pCVQZ a pure state was sufficient.

As in a finite basis set we effectively will require that the 1RDMs are identical, it is more natural to attempt to define a 1RDM functional theory for finite basis sets. Since in 1RDM functional theory we use $\gamma$ as basic functional variable in contrast to DFT, which only uses the diagonal of the 1RDM in a spatial representation, also its conjugate variable will change. In order to be able to control the full 1RDM and to set up a suitable one-to-one correspondence, 1RDM functional theory allows for \emph{non-local potentials} $v_{ij}$ that give rise to a corresponding non-local potential operator
\begin{equation}
\hat{V}_v \isDefinedAs \sum_{ij}v_{ij}\crea{a}_i\anni{a}_j \, .
\end{equation}
That a purely local potential $v_{i}\delta_{ij}$ is not the appropriate conjugate variable to $\gamma_{ij}$ is evident from the different dimensionalities.
Thus we need to find conditions under which we can establish a one-to-one correspondence between $v$ and the resulting $\gamma$. The set of non-local potentials for which this is possible we denote by $\FiniteNpot$ and the set of induced 1RDMs, the so-called $v$-representable 1RDMs, we denote by $\VoneMat$. In the following we will discuss the theoretical set up for which we want to establish rigorous foundations of 1RDM functional theory.

\subsection{Non-uniqueness in 1RDM functional theory}

It has been observed already some decades ago that the same ground state 1RDM $\gamma$ can come from different non-local potentials $v$ which differ by more than a simple constant as in DFT \citep{Gilbert1975, Pernal2005, Leeuwen2007, PhD-Baldsiefen2012}. Though there has been some progress by giving a full account of the non-uniqueness in 1RDM functional theory \citep{Giesbertz2015} in the non-degenerate case, it would be convenient to circumvent this difficulty. The difficulty of a non-unique non-local potential is readily avoided by working at a finite temperature \citep{Leeuwen2007, PhD-Baldsiefen2012}. Working at finite temperature means that all states in the Hilbert space are participating in the ensemble, which avoids the possibility of `blind spots' as in the zero temperature case \citep{Giesbertz2016}. Additionally, problems with degenerate states are avoided, as the Boltzmann factors always select the equi-ensemble \citep{Valone1980b, Eschrig2010}.

In fact, we will even work with the grand canonical ensemble, which allows one to vary the particle number with the constant of the potential (chemical potential). This eliminates even all degrees of freedom in the potential and a strict one-to-one relation is obtained between the equilibrium 1RDM and the non-local potential, similar to its finite temperature DFT counterpart \citep{Mermin1965}. In particular, the constant of the potential acts as minus the chemical potential and controls the number of particles. The number of particles does not need to be integer anymore, as the particle number is now an average over states with different particle number. Hence, this will be the setting in which we wish to establish 1RDM functional theory for both fermions and bosons.

\subsection{Problems in the full-space case}
\label{sec:grandCanonIntro}

In statistical quantum mechanics one needs to allow for the possibility that the quantum state of a system is not completely determined. Instead one can only attribute a certain probability $w_i$ to encounter the system in the quantum state $\ket{\Psi_i}$. This uncertainty in the quantum state can conveniently be described with the help of the density-matrix operator
\begin{equation}\label{eq:DensityMatrixOp}
\hat{\rho} \isDefinedAs \sum_iw_i\ket{\Psi_i}\bra{\Psi_i} \, ,
\end{equation}
where $w_i \geq 0$ as they are probabilities and $\sum_iw_i = 1$, since the probability to encounter the system in any of the quantum states should be one.

To be able to determine the expectation value of a physical observable from the density-matrix operator, we will define the trace of an operator. The trace of an operator, $\Trace\{\cdot\}$, is defined as summing the expectation values of any complete basis of the Hilbert space under consideration. So for an operator $\hat{A}$ we have
\begin{equation}\label{eq:Trace}
\Trace\{\hat{A}\} \isDefinedAs \sum_i\brakket{\Psi_i}{\hat{A} }{ \Psi_i} \, ,
\end{equation}
where $\ket{\Psi_i}$ is a complete basis for the Hilbert space. Expectation values of observables are now evaluated by taking the trace of the density-matrix operator and the corresponding operator.
\begin{equation}
O = \av{\hat{O}} = \Trace\{\hat{\rho}\,\hat{O}\} = \sum_i\brakket{\Psi_i}{\hat{\rho}\,\hat{O} }{ \Psi_i}
=\sum_{i,k}\brakket{\Psi_i}{w_k\ket{\Psi_k}\bra{\Psi_k}\hat{O} }{ \Psi_i}
= \sum_iw_i\brakket{\Psi_i}{\hat{O} }{ \Psi_i} \, ,
\end{equation}
where we have chosen the eigenstates of the density-matrix operator as the orthonormal basis, since this allowed us to exploit the diagonal representation of the density-matrix operator. As expected we simply got the weighted average of the expectation value of the operator for each state. For later convenience we also introduce the notation $\trace\{ \cdot \}$ to indicate traces of the 1RDM $\gamma$ and objects with the same dimensionality, e.g., the non-local potential $v$. This distinction is useful because it will more clearly highlight where we make explicit use of the fact that we work with finite dimensions.

Up to this point we did not specify which Hilbert space to consider for the state $\ket{\Psi_i}$. There are two important cases to distinguish. The first option is to use a Hilbert space $\oneH^N$ with a fixed number of particles, $N$. This Hilbert space would be suitable for the canonical ensemble, since the number of particles is fixed in this ensemble. As the grand-canonical ensemble allows for an arbitrary amount of particles, this Hilbert space does not offer sufficient flexibility. We therefore need to resort to the other option to describe a grand-canonical ensemble: a Hilbert space with an arbitrary amount of particles. Such a Hilbert space can be be constructed for any quantum system by adding all `fixed number' Hilbert spaces leading to a new Hilbert space which is called the Fock space, $\Fock$. The procedure to construct the Fock space will be described in more detail later in Sec.~\ref{sec:FockSpace}.

The quantum-mechanical grand potential is now defined analogously to the classical case as
\begin{equation}\label{def:grandPotential}
\Omega_v[\hat{\rho}] \isDefinedAs E_v[\hat{\rho}] - \beta^{-1}S[\hat{\rho}] \, ,
\end{equation}
where
\begin{equation}\label{def:Energy}
E_v[\hat{\rho}] \isDefinedAs \Trace\{\hat{\rho}\,\hat{H}_v\}
\end{equation}
is the energy of the system with the Hamiltonian $\hat{H}_v \isDefinedAs \hat{H}_0 + \hat{V}_v$. Note that $\mu = -\trace\{v\} = -\sum_{i} v_{ii}$ already serves as the chemical potential with $\mu \hat{N} = \mu \sum_{i}\crea{a}_i \anni{a}_i$, so there is no need to add this term separately.
In the last term we have the inverse temperature, $\beta = 1/T$, and the entropy \citep{Neumann1927}
\begin{equation}\label{def:entropy}
S[\hat{\rho}] \isDefinedAs -\Trace\{\hat{\rho}\ln(\hat{\rho})\} \, .
\end{equation}
The thermodynamic equilibrium state of the system is defined as the density-matrix operator which minimises the grand potential. We will call the minimiser $\hat{\rho}_v$ the canonical density-matrix operator, which also sometimes referred to as the Gibbs state. To find the minimum, we simply follow the standard procedure and make it stationary with respect to variations in the density-matrix operator \citep{Mermin1965, Leeuwen2007}
\begin{equation}
0 = \Trace\bigl\{\delta\hat{\rho}\bigl(\hat{H}_v + \beta^{-1}\ln(\hat{\rho}_v)\bigr)\bigr\} + \beta^{-1}\Trace\{\delta\hat{\rho}\} \, .
\end{equation}
The unit trace condition requires that we only consider variations for which $\Trace\{\delta\hat{\rho}\} = 0$, so we obtain the solution
\begin{equation}\label{eq:interimSol}
\frac{1}{\beta}\ln(\hat{\rho}_v) + \hat{H}_v = C \, ,
\end{equation}
where $C$ is a constant to be determined by the unit trace condition. This equation is readily worked out as
\begin{equation}\label{eq:equiRho}
\hat{\rho}_v = \e^{-\beta\hat{H}_v} \big/ Z[v] \, ,
\qquad \text{where} \qquad
Z[v]\isDefinedAs \Trace\bigl\{\e^{-\beta\hat{H}_v}\bigr\} \, .
\end{equation}
It is clear that this procedure yields only a proper solution when $0 < Z[v] < \infty$. As might be unexpected, the case $Z[v] = \infty$ is actually the typical case for the quantum systems considered in chemistry and physics in full space, i.e., the particles are considered in $\mathbb{R}^3$. For example, consider the hydrogen atom and let us try to calculate the contribution from only the bound states in the one-particle sector. As the bound states have energies $\epsilon_n = -1/(2n^2)$ (in atomic units) for $n=1,2,\dotsc$ and an $n^2$-fold degeneracy the contribution to the partition function becomes
\begin{equation}
Z^{\text{bounded}}_{N=1} = \sum_{n=1}^{\infty}\sum_{l=0}^{n-1}\;\sum_{\mathclap{m=-l}}^l
\brakket[\big]{nlm}{\e^{-\beta\hat{H}}}{nlm}
= \sum_{n=1}^{\infty} n^2\e^{\beta/(2n^2)}
\geq \sum_{n=1}^{\infty}n^2 = \infty \, .
\end{equation}
Since the partition function already does not converge when we only include the bound states in the one-particle sector, the full partition function will definitely not converge. If this is already the case for the hydrogen atom, one quickly realises that this implies that the partition function of any molecule or solid is infinite. The problem is that all these systems have a Rydberg series and\slash{}or a continuum of states which makes the partition function divergent. More generally we can state that any Hamiltonian with an accumulation point or continuous part in its spectrum will yield a divergent partition function. The argument is along the same lines as before. As we have an accumulation point, there exists an infinite sequence of eigenstates $\{\Psi_k\}$, such that their energies $E_k \leq L < \infty$. The contribution to the partition function from these states is readily estimated as
\begin{equation}
Z^{\text{acc.}} = \sum_k^{\infty}\e^{-\beta E_k} \geq \sum_k^{\infty}\e^{-\beta L} = \infty \, .
\end{equation}
A similar argument can be used also for the continuum case, where we can find arbitrarily many approximate eigenfunctions (distributional eigenfunctions integrated over an arbitrarily small but finite spectral interval) within the continuous spectrum \citep{derezinski2013, thirring2013}.

There are two major approaches in practice to deal with this problem. The first one is to treat the volume as an extensive quantity explicitly and enclose everything in a box or in an infinitely large confining potential like an harmonic oscillator. By placing the molecule in a box or harmonic potential, we get rid of the Rydberg series and continuum states. The infinite-space limit is now obtained (provided it exists \citep{thirring2013}) by taking the limit of an infinitely large box at the end of the calculation, or by taking the limit of a very shallow harmonic potential.

The other option is to assume that the relevant physics only occurs in a small part of the Fock space and the remainder is relatively unimportant. So the second procedure is to simply truncate the Fock space to a finite-dimensional space. In this case quantum physics becomes simple linear algebra and there are no accumulation points or continua in the spectrum, since the Hamiltonian will just reduce to a finite-dimensional matrix. Hence, for a finite-dimensional Fock space the partition function is always finite. The approach would now be to calculate the desired properties for an increasing dimension of the Fock space and to see whether the answers converge.

Here we will follow a route in between. The restriction to a finite one-particle basis will, for the fermionic case, result due to the Pauli exclusion principle in a finite-dimensional Fock space (see Sec.~\ref{sec:FockSpace}). Hence for the fermionic case the mathematics will be comparatively simple. The fermionic Hamiltonians we consider (see Sec.~\ref{sec:Hamiltonians} and \ref{sec:potentials}) are matrices, a fermionic ground state will always exist and the 1RDM is defined for any fermionic density-matrix operator (see Sec.~\ref{sec:1RDM}). Further, the necessary properties of the grand potential and the universal functional will be easily determined (see Sec.~\ref{sec:fermions}). 

For the bosonic case though, the restriction to a finite one-particle basis will not lead to a finite-dimensional Fock space, since infinitely many bosons can occupy the same quantum state (e.g.\ in a Bose--Einstein condensate). Consequently we will have to deal with unbounded operators, the hallmark of quantum physics.\footnote{Note that the most fundamental relation of quantum mechanics, i.e., $[\hat{x}, \hat{p}] = \im \hbar \hat{1}$, necessarily needs unbounded operators \citep{blanchard2003}.} And in this case we can encounter again the case $Z=\infty$. For example, consider only a single bosonic mode and a non-interacting Hamiltonian, $\hat{H} = \epsilon\,\crea{a}\anni{a} = \epsilon\hat{N}$. In that case, the partition function is readily worked out as
\begin{equation}
Z = \sum_{n=0}^{\infty}\e^{-\beta\epsilon n}
= \begin{cases}
\bigl(1 - \e^{-\beta\epsilon}\bigr)^{-1}	&\text{if $\epsilon > 0$} \\
\infty							&\text{if $\epsilon \leq 0$} \, .
\end{cases}
\end{equation}
So only when $\epsilon > 0$, we obtain a finite value for the partition function and otherwise $Z$ diverges. This is actually not so strange, since if $\epsilon < 0$ the Hamiltonian is unbounded from below, because we can make the energy arbitrarily low by adding more and more particles.
Therefore, we should at least require that the Hamiltonian is bounded from below, i.e., the energy expectation value on the domain of the Hamiltonian has a lower bound. The domain, i.e., for which states the Hamiltonian is well-defined, is usually not the full infinite-dimensional Hilbert space. Take for instance the state $\ket{\Psi} = \sum_{n=1}^{\infty} \ket{n}/n$ in the case above. It is normalised to $\braket{\Psi}{\Psi} = \pi^2/6$, but if we act with the above Hamiltonian for an $\epsilon \neq 0$ on it then $\braket{\hat{H}\Psi}{\hat{H}\Psi} \rightarrow \infty$. Thus, such a state will not be in the domain. A proper account of the domain of the bosonic Hamiltonians, their self-adjointness and existence of ground states will be given in Sec.~\ref{sec:Hamiltonians} and \ref{sec:potentials}. Further, in Sec.~\ref{sec:densityoperators} and \ref{sec:1RDM} we then provide the details of the bosonic density-matrix operators and 1RDMs. The necessary properties of the bosonic grand potential and universal functional will then be derived in Sec.~\ref{sec:bosons}. Unfortunately, the existence of a ground state is not sufficient to ensure a finite partition function and hence a well-behaved grand potential. Consider for example the following Hamiltonian for the case of one bosonic mode, $\hat{H} = 1/(\hat{N} + 1)$. Formally this Hamiltonian can also be written as $\hat{H} = \sum_{n=0}^{\infty}(-\hat{N})^n$. Though this Hamiltonian has a ground state, it has an accumulation point as well. So again we have $Z = \infty$. To avoid such accumulation points, one would expect that if we introduced a highest-order repulsive interaction  between the bosons (basically stop the above expansion at some finite order $2n$) that we can avoid this `infinite boson' catastrophe. This is indeed the case as we will proof later in Sec.~\ref{sec:bosonpartition}.

\subsection{General approach for 1RDM functional theory}
\label{sec:general}

The general approach in density-functional like theories is to partition the minimisation in the canonical grand potential (energy in the zero temperature case) as
\begin{equation}\label{eq:OmegaVpart}
\Omega[v] = \inf_{\gamma}\bigl(F[\gamma] + \trace\{v\,\gamma\}\bigr) \, ,
\end{equation}
where
\begin{equation}\label{def:universalFunction}
F[\gamma] \isDefinedAs \inf_{\crampedclap{\hat{\rho} \to \gamma}}\Omega_0[\hat{\rho}]
= \inf_{\crampedclap{\hat{\rho} \to \gamma}}
\Trace\bigl\{\hat{\rho}\bigl(\hat{H}_0 + \beta^{-1}\ln(\hat{\rho})\bigr)\bigr\}
\end{equation}
is called the universal functional. In case no $\hat{\rho} \to \gamma$ exists, we define $F[\gamma \nleftarrow \hat{\rho}] = \infty$. This functional is universal in the sense that for a given interaction (fixed $\hat{H}_0$) it can be used for any system with an extra one-body potential $v$ ($\hat{H}_v = \hat{H}_0 + \hat{V}_v$). The use of the universal functional is obvious. If we had a manageable expression for $F[\gamma]$, we do not need to calculate the canonical density-matrix operator in the full Fock space to evaluate $\Omega[v]$ and hence find the exact $\gamma$. The main objective of this work is therefore to study the properties of this universal function $F[\gamma]$ (see Sec.~\ref{sec:grandpotential}). To do so we will take full advantage of the fact that we work in a finite one-particle basis set which makes $\gamma$ a finite-dimensional matrix and the universal functional will thus have a finite-dimensional domain. Since it can be shown that $F[\gamma]$ is strictly convex (Theorem~\ref{thm:Fconvex} in Sec.~\ref{sec:grandpotential}) we can take advantage of well-known properties of such functions (for a finite-dimensional domain):
\begin{enumerate*}[label=\arabic*)]
\item local Lipschitz continuity,
\item the directional derivative exists in all directions,
\item the subdifferential is non-empty and
\item if the subdifferential contains only one element, the function is differentiable and the subgradient equals the gradient.
\end{enumerate*}
All of these concepts will be defined more precisely and explained in more detail in Sec.~\ref{sec:convex}. The most important consequence is that the universal functional will be differentiable, if we are able to show uniqueness of the subdifferential. Differentiability of $F[\gamma]$ allows to find the minimiser~\eqref{eq:OmegaVpart} via
\begin{equation}
\label{eq:minimiser}
\frac{\du F}{\du\gamma} = -v \, .
\end{equation}
Strict convexity implies that there is only one solution to~\eqref{eq:minimiser} and that it yields a global minimum. To demonstrate uniqueness of the subdifferential, it will first be shown that any $\gamma$ (as defined in Sec.~\ref{sec:grandpotential}) is $v$-representable. The subdifferential can now be identified with the potentials generating the particular $\gamma$. By repeating Mermin's proof, we will show that the potential generating a 1RDM is actually unique, which implies that the subdifferential contains one element, and hence, that $F[\gamma]$ is differentiable.

Differentiability is also important if one desires to setup a KS like construction to approximate $F[\gamma]$ with the one from a non-interacting system. For a non-interacting system ($\hat{H}_s =\sum_{ij} (h_s)_{ij}\crea{a}_i \anni{a}_j$) the grand potential as a functional of the 1RDM can be worked out as (see Appendix~\ref{ap:nonInteracting}),
\begin{equation}
\Omega^{\pm}_s[\gamma] = \mp\frac{1}{\beta}\trace\bigl\{\ln(1 \pm \gamma)\bigr\} \, ,
\end{equation}
where the upper and lower sign refer to bosons and fermions respectively.
Since we have the grand potential as an explicit functional of the 1RDM, we can also construct an explicit expression for the non-interacting universal functional
\begin{equation}
F_s^{\pm}[\gamma]
= E^{\pm}_{s,0} - \beta^{-1}S_s^{\pm}
= \trace\bigl\{\gamma \, h_{s,0}\bigr\} + \beta^{-1}\trace\bigl\{\gamma\ln(\gamma) -
(\gamma \pm 1)\ln(1 \pm \gamma)\bigr\} \, .
\end{equation}
Note how much simpler the 1RDM functional is for the non-interacting system compared to the density-functional version: it is even explicit! In the density-functional version one would first need to find a local potential such that the non-interacting system yields the required density. Only from its solution, one could then finally determine $F_s[n]$. 

Since the complete one-body part of the energy in 1RDM functional theory is already explicit, i.e., $E_{\text{s},0}[\gamma]= \trace\{\gamma \, h_{s,0}\bigr\} $, we only need to find approximations to the many-body parts. Following the notation of electronic-structure theory we call the missing many-body energy terms the grand-canonical Hartree-exchange-correlation energy functional $E_{\text{Hxc},0}[\gamma]$ which is defined by
\begin{equation}
	E_0[\gamma] = \bigl(E_0[\gamma] - E_{\text{s},0}[\gamma]\bigr) + E_{\text{s},0}[\gamma]
	\sAdenifeDsI E_{\text{Hxc},0}[\gamma] + E_{\text{s},0}[\gamma] \, .
\end{equation} 
We note that with respect to density-functional approximations, here the energy expression does \emph{not} contain kinetic-energy contributions, since they are contained exactly in the one-body part~\footnote{This holds provided we have the standard case that the internal single-particle part of the KS Hamiltonian is the same as the interacting one, i.e., $h_{\text{s},0} = h_{0}$. If for some reason the interacting free single-particle term would obey $\hat{h}_0 \leq 0$, then in the bosonic case it would become necessary to use a strictly positive $h_{\text{s},0}$ (see Sec.~\ref{sec:potentials}).}. But in the grand-canonical setting it is not only the interaction terms in the energy expression that need to be approximated but also the entropy is a many-body quantity. Here the KS system allows a first useful approximation
\begin{equation}
S[\gamma] = \bigl(S[\gamma] - S_{\text{s}}[\gamma]\bigr) + S_\text{s}[\gamma]
\sAdenifeDsI S_{\text{c}}[\gamma] + S_{\text{s}}[\gamma] \, ,
\end{equation}
where we call $S_{\text{c}}[\gamma]$ the correlation entropy functional. Together this leads to
\begin{equation}
F[\gamma] = \bigl(F[\gamma] - F_{\text{s}}[\gamma]\bigr)+ F_{\text{s}}[\gamma]   \sAdenifeDsI F_{\text{Hxc}}[\gamma] + F_{\text{s}}[\gamma] \, ,
\end{equation}
where we borrow the name for the remainder $F_{\text{Hxc}}[\gamma] = E_{\text{Hxc},0}[\gamma] + S_{\text{c}}[\gamma]$ from electronic-structure theory: the grand-canonical Hartree-exchange-correlation functional. It is this term that needs to be approximated in practice.

Depending on the particular form of the many-body terms in the Hamiltonian, the energetic part can be further decomposed. As a concrete example, consider the typical setting that only number conserving two-body interactions are additionally present, e.g., electrons that interact via the Coulomb interaction. The interacting Hamiltonian in this case reads
\begin{equation}
 \hat{H}_0 = \sum_{ij} h_{ij}\crea{a}_i \anni{a}_j + \frac{1}{2}\sum_{ijkl} w_{ijkl}\crea{a}_{i}\crea{a}_{j}\anni{a}_{k}\anni{a}_{l} \, .
\end{equation}
It is natural to choose the kinetic energy as the one-body part of the reference Hamiltonian
\begin{subequations}
\begin{equation}
h_{ij} = t_{ij} = -\half\sum_s\integ{\vecr}\psi_i^*(\vecr s)\nabla^2\psi_j(\vecr s) ,
\end{equation}
where $\psi_i(\vecr s)$ is the space-spin representation of the one-body basis. The interaction matrix elements for the Coulomb interaction are given by
\begin{equation}
w_{ijkl} = \sum_{s_1,s_2}\iinteg{\vecr_1}{\vecr_2}
\frac{\psi_i^*(\vecr_1 s_1)\psi_j^*(\vecr_2 s_2)\psi_k(\vecr_2 s_2)\psi_l(\vecr_1 s_1)}
{\abs{\vecr_1 - \vecr_2}} .
\end{equation}
\end{subequations}
The Hartree and the exchange part are then explicitly given in terms of the 1RDM as
\begin{subequations}
\begin{align}
E_{\text{H}}[\gamma] &\isDefinedAs \half\sum_{ijkl}w_{ijkl}\gamma_{li}\gamma_{kj} \, ,\\
E_{\text{x}}[\gamma] &\isDefinedAs \half\sum_{ijkl}w_{ijkl}\gamma_{lj}\gamma_{ki} \, .
\end{align}
The correlation part of the interaction energy is now simply defined as the remaining part
\begin{equation}
E_{\text{c}}[\gamma] \isDefinedAs E_0[\gamma] - E_{\text{H}}[\gamma] - E_{\text{x}}[\gamma] \, .
\end{equation}
\end{subequations}
For this specific form of the many-body term we are therefore left to find approximations to $F_{\text{c}}[\gamma] = E_{\text{c}}[\gamma] + S_{\text{c}}[\gamma]$.
The simplest approximation would be to simply neglect all the correlation parts, i.e.\ to set $F_c = 0$, in which case the formalism becomes identical to grand-canonical Hartree--Fock.

At this point we like to stress that no exact expressions for the universal functional $F[\gamma]$ for any interaction are (yet) known, which could serve as a guiding principle for the construction of approximate functional. This is in contrast to 1RDM functional theory at the zero-temperature, e.g.\ 2-electrons~\cite{LowdinShull1956, CioslowskiPernal2001, PhD-Giesbertz2010} and single impurity Anderson model~\cite{TowsPastor2011, TowsPastor2012}, which have inspired the majority of approximations at zero temperature. For an overview with detailed descriptions, consult the recent review~\cite{PernalGiesbertz2015}. The only approximate temperature-dependent 1RDM functionals so far have been developed by T.~Baldsiefen et al.~\cite{BaldsiefenEichGross2012, BaldsiefenGross2012b, PhD-Baldsiefen2012}. Part of the work is complicated by the fact that data for the homogeneous gas for general non-local potentials is currently not available.

\subsection{Outline}

To summarise, we will first introduce in detail the fermionic and bosonic Fock spaces in Sec.~\ref{sec:FockSpace}, then discuss the creation and annihilation operators as well as Hamiltonian operators in Sec.~\ref{sec:Hamiltonians}. In Sec.~\ref{sec:potentials} we provide the spaces of the non-local potentials and discuss the properties of the density-matrix operators and the 1RDMS in Sec.~\ref{sec:densityoperators} and \ref{sec:1RDM}, respectively. We recapitulate properties of finite-dimensional convex (concave) functions in Sec.~\ref{sec:convex}. If we then assume that the partition function is finite and the universal functional is strictly convex then in Sec.~\ref{sec:grandpotential} we show $v$-representability of the 1RDMs. These assumptions will be proved to hold for the simple fermionic case in Sec.~\ref{sec:fermions} and then for the more advanced infinite-dimensional bosonic case in Sec.~\ref{sec:bosons}. This will be done in three successive steps, where in Sec.~\ref{sec:minimumgrandpotential} we first show that the bosonic grand potential has a minimum, in Sec.~\ref{sec:bosonpartition} we show under which conditions the bosonic partition function is finite and then in Sec.~\ref{sec:bosonicfunctional} we show that the bosonic universal functional has a minimum and is strictly convex. In Sec.~\ref{sec:infiniteDiscus} we discuss issues that arise if we build our theory on an infinite-dimensional one-particle space. Finally, in Sec.~\ref{sec:conclusion} we give a concise recapitulation of the complete setting in which 1RDM functional theory can be made mathematically rigorous before we discuss implications and future perspectives. In the appendices we then provide further details of expressions and theorems employed for completeness.

We like to stress again that the mathematics to handle the bosonic case can be quite formidable and intimidating. We would therefore advice the less mathematically inclined reader to focus on the fermionic case on a first-time reading and to only glance over the details of dealing with infinite-dimensional spaces.

\section{Setting the stage}
\label{sec:stage}

\subsection{Many-particle spaces}
\label{sec:FockSpace}

We will consider a quantum many-body system, where the particles can only occupy a finite number of single particle states, $\ket{i}$, for $i \in \Nbas \isDefinedAs \{1, \dotsc, N_b\}$ and $N_b < \infty$. The single particle states are assumed to be orthonormal, so $\braket{i}{j} = \delta_{ij}$.
From these orthonormal single-particle states, we can construct a one-particle Hilbert space $\oneH \isDefinedAs \bigl\{\ket{1},\dotsc,\ket{N_b}\bigr\}$, whose elements are linear combinations of the basis states, i.e.\ the single particle states $\{\ket{i}\}$
\begin{equation}
\ket{\psi} = \sum_{i=1}^{N_b}\ket{i}\psi_i \, .
\end{equation}
Let us stress that we work with a spin-dependent basis, so the index $i$ also runs over the different spin states.

The orthonormality of the basis states induces the following inner product on the one-particle Hilbert space
\begin{equation}
\braket{\phi}{\psi} \isDefinedAs \sum_{i,j=1}^{N_b}\phi_i^*\braket{i}{j}\psi_j
= \sum_{i=1}^{N_b}\phi_i^*\psi_i\, ,
\end{equation}
so the Hilbert space is isomorphic (surjective isometry) to the $N_b$-dimensional sequence or Euclidean spaces $\oneH \cong l^2(N_b) \cong \mathbb{C}^{N_b}$. The inner product yields the usual square norm $\norm{\psi} = \sqrt{\braket{\psi}{\psi}}$.

To accommodate $N$ particles, we will use tensor products of the one particle Hilbert space
\begin{equation}
\oneH^N \isDefinedAs \bigotimes_{i=1}^N\oneH
= \bigl(\underbrace{\oneH \otimes \dotsb \otimes \oneH}_{\text{$N$ times}}\bigr) \, .
\end{equation}
The basis states are readily constructed from the one particle basis as tensor products
\begin{equation}
\ket{i_1}\ket{i_2}\dotsb\ket{i_N}
\isDefinedAs \ket{i_1}\otimes\ket{i_2}\otimes\dotsb\otimes\ket{i_N} \, ,
\end{equation}
where $N$ is the order of the tensor products, i.e.\ the number of particles we will be dealing with. There are different ways how we can define the inner product, but the Dirac bra-ket notation appeals to the following definition
\begin{align}\label{def:NpartInnerProd}
\bra{i_N}\dotsb\bra{i_2}\bra{i_1}\ket{j_1}\ket{j_2}\dotsb\ket{j_N}
\isDefinedAs{}& \bra{i_N}\dotsb\bra{i_2}\braket{i_1}{j_1}\ket{j_2}\dotsb\ket{j_N} \notag \\
={}& \braket{i_1}{j_1}\bra{i_N}\dotsb\braket{i_2}{j_2}\dotsb\ket{j_N}
= \prod_{k=1}^N\braket{i_k}{j_k} \, .
\end{align}
The norm of $\ket{\Psi_N} \in \oneH^N \cong l^2(N_b^N) \cong \mathbb{C}^{N_b^N}$ is defined in the same manner as in the one-particle Hilbert space via the inner product as $\norm{\Psi_N} \isDefinedAs \sqrt{\braket{\Psi_N}{\Psi_N}}$.

The Hilbert space $\oneH^N$ is suitable for the description of distinguishable quantum particles. The description of indistinguishable quantum particles, however, requires the states to be symmetric (bosons) or anti-symmetric (fermions). These states belong to one of the following subspaces of $\oneH^N$
\begin{equation}
\oneH^N_{\pm} \isDefinedAs S_{\pm}\bigotimes_{i=1}^N\oneH
= S_{\pm}\bigl(\underbrace{\oneH \otimes \dotsb \otimes \oneH}_{\text{$N$ times}}\bigr) \, .
\end{equation}
The operator $S_+$ is a symmetriser and $S_-$ an anti-symmetriser, depending if we are dealing with bosons or fermions respectively. Note that the anti-symmetry of the fermions implies that the $\oneH^{N_b}_-$ is the fermionic many-particle Hilbert space with the largest possible number of fermions. This can be equivalently stated as $\oneH^{N > N_b} = \emptyset$.

The basis states of these $N$-particle Hilbert spaces are not merely tensor products of the one-particle states, but should also exhibit \mbox{(anti-)}symmetry. There are several possibilities regarding normalisation and sign convention. We will use the following definition to define a basis for the subspaces $\oneH^N_{\pm}$ \citep{StefanucciLeeuwen2013}
\begin{equation}\label{eq:symBasisStates}
\ket{i_1\dotsc i_N}
\isDefinedAs \frac{1}{\sqrt{N!}}\sum_{\perm}(\pm)^{\perm}
\ket{\perm(i_1)} \dotsb \ket{\perm(i_N)} \, ,
\end{equation}
where $\perm$ is a permutation of $i_1,\dotsc,i_N$ and $(\pm)^{\perm} = 1$ for even permutations and $(\pm)^{\perm} = \pm$ for odd permutations. Note that these basis states are unique up to an arbitrary permutation of their indices which only might induce a change in their phase factor. From the definition it is clear that the basis states are orthogonal if their indices are distinct. To be more precise, from the inner product~\eqref{def:NpartInnerProd} it follows that the inner product of the basis states is
\begin{equation}\label{eq:symBasisStatesInner}
\braket{i_N \dotsc i_1}{j_1 \dotsc j_N} = \begin{vmatrix}
\delta_{i_1j_1}	&\ldots	&\delta_{i_1j_N} \\
\vdots		&\ddots	&\vdots \\
\delta_{i_Nj_1}	&\ldots	&\delta_{i_Nj_N}
\end{vmatrix}_{\mathrlap{\pm}} \, ,
\end{equation}
where $\abs{A}_+$ denotes the permanent (bosons) and $\abs{A}_-$ denotes the determinant (fermions).

We therefore find that the states~\eqref{eq:symBasisStates} constitute an orthonormal basis in the fermionic case if we take for example $i_1 < \dotsb < i_N$. (Allowing for arbitrary $i_1 , \dotsc , i_N$ would yield an overcomplete basis). For this particular choice, the fermionic unit operator can be expressed as
\begin{subequations}\label{eq:resNidentity}
\begin{equation}\label{eq:fermUnit}
\hat{1}_{N} = \sum_{i_1=1}^{N_b}\sum_{i_2 > i_1}^{N_b}\dotsi\sum_{\crampedclap{i_N > i_{N-1}}}^{N_b}
\;\ket{i_N\dotsc i_1}\bra{i_1\dotsc i_N}
= \frac{1}{N!}\sum_{i_1=1}^{N_b}\dotsb\sum_{i_N=1}^{N_b}
\ket{i_N\dotsc i_1}\bra{i_1\dotsc i_N} \, .
\end{equation}
To disentangle the summations, we used that fermionic states in which particles occupy the same one-particle state do not exist, so drop out of the summation. Further, the factor $1/N!$ compensates for summing over equivalent states.

In the bosonic case, we also have the possibility that particles may occupy the same one-particle state. For these states~\eqref{eq:symBasisStatesInner} reveals that those states are not normalised. The bosonic unit operator therefore becomes
\begin{equation}
\hat{1}_{N} = \sum_{i_1=1}^{N_b}\sum_{i_2 \geq i_1}^{N_b}\dotsi\sum_{\crampedclap{i_N \geq i_{N-1}}}^{N_b}
\frac{\ket{i_N\dotsc i_1}\bra{i_1\dotsc i_N}}{\braket{i_N \dotsc i_1}{i_1 \dotsc i_N}}
= \frac{1}{N!}\sum_{i_1=1}^{N_b}\dotsb\sum_{i_N=1}^{N_b}
\ket{i_N\dotsc i_1}\bra{i_1\dotsc i_N} \, .
\end{equation}
\end{subequations}
Note that very conveniently, the number of equivalent states times the norm of the bosonic basis states~\eqref{eq:symBasisStates} is exactly $N!$, so disentangling the summations on the right-hand side yields the same form as in the fermionic case~\eqref{eq:fermUnit}.

An $N$-body quantum state for indistinguishable particles can now be expanded in terms of the basis states as
\begin{equation}
\ket{\Psi_N} = \frac{1}{N!}\sum_{i_1=1}^{N_b}\dotsb\sum_{i_N=1}^{N_b}
\ket{i_1\dotsc i_N}\Psi_{i_N\dotsc i_1} \, ,
\end{equation}
where $\Psi_{i_N\dotsc i_1} = \braket{i_N \dotsc i_1}{\Psi_N}$.
A practical advantage of this construction is that even if we choose a sequence $\{\Psi_{i_N \dotsc i_1}\} \in l^2(N_b^N)$ that does not have the right (anti-)symmetry, the resulting $\ket{\Psi_N}$ does. From the resolution of the identity~\eqref{eq:resNidentity} it follows that the inner product for two state $\ket{\Phi_N},\ket{\Psi_N} \in \oneH^N_{\pm}$ can be evaluated as
\begin{equation}
\braket{\Phi_N}{\Psi_N}
= \frac{1}{N!}\sum_{i_1=1}^{N_b}\dotsb\sum_{i_N=1}^{N_b}
\Phi_{i_1\dotsc i_N}^{\dagger}\Psi^{\vphantom{\dagger}}_{i_N\dotsc i_1} \, ,
\end{equation}
where $\Phi^{\dagger}_{i_1\dotsc i_N} \isDefinedAs \Phi^*_{i_N\dotsc i_1}$. The norm is defined in the usual manner via the inner product as $\norm{\Psi_N} \isDefinedAs \sqrt{\braket{\Psi_N}{\Psi_N}}$.

The special case $\oneH^0 = \oneH^0_{\pm}$ is defined to have one state: the vacuum state $\ket{0}$. Since $\oneH^0$ contains only one state, it is isomorphic to the complex numbers, $\oneH^0 \cong \Complex$.

The bosonic\slash{}fermionic Fock space can now be constructed by adding all the $N$-particle Hilbert spaces
\begin{equation}\label{def:FockSpace}
\Fock_{\pm} \isDefinedAs \bigoplus_{n=0}^{\infty}\oneH^n_{\pm} \, .
\end{equation}
An important difference between bosons and fermions is that the Fock space allows for an arbitrary number of bosons, so $\Fock_+$ is infinite-dimensional. On the contrary, the highest number of fermions which can be accommodated in $N_b$ one-particle states is $N_b$, so the fermionic Fock space is finite-dimensional ($2^{N_b}$) and the sum in~\eqref{def:FockSpace} only needs to run up to $n = N_b$. This difference might not seem to be very important at this point, but we will see that the bosonic case requires much heavier mathematics than the fermionic case to properly set up 1RDM functional theory.

The inner product on the Fock space is defined by adding the inner product in each particle sector. So given states $\ket{\Phi},\ket{\Psi} \in \Fock_{\pm}$
\begin{equation}\label{eq:FockStates}
\begin{split}
\ket{\Phi} &= a_0\ket{0} \oplus a_1\ket{\Phi_1} \oplus \dotsb
\oplus a_n\ket{\Phi_n} \oplus \dotsb \, , \\
\ket{\Psi} &= b_0\ket{0} \oplus b_1\ket{\Psi_1} \oplus \dotsb
\oplus b_n\ket{\Psi_n} \oplus \dotsb \, ,
\end{split}
\end{equation}
where $\ket{\Phi_n},\ket{\Psi_n} \in \oneH^n$,
the inner product is naturally defined as
\begin{equation}
\braket{\Phi}{\Psi} \isDefinedAs \sum_{n=0}^{\infty}a^*_n\braket{\Phi_n}{\Psi_n}b_n \, .
\end{equation}
The norm induced by the inner product is the usual square norm, $\norm{\Psi} = \sqrt{\braket{\Psi}{\Psi}}$.

At this point we find the first real mathematical differences between fermions and bosons. Since the fermionic Fock space is finite-dimensional it is isomorphic to finite-dimensional sequence or Euclidean spaces $\Fock_{-} \cong l^{2}(2^{N_b}) \cong \mathbb{C}^{2^{N_b}}$. Hence every possible state can be represented by a sequence of complex numbers $\ket{\Psi} \cong \{a_{1},\dotsc,a_{2^{N_b}}\}$ and is guaranteed to have a finite square norm. Indeed, for the finite-dimensional case all norms are equivalent, as they agree on the ordering of the vectors by length. This can be seen most easily for the case of the so-called $l^{p}$-norms which are defined by
\begin{equation}
\norm{\Psi}_p \isDefinedAs \left(\sum_{i=1}^{d}\abs{a_i}^p\right)^{\mathrlap{\nfrac{1}{p}}} \, ,
\end{equation}
where $1<p<\infty$ and $d$ is the dimensionality of the sequence space. In the case of $p=\infty$ we choose $\norm{\Psi}_\infty \isDefinedAs \sup_i\abs{a_i}$. If $d<\infty$ we have for $q>p$ on the one hand $\norm{.}_{p} \geq \norm{.}_{q}$ and due to Hölder's inequality $\norm{.}_{p} \leq d^{1/p-1/q}\norm{.}_{q}$. Thus for any two norms $ \norm{\cdot}_q \leq \norm{\cdot}_p \leq d^{1/p-1/q} \norm{\cdot}_{q}$. Consequently, it does not really matter for the fermionic case which norm we use in our calculations. If the norm of $\Psi$ is finite in some norm, it will be finite in any norm.

On the contrary, for the bosonic case we have infinite dimensions and hence the square norm is no longer automatically finite. An obvious example is the state we generated in the case of a single bosonic mode in Sec.~\ref{sec:grandCanonIntro}, i.e., $\ket{\hat{H} \Psi} = \sum_{n=1}^{\infty} \epsilon \ket{n}$, where we identify $\ket{1_1 \dotsc 1_n} \equiv \ket{n}$ for the one-particle Hilbert space $\mathcal{H}=\{ \ket{1}\}$. By retaining only those $\ket{\Psi}$ for which $\norm{\Psi}_2 < \infty$, we can turn the bosonic Fock space into a proper Hilbert space $\Fock_{+} \cong l^2(\mathbb{N})$, i.e.\ it is complete (every Cauchy sequence converges) with respect to the norm induced by the inner product.
Further, for the bosonic case it matters which norm we choose. Obviously, while for the above state $\norm{\hat{H}\Psi}_2 \rightarrow \infty$ we clearly have $\norm{\hat{H}\Psi}_\infty = \epsilon$. Another example would be with $\braket{\Phi_n}{\Phi_n} = 1$ the state
\begin{equation}
\ket{\Phi} = \sum_{n=1}^{\infty}\frac{1}{n^k}\ket{\Phi_n} \, ,
\end{equation}
which only for $k > 1/2$ obeys $\norm{\Phi}_2 < \infty$. If we choose a different $l^p$-norm, then we select a different set of states since only for $k>1/p$ we have $\norm{\Phi}_p < \infty$. It also shows that for $q>p$ we still have the inequality $\norm{\cdot}_p \geq \norm{\cdot}_q$ but we do no longer have the second inequality to make the norms equivalent. Further, it implies that for $q>p$ we have $l^{p}(\Nats)\subset l^{q}(\Nats) \subseteq l^{\infty}(\Nats)$. So why don't we use a different norm that allows for more general states? The reason is that only for $p=2$ the sequence space is a Hilbert space, i.e., has an inner product. And this structure allows us to properly define self-adjoint operators in the infinite-dimensional case. That we have self-adjoint Hamiltonians will become especially important when we want to define the exponentiation of an operator, e.g., for the canonical density-matrix operators or Gibbs states.

\subsection{Hamiltonians}
\label{sec:Hamiltonians}

Now let us turn our attention to the Hamiltonians on the Fock spaces. We will use creation and annihilation operators to define the Hamiltonians and divide them into two categories: number conserving and number non-conserving Hamiltonians. Again, the fermionic case will be trivial, while the bosonic case needs some more details.

We define the (formally adjoint) creation and annihilation operators%
\footnote{Sometimes, the creation operator has a plus symbol instead of the dagger, so $\hat{a}^+$ rather than $\crea{a}$, to stress that $\hat{a}^+$ adds a particle.}
by \citep{StefanucciLeeuwen2013}
\begin{subequations}
\begin{align}
\crea{a}_i \ket{i_1 \dotsc i_N} &=  \ket{i_1 \dotsc i_N \, i} \, , \\
\anni{a}_i \ket{i_1 \dotsc i_N}
&= \sum_{k=1}^{N}(\pm)^{N+k}\delta_{i_ki} \ket{i_1 \dotsc i_{k-1} i_{k+1} \dotsc i_{N}} \, ,
\end{align}
\end{subequations}
where the upper\slash{}lower sign refers to the bosonic\slash{}fermionic case and they obey the commutation\slash{}anti-communication relations for bosons\slash{}fermions
\begin{equation}\label{eq:creaAnniCommut}
\bigl[\anni{a}_k,\crea{a}_l\bigr]_{\mp} = \delta_{kl}
\qquad \text{and} \qquad
\bigl[\anni{a}_k, \anni{a}_l\bigr]_{\mp} = 0
\end{equation}
on their common domain.
Since the fermionic Fock space is finite, the operators on this Hilbert space are bounded (and continuous) by construction and have the full Fock space as their domain
\begin{equation}
 \anni{a}_i,\crea{a}_{i} \colon \Fock_{-} \mapsto \Fock_{-}.
\end{equation}
This also makes the creation operator to be the adjoint of the annihilation operator, since no subtleties with respect to domains arise. Further, bounded operators form an algebra and hence multiplication of bounded operators is again a bounded operator. Therefore, in the fermionic case any combination of creation and annihilation operators will be a bounded operator and thus defined on the full Fock space.

In the bosonic case this is no longer true. As an examplification we will use the state $\ket{\Psi}= \sum_{n=1}^{\infty} \ket{n}/n$ already employed in Sec.~\ref{sec:grandCanonIntro}. We then have for $\norm{\anni{a}\Psi}_2^2 = \brakket{\Psi}{\crea{a} \anni{a}}{\Psi} = \sum_{n=1}^{\infty}1/n \rightarrow \infty$, and thus we cannot have the full infinite-dimensional bosonic Fock space as domain of $\anni{a}$. This also holds for combinations of creation and annihilation operators and thus for the bosonic Hamiltonians. Although in physics and chemistry often ignored, self-adjoint is not the same as hermitian. That an operator $\hat{H}$ is hermitian means that $\braket{\hat{H}\Phi}{\Psi} = \braket{\Phi}{\hat{H}\Psi}$ for $\Phi, \Psi \in \domain(\hat{H})$. An operator is self-adjoint if it has the additional property that $\domain(\hat{H}) = \domain(\hat{H}^{\dagger})$~\cite{blanchard2003}. The latter condition is always fulfilled for a bounded linear operator, so in particular for any operator in a finite dimensional case. Hence, in the \emph{finite} dimensional case no distinction is needed and both terms can be used interchangeably. On the contrary, on infinite dimensional spaces unbounded linear operators exist for which the condition $\domain(\hat{H}) = \domain(\hat{H}^{\dagger})$ cannot taken be for granted. Only self-adjointness guarantees that the spectrum is real and that it has (generalised) eigenfunctions. The existence of this eigendecomposition allows us to define the exponential needed in the evaluation of the partition function~\eqref{eq:equiRho}, so this property is important for the Hamiltonian. A well-known example of an operator that is hermitian but not self-adjoint is the momentum operator in a box \citep{RuggenthalerPenzLeeuwen2015}. A simple way to see this is that no eigenfunctions for the momentum operator with zero boundary conditions exist. However, self-adjointness for a separable Hilbert space is equivalent to the existence of a diagonal representation in terms of (possibly distributional) eigenfunctions with real eigenvalues. Therefore the momentum operator in a box cannot be self-adjoint.

To analyse the bosonic situation in more detail, we use the approach by Cook \citep{Cook1953, Emch1972}.
Consider the subspace of the Fock space, $\preFock \subseteq \Fock$, which only contains vectors of finite, though arbitrary length. So in $\preFock$ we only include Fock states of the following form, cf.~\eqref{eq:FockStates}
\begin{equation}
\ket{\Phi} = a_0\ket{0} \oplus a_1\ket{\Phi_1} \oplus a_2\ket{\Phi_2} \oplus \dotsb \oplus a_n\ket{\Phi_n} \, ,
\end{equation}
with $n < \infty$ and where $\ket{\Phi_n} \in \oneH^n$. For each of these $n$-particle components we have $\crea{a}_i \colon \oneH^n \to \oneH^{n+1}$ and $\anni{a}_i \colon \oneH^n \to \oneH^{n-1}$ respectively. So for Fock states of finite length this implies that acting with a creation or annihilation operator on them yields a new state also of finite length, so we have 
$\crea{a}_i \colon \preFock \to \preFock$
and
$\anni{a}_i \colon \preFock \to \preFock$. So $\crea{a}_i$ and $\anni{a}_i$ can be defined to have a common domain $\preFock$ on which they are each others adjoint and the commutation relations~\eqref{eq:creaAnniCommut} are well defined. As the ranges are also $\preFock$, this implies that for any string of a finite number of creation and annihilation operators we have
\begin{equation}
\crea{a}_{i_1}\crea{a}_{i_2}\dotsb\crea{a}_{i_n}\anni{a}_{j_1}\anni{a}_{j_2}\dotsb\anni{a}_{j_m} \colon \preFock \to \preFock \, .
\end{equation}
Though $\preFock$ is not complete (its completion is $\Fock$, which also includes states with $n=\infty$, though finite norm), it has the important property that it is dense in $\Fock$. Dense in $\Fock$ means that any $\ket{\Psi} \in \Fock$ can be arbitrarily closely approximated by states in $\preFock$. So we can always find a state $\ket{\Phi} \in \preFock$ such that $\norm{\Phi - \Psi} < \epsilon$ for any $\epsilon > 0$. The fact that we can define arbitrary strings of creation and annihilation operators on a dense subspace of $\Fock$ turns out to be very useful to guarantee that Hamiltonians defined as combinations of such strings are self-adjoint.
 
To guarantee that the considered Hamiltonians are self-adjoint, we will make use of the following. Provided the operator is bounded from below by some number $\lambda \in \Reals$, hermitian and has a dense domain, then there exists a self-adjoint extension of the operator called the Friedrichs extension \citep{blanchard2003}.
As an example consider the number operator defined as
\begin{equation}\label{eq:NumberOp}
\hat{N} \isDefinedAs \sum_{i=1}^{N_b} \crea{a}_i \anni{a}_i \, .
\end{equation}
Since the number operator is defined by a linear combination of creation-annihilation operator strings of finite length, $\hat{N} \colon \preFock \to \preFock$, it is defined on a dense domain. It is obviously hermitian on this domain. Further, since $\brakket{\Psi}{\hat{N}}{\Psi} = \sum_{i=1}^{N_b}\braket{\anni{a}_i\Psi}{\anni{a}_i\Psi} = \sum_{i=1}^{N_b} \norm{\anni{a}_i\Psi}^2_2 \geq \lambda \norm{\Psi}_2$ for $\lambda = 0$, $\hat{N}$ is bounded from below. On the other hand, it is also obvious that the operator is not bounded from above. As the number operator is bounded from below, is hermitian and has a dense domain, we know that it has a self-adjoint realisation which has a spectral representation with real spectrum. In fact, as the Fock space was constructed from the eigenstates of the number operator, we actually know its self-adjoint realisation in its spectral form
\begin{equation}\label{eq:NumberSpectral}
\hat{N} = \bigoplus_{n=0}^{\infty} n \, \hat{1}_n \, .
\end{equation}
By constructing the Hamiltonians as linear combinations of creation-annihilation operator strings of finite length, we immediately ensure that the Hamiltonians are defined on a dense domain $\preFock$. Including only hermitian combinations of strings immediately makes them hermitian. The only thing we still need to worry about in the bosonic case is whether the Hamiltonians are bounded from below.

Let us consider these Hamiltonians in some more detail, to understand which kind of physical situations they can describe. This will also allow us to give conditions which guarantee that the Hamiltonian is bounded from below, and hence, has a self-adjoint realisation. The first step will be to split the Hamiltonians in a number conserving part $\hat{H}^{\text{c}}$ and a non-conserving part $\hat{H}^{\text{nc}}$
\begin{equation}
\hat{H} = \hat{H}^{\text{c}} + \hat{H}^{\text{nc}}.
\end{equation}
The number conserving Hamiltonians have the general form
\begin{equation}\label{eq:ncHam}
\hat{H}^{\text{c}} = 
\sum_{n=0}^{\nmax}\sum_{\substack{i_1,\dotsc, i_n=1 \\ j_1,\dotsc, j_n=1}}^{N_b}\!\!\!\!\!
h^{(n)}_{i_1\dotsc i_n,j_n \dotsc j_1}\crea{a}_{i_1}\dotsb\crea{a}_{i_n}\anni{a}_{j_n}\dotsb\anni{a}_{j_1} \, ,
\end{equation}
where $0 < \nmax < \infty$ is the maximum order of the interactions and $h^{(n)} \in \Complex^{2nN_b}$ is hermitian in the sense that
\begin{equation}
\bigl(h^{(n)}_{j_1\dotsc j_n,i_n \dotsc i_1}\bigr)^{\mathrlap{*}}
= h^{(n)}_{i_1\dotsc i_n,j_n \dotsc i_1} \, .
\end{equation}
This ensures that the particle-conserving part is hermitian. In the bosonic case, we additionally require that the matrix elements of the maximum order of interaction obey
\begin{equation}
\sum_{\substack{i_1,\dotsc, i_n=1 \\ j_1,\dotsc, j_n=1}}^{N_b}\!\!\!\!\!
u^*_{i_1\dotsc i_n}h^{(\nmax)}_{i_1\dotsc i_n,j_n \dotsc j_1}u_{j_1\dotsc j_n} > 0
\end{equation}
for all $u \in \Complex^{nN_b}$. In other words, in the bosonic case we require that $h^{(\nmax)}$ is positive definite.
Due to Theorem~\ref{thm:eqDomHNn} we know that $\domain{(\hat{H}^{\text{c}})} = \domain{(\hat{N}^{\nmax})}$ is dense in $\Fock_{+}$, and according to corollary~\ref{thm:boundedfrombelow} we know that $\hat{H}^c$ will be bounded from below. That the highest-order interaction is supposed to be positive to assure boundedness from below in the bosonic case is physically quite intuitive. Since if the highest-order interaction would not be positive, the energy could be lowered indefinitely by adding more and more particles.

Now, what do the different orders in $\hat{H}^{\text{c}}$ correspond to? The term with $n=0$ corresponds to a constant in the Hamiltonian which only shifts the eigenvalue spectrum. This could be the repulsion between the nuclei in a molecule when we only describe the electrons quantum mechanically. The next order, $n=1$, contains the one-body part of the Hamiltonian. The one-body part comprises at least the kinetic energy (= hopping matrix elements) and can also contain effects due to a one-body potential, e.g.\ a dipole field or the electrostatic field generated by nuclei. As already mentioned before, in the grand canonical setting, the negative of the trace of the one-body potential acts as the chemical potential. This is easily understood as the constant in the potential sets the potential relative to infinity, which acts as the bath with which the particles can be exchanged.
The second term, $n = 2$, contains the two-body interactions, e.g.\ the Coulomb interaction between electrons or the Hubbard $U$ onsite interaction. The higher-order terms then correspond to more complicated many-body interactions.

Let us next turn to the number non-conserving parts of the Hamiltonian $\hat{H}^{\text{nc}}$. We want to allow for Hamiltonians which mix the states of different particle numbers. The major requirement for the number non-conserving terms is that they are hermitian. That is obviously enough for the fermionic case to guarantee that the total Hamiltonian is self-adjoint. For the bosonic case we again need to ensure that the full Hamiltonian is bounded from below. This roughly means that we need to ensure that the non-conserving parts in the Hamiltonian do not become too large compared to the conserving parts.

The lowest order non-conserving term is of the form of a source or $1/2$-body operator \citep{DominicisMartin1964a}
\begin{equation}\label{eq:sourceH}
\sum_{i}\Bigl(\bigl(h_i^{(\nhalf)}\bigr)^*\crea{a}_i + h_i^{(\nhalf)}\anni{a}_i\Bigr) \, .
\end{equation}
In the context of photons, for instance, this term corresponds to the coupling to an external current or dipole \citep{greiner2013, grynberg2010}.

The next higher order term is used in the Bardeen--Cooper--Schrieffer (BCS) Hamiltonian to model the formation of Cooper pairs to explain superconductivity, the (anomalous) pairing field
\begin{equation}\label{eq:pairingH}
\sum_{ij}\bigl(D^{\dagger}_{ij}\crea{a}_i\crea{a}_j + 
D^{\vphantom{\dagger}}_{ij}\anni{a}_i\anni{a}_j\bigr) \, .
\end{equation}
In Appendix~\ref{ap:non-interatingSolution}, we work out the solution of a non-interacting Hamiltonian of the most general form, i.e.\ with both a source term~\eqref{eq:sourceH} and a pairing field~\eqref{eq:pairingH}.
The source term only shifts the spectrum as a whole, so no additional restrictions on the source term are needed.

However, in the bosonic case, a too strong pairing field leads to an unbounded operator with a pure continuous spectrum \citep{Chruscinski2003}, so the pairing matrix $D$ cannot be chosen arbitrarily large for bosons. This is readily clarified by writing the pairing field in terms of position and momentum operators
\begin{equation}
\frac{\omega}{2}\bigl(\crea{a}\crea{a} + \anni{a}\anni{a}\bigr)
= \half m\omega^2x^2 - \frac{\hat{p}^2}{2m} \, .
\end{equation}
Adding this perturbation to the harmonic oscillator Hamiltonian with strength $d \in \Reals$, we find
\begin{equation}
\hat{H}
= \omega\,\crea{a}\anni{a} + d\,\frac{\omega}{2}\bigl(\crea{a}\crea{a} + \anni{a}\anni{a}\bigr)
= (1 - d)\frac{\hat{p}^2}{2m} + \frac{1+d}{2} m \omega^2 x^2
= \frac{\hat{p}^2}{2m_r} + \half m_r\omega_r^2x^2 \, 
\end{equation}
where
\begin{equation}
m_r = m/(1 - d)
\qquad \text{and} \qquad
\omega_r = \omega\sqrt{1 - d^2} \, .
\end{equation}
Hence, we get a renormalised version of the harmonic oscillator for $\abs{d} < 1$. For $d = -1$, we exactly eliminate the harmonic potential and the Hamiltonian of a free particle with half the original mass remains, $m_r = m/2$. For $d = 1$, the effective mass becomes infinite, $m_r = \infty$ and we are left with a Hamiltonian without any kinetic energy and only $m\omega^2x^2$ remains. The resulting operators for $\abs{d} = 1$ are still self-adjoint, and it is clear that we have a completely continuous spectrum. For $\abs{d} > 1$ the system corresponds to an inverted harmonic oscillator, which only serves as a scattering potential (see Fig.~\ref{fig:bosonFailure}). One therefore expects a purely continuous spectrum $(-\infty,\infty)$, which is indeed the case, as demonstrated in Ref.~\citep{Chruscinski2003}.

\begin{figure}[t]
\centering
\begin{tikzpicture}[domain=-2:2, scale=0.75]
  \node at (-3,2.5) {$\hat{H}_0 = \omega\,\crea{a}\anni{a}$};
  \node at (-3,-2.5) {$V(x) = \half m\omega x^2$};
  \draw[thick, smooth] plot(\x-3,\x*\x-2);
  \draw[thick, ->] (-0.5,0) -- (0.5,0);
  \draw[thick, smooth] plot(\x+3,2-\x*\x);
  \node at (3,2.5) {$\hat{H}_0 - \half m\omega(\crea{a}\crea{a} + \anni{a}\anni{a})$};
  \node at (3,-2.5) {$V(x) = -\half m\omega x^2$};
\end{tikzpicture}
\caption{Plot of the potential when adding a too strong pairing field to the (bosonic) harmonic oscillator Hamiltonian $\hat{H}_0$. The particles will be unbound and the spectrum of the perturbed Hamiltonian will be completely continuous \citep{Chruscinski2003}.}
\label{fig:bosonFailure}
\end{figure}
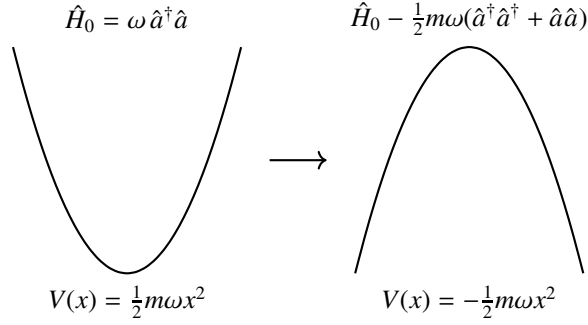

Higher order non-conserving terms can be devised in a similar manner as the lowest order terms to ensure that the Hamiltonian is hermitian.
An example would be the term
\begin{equation}
\sum_{\mathclap{ijklm}}\bigl(
d^{\vphantom{*}}_{ijklm}\crea{a}_i\crea{a}_j\crea{a}_k\anni{a}_l\anni{a}_m +
d^*_{ijklm}\crea{a}_m\crea{a}_l\anni{a}_k\anni{a}_j\anni{a}_i\bigr) \, .
\end{equation}
From the discussion on the bosonic pairing field it is clear that additional constraints on the strength of these general non-conserving parts are needed in the bosonic case to ensure that the Hamiltonian is bounded from below and has a discrete spectrum without accumulation points. Sufficient bounds are discussed in the specialised Section~\ref{sec:bosonpartition} and the relevant inequalities are presented in Table~\ref{tab:inequalities}.

\subsection{One-body potentials}
\label{sec:potentials}

Since we expect a one-to-one correspondence between the 1RDM and (non-local) one-body potentials, we will consider perturbations of the Hamiltonian by a one-body potential
\begin{equation}\label{eq:potHam}
\hat{H}_v \isDefinedAs \hat{H} + \hat{V}_v
\isDefinedAs \hat{H} +  \sum_{ij} v_{ij} \crea{a}_i\anni{a}_j \, ,
\end{equation}
where $v = v^{\dagger}$ to keep the full Hamiltonian $\hat{H}_v$ hermitian. To have a properly defined canonical density-matrix operator, we need that the partition function is finite. We will therefore only use potentials in the following set
\begin{equation}
\FiniteNpot \isDefinedAs \set*{v \in \HermitanMat(N_b)}{Z[v] < \infty} \, ,
\end{equation}
where $\HermitanMat(N_b)$ denotes the set of hermitian $N_b \times N_b$ matrices, i.e.
\begin{equation}\label{def:hermitianMatrix}
\HermitanMat(N_b)
\isDefinedAs \set*{h \in \Nbas \times \Nbas \to \Complex}{h = h^{\dagger}} \, .
\end{equation}
Later we will show in theorem~\ref{thm:finiteZ} that if the Hamiltonian has a highest-order interaction, i.e.\ $\nmax < \infty$, and is bounded from below, then $Z[v] < \infty$. Let us therefore discuss when we can expect the perturbed Hamiltonian $\hat{H}_v$ to be bounded from below. This will also guarantee that the resulting Hamiltonian is self-adjoint as discussed in Sec.~\ref{sec:Hamiltonians}.

First note, that since the fermionic Fock space is finite-dimensional, $\hat{H}_v$ will always be bounded from below (and above). Thus the fermionic space of non-local potentials is just the full space of hermitian matrices, $\FiniteNpot_{-} = \HermitanMat(N_b)$.

In the bosonic case, however, even if $\hat{H}$ is bounded from below, $\hat{H}_v$ might not be bounded from below for general $v \in \HermitanMat(N_b)$. Take, for instance, the non-interacting bosonic case where we have $\hat{H}_v = \sum_{ij}(h^{(1)}_{ij} + v_{ij}) \crea{a}_i\anni{a}_j$. By choosing $v$ such that $h^{(1)} + v$ has a negative eigenvalue, we can lower the energy by an arbitrary amount by putting more and more bosons in this negative energy state. This is an important difference to the usual zero-temperature and fixed-number-of-particles case, where we can shift the single-particle energy spectrum by an arbitrary amount and not influence the physics. In the non-interacting grand-canonical case we always have to choose this constant such that the single-particle ground-state energy is positive, and the choice of this constant influences the physical result. To put it differently: since the constant that shifts the spectrum influences directly the chemical potential $-\trace\{v\}$, we need to choose a constant that makes it energetically unfavourable to add an infinite amount of particles to a non-interacting system in the grand-canonical situation. Also a zero eigenvalue should be avoided, since this leads to an infinite number of many-particle states with the same energy and prevents the partition function from being finite (see Sec.~\ref{sec:grandCanonIntro}). So for non-interacting bosons, we readily find that
\begin{equation}\label{eq:FiniteNPot}
\FiniteNpot_{+}^{\text{nonint}} = \set*{v \in \HermitanMat(N_b)}{h^{(1)} + v >0} \, ,
\end{equation}
which makes $\hat{H}_v > 0$. To have a properly defined reference system, we need $v = 0$ to be contained in $\FiniteNpot_{+}^{\text{nonint}}$, so one would need $h^{(1)} > 0$. The most natural choice is to use the kinetic energy operator for $h^{(1)}$, since that is a part we usually cannot manipulate in the experiment and is strictly positive definite, i.e.
\begin{equation}
\sum_{ij} h^{(1)}_{ij}\crea{a}_i\anni{a}_j = \hat{T} = \sum_{ij} t_{ij}\crea{a}_i\anni{a}_j > 0 \, .
\end{equation}
It should be clear that other choices for $h^{(1)}$ are definitely possible under the aforementioned conditions.

In the interacting case we have due to the assumptions $\infty > \nmax > 1$ and $h^{(\nmax)}>0$ from Sec.~\ref{sec:Hamiltonians} that a perturbation in the first-order terms does not make $\hat{H}_{v}$ unbounded from below. Thus in the interacting bosonic case we have again $\FiniteNpot_{+} = \HermitanMat(N_b)$.

\subsection{Density-matrix operators}
\label{sec:densityoperators}

We have already introduced the density-matrix operators in Sec.~\ref{sec:grandCanonIntro}, but we will also need a norm (distance) between them. The set of density-matrix operators on a Fock space $\Fock_{\pm}$ space can be defined as
\begin{equation}\label{def:NdensMat}
\NdensMat_{\pm} \isDefinedAs \set*{ \hat{\rho} \colon \Fock_{\pm} \to \Fock_{\pm} }{ 
\hat{\rho} = \hat{\rho}^{\dagger}, \hat{\rho} \geq 0, \Trace\{\hat{\rho}\} = 1 } \, .
\end{equation}
The condition $\hat{\rho} = \hat{\rho}^{\dagger}$ means that the operator is self-adjoint, the condition $\hat{\rho} \geq 0$ means that the density-matrix operators are positive semidefinite ($w_i \geq 0$) and the last condition, $\Trace\{\hat{\rho}\} = 1$, means that the weights should sum to one. Further, we will assume that the weights are arranged in decreasing order.

There are now many possibilities to define a norm on $\NdensMat_{\pm}$. We will use the observation that the set of density-matrix operators can be regarded as a subspace of a larger space, the space of trace-class operators
\begin{equation}\label{def:TraceClass}
\TraceClass_{\pm} \isDefinedAs \set[\big]{\hat{A} \colon \Fock_{\pm} \to \Fock_{\pm} }{
\norm{\hat{A}}_1 < \infty } \, ,
\end{equation}
where the trace norm $\norm{\cdot}_1$ is a special case of the following norms for $p \geq 1$
\begin{subequations}
\begin{equation}
\norm{\hat{A}}_p = \Bigl(\Trace\{\abs{\hat{A}}\}\Bigr)^{\mathrlap{\nfrac{1}{p}}}
\isDefinedAs \biggl(\sum_i\brakket{\Psi_i}{(\hat{A}^{\dagger}\hat{A})^{\nfrac{p}{2}}}{\Psi_i}\biggr)^{\mathrlap{\nfrac{1}{p}}} \, ,
\end{equation}
and $p = \infty$ we define to be the operator norm (maximum possible amplification)
\begin{equation}
\norm{\hat{A}}_\infty \isDefinedAs \sup\set[\big]{\norm{\hat{A}\Psi}}{\text{$\Psi \in \Fock_{\pm}$ with $\norm{\Psi} \leq 1$}}\, .
\end{equation}
\end{subequations}
These norms are known as Schatten norms and are a generalisation of the $l^p$ norms (Sec.~\ref{sec:FockSpace}) to operators. 
The Schatten norms obey the same sequence of inequalities as the $l^p$ norms
$\norm{\hat{A}}_p \geq \norm{\hat{A}}_q$ for any $1 \leq p \leq q \leq \infty$. Thus the unit trace condition immediately implies that the density-matrix operators are bounded, as $1 = \norm{\hat{\rho}}_1 \geq \norm{\hat{\rho}}_{\infty}$. Thus the density-matrix operator is defined on the whole Fock space and self-adjoint in both the fermionic and the bosonic case.

The inequality $\norm{\hat{A}}_1 \geq \norm{\hat{A}}_p$ also implies that $\TraceClass_{\pm}$ is a subspace of all the other spaces induced by the according $p$-norms, so we could use any of those norms. But the trace norm is the natural choice for $\NdensMat_{\pm}$, as it corresponds to the unit trace condition on the density-matrix operators. Since $\norm{\cdot}_1$ is a proper norm on $\TraceClass_{\pm}$, we can also use it on the subspace $\NdensMat_{\pm} \subset \TraceClass_{\pm}$. The set of density-matrix operators $\NdensMat_{\pm}$ can then be classified as positive trace-class operators with $\norm{\hat{A}}_1 = 1$. In other words, the density-matrix operators live on the positive orthant of the surface of a ball with radius 1 in $\TraceClass_{\pm}$.

The norms with $p > 1$ can be used to separate the pure states from the mixed states. A pure state is a density-matrix operator for which only one state is needed, i.e., which can be expressed as
\begin{equation}\label{eq:pureState}
\hat{\rho}_{\text{pure}} = \ket{\Psi}\bra{\Psi} \, .
\end{equation}
Mixed states are simply all other density-matrix operators. Only the pure states possess the property
\begin{subequations}
\begin{equation}
\norm{\hat{\rho}_{\text{pure}}}_p = 1 \quad \text{for any $1\leq p \leq \infty$.} 
\end{equation}
On the other hand, for mixed states we always have
\begin{equation}
\norm{\hat{\rho}_{\text{mixed}}}_p < 1 \quad \text{for any $1 < p \leq \infty$.} 
\end{equation}
\end{subequations}
The difference between the pure and mixed states also becomes apparent in their value of the entropy \citep{Wehrl1978}
\begin{equation}
S[\hat{\rho}_{\text{pure}}] = 0
\qquad \text{and} \qquad
S[\hat{\rho}_{\text{mixed}}] > 0 \, .
\end{equation}
As already discussed at the end of Sec.~\ref{sec:FockSpace}, it does not really matter which norm we use in the case of a finite-dimensional space. Thus in the fermionic case all the norms are equivalent. In the bosonic case the choice does matter. We further note, that not for all $\hat{\rho} \in \NdensMat_{+}$ the operators $\hat{H}_v \hat{\rho}$ and $\hat{\rho}\ln(\hat{\rho})$ are again trace-class. Thus, the domain of $\Omega_v$, i.e., $\Omega_v[\rho] < \infty$, will be a subset of $\NdensMat_{+}$. Indeed, an explicit characterisation is not so simple and we will see in Theorem~\ref{thm:domainSbosonic} that for any $\hat{\rho}$ in the domain there is another density-matrix operator arbitrarily close that is not in the domain. However, we will still be able to show that $\Omega_v[\hat{\rho}]$ is strictly convex. This property, briefly recapitulated in Sec.~\ref{sec:convex}, allows us to further characterize $\Omega_v[\hat{\rho}]$ in Sec.~\ref{sec:grandpotential}.

\subsection{The 1RDM}
\label{sec:1RDM}

The 1RDM operator is defined as an $N_b \times N_b$ matrix of operators $\hat{\gamma}_{ij} \isDefinedAs \crea{a}_j\anni{a}_i \colon \domain{(\hat{\gamma}_{ij})} \to \Fock$. The 1RDM operator is hermitian in the sense that $\hat{\gamma}^{\dagger} = \hat{\gamma}$, or worked out in components
\begin{equation}
\hat{\gamma}_{ij}^{\dagger} = \hat{\gamma}_{ji}^* = \bigl(\crea{a}_i\anni{a}_j\bigr)^*
= \bigl(\anni{a}_j\bigr)^*\bigl(\crea{a}_i\bigr)^* = \crea{a}_j\anni{a}_i = \hat{\gamma}_{ij} \, ,
\end{equation}
so includes the matrix transposition. The number operator is obtained by taking the trace over this matrix of operators
\begin{equation}\label{eq:numbOpFrom1RDM}
\hat{N} = \trace\{\hat{\gamma}\} \isDefinedAs \sum_{i=1}^{N_b}\hat{\gamma}_{ii} \, .
\end{equation}
A number of properties of 1RDM operators are easy to derive \citep{Lowdin1955a}. For expectation values we have for the diagonal entries
\begin{equation}
0 \leq \norm{\anni{a}_i \Psi}^2_2
= \brakket{\Psi}{\crea{a}_i\anni{a}_i}{\Psi}
= \brakket{\Psi}{\hat{\gamma}_{ii}}{\Psi}
= \brakket{\Psi}{1 \pm \anni{a}_i\crea{a}_i}{\Psi} = 1 \pm \norm{\crea{a}_i \Psi}^2_2 \, .
\end{equation}
Therefore, we find that the diagonal elements are positive and that for fermions they have a maximum value 1 (Pauli principle: no state can be occupied by more than one particle). For bosons there is no upper bound.

Now let us derive a condition on the off-diagonal elements of the 1RDM operator, which is basically the proof of the Cauchy--Schwarz inequality. Since for every state $\ket{\Psi} \in \domain{(\hat{\gamma}_{ij})}$ we have for any $\lambda \in \Complex$ that
\begin{equation}
0 \leq \norm{\bigl(\anni{a}_i - \lambda \anni{a}_j\bigr) \Psi}^2_2 
= \brakket[\big]{\Psi}{\bigl(\crea{a}_i - \lambda^* \crea{a}_j\bigr)\bigl(\anni{a}_i - \lambda \anni{a}_j\bigr)}{\Psi}
= \gamma_{ii} + \abs{\lambda}^2\gamma_{jj} - \lambda\gamma_{ji} - \lambda^*\gamma_{ij} ,
\end{equation}
where we used $\gamma_{ij} = \brakket{\Psi}{\hat{\gamma}_{ij}}{\Psi}$ as an abbreviation. Now setting $\lambda = \gamma_{ij} / \gamma_{jj}$, we find%
\footnote{The choice $\lambda = \gamma_{ij}/\gamma_{jj}$ is problematic if $\gamma_{jj} = 0$. This case is easily circumvented by interchanging the roles of $i$ and $j$ if $\gamma_{ii} \neq 0$. In the case also $\gamma_{ii} = 0$, we immediately find that $\gamma_{ij} = 0$. Further, it is clear that $\abs{\gamma}_{ij} = \infty$ is only possible for any $\lambda$ if also $\gamma_{ii} = \gamma_{jj} = \infty$.}%
\begin{equation}\label{eq:gammaOffDiagBound}
\gamma_{ii}\gamma_{jj} \geq \abs{\gamma_{ij}}^2 \, .
\end{equation}
Consequently the off-diagonal elements are bounded by the diagonal elements. This is a necessary condition for a matrix to be positive semidefinite. Indeed, since the expectation value of the 1RDM operator is obviously a hermitian matrix, it can be diagonalised by a unitary transformation of the one-particle basis. Since the conditions derived for the diagonal entries are valid in any one-particle basis, this implies that the 1RDM operator is always a positive semidefinite matrix for any $\ket{\Psi} \in \domain{(\hat{\gamma}_{ij})}$. Hence, the 1RDM operator is a positive semidefinite operator, i.e., $\hat{\gamma} \geq 0$, and thus has a self-adjoint realisation. Due to the Pauli exclusion principle for fermions, we have additionally the natural upper bound on the eigenvalues, $n_i \leq 1$, which can alternatively be expressed as $n_i^2 \leq n_i$. Since this applies to any $\ket{\Psi} \in \Fock_{-}$, the last inequality can be translated back to the 1RDM operator as $\hat{\gamma}^2 \leq \hat{\gamma}$ in the fermionic case.

The next step is to identify the states for which the 1RDM operator is well defined, i.e., for which the expectation values of all its elements are finite. This identification is important, since the 1RDM will be the central quantity in the theory.
In the simplest setting that we only work with one single-particle state, $N_b = 1$, the 1RDM is identical to the number operator $\hat{\gamma}_{11} = \crea{a}_1\anni{a}_1 = \hat{N}$, so also has the same domain.
For example, the state $\ket{\Psi}= \sum_{n=1}^{\infty} \ket{n}/n$ introduced earlier at the beginning of Sec.~\ref{sec:Hamiltonians} is an example which yields $\gamma_{11} = \brakket{\Psi}{\hat{\gamma}_{11}}{\Psi} = \infty$, so is outside its domain.
Since the number operator is obtained by summing over all entries~\eqref{eq:numbOpFrom1RDM}, one would expect that might be the same in general. This is indeed the case, since the off-diagonal entries are bounded by the diagonal part of the 1RDM~\eqref{eq:gammaOffDiagBound}. We have put this convenient result as the following proposition.

\begin{prop}
$\domain(\hat{\gamma}) = \domain(\hat{N})$.
\end{prop}

\begin{proof}
First we show $\domain(\hat{\gamma}) \subseteq \domain(\hat{N})$. For any $\ket{\Psi} \in \domain(\hat{\gamma})$, we have $\infty > \trace\{\brakket{\Psi}{\hat{\gamma}}{\Psi}\} = \brakket{\Psi}{\hat{N}}{\Psi}$, because $\gamma_{ii} < \infty$ and the $\trace\{\cdot\}$ only sums over a finite number of elements.

Now we show that $\domain(\hat{\gamma}) \supseteq \domain(\hat{N})$. Since we have for any state $\ket{\Psi} \in \domain(\hat{N})$ that $\infty > \brakket{\Psi}{\hat{N}}{\Psi} = \trace\{\brakket{\Psi}{\hat{\gamma}}{\Psi}\} \geq \brakket{\Psi}{\hat{\gamma}_{ii}}{\Psi}$, because of the positivity of the 1RDM operator. The inequality~\eqref{eq:gammaOffDiagBound} immediately gives $\infty > \gamma_{ii}\gamma_{jj} \geq \abs{\gamma_{ij}}^2$ for any $i,j \in \Nbas$.
\end{proof}

For a given density-matrix operator $\hat{\rho}$ the 1RDM can then be found by $\gamma_{ij}[\hat{\rho}] = \Trace\{\hat{\rho}\,\hat{\gamma}_{ij}\}$. Again, for the bosonic case this is not defined for every possible $\hat{\rho} \in \NdensMat_{+}$.
However, with theorem~\ref{thm:finiteNk}, which implies that $\Trace\{\hat{N} \hat{\rho}_v\} < \infty$, and with~\eqref{eq:gammaOffDiagBound} we have $\abs{\Trace\{\hat{\gamma}_{ij} \hat{\rho}_v\}} < \infty$. So the relevant space to consider for the 1RDMs are hermitian $N_b \times N_b$ matrices with the appropriate constraints for bosons ($+$) and fermions ($-$) 
\begin{subequations}\label{def:NoneMat}
\begin{align}
\NoneMat_+ &\isDefinedAs \set*{\gamma \in \HermitanMat(N_b) }{ \gamma \geq 0 } \, , \\
\NoneMat_- &\isDefinedAs \set*{\gamma \in \HermitanMat(N_b) }{ 
\gamma \geq 0, \gamma^2 \leq \gamma } \, .
\end{align}
\end{subequations}
Note that the fermionic 1RDMs are a subset of the bosonic 1RDMs, $\NoneMat_- = \set*{\gamma \in \NoneMat_+ }{ \gamma^2 \leq \gamma }$.

Per-Olov Löwdin gave the eigenvalues of the 1RDM a special name: natural occupation numbers \citep{Lowdin1955a}. The eigenstates he named the natural (spin-)\linebreak[0]orbitals. He conjectured that the natural orbitals would be the orbitals which would yield the fastest convergence of a configuration(s) interaction (CI) expansion of the wave function. Unfortunately, this is only a peculiar property of the the two-electron system and does not hold for general $N$-electron systems \citep{BytautasIvanicRuedenberg2003, Giesbertz2014}.

The most important use of the natural occupation numbers for our purpose is a famous theorem by Coleman establishing ensemble integer $N$-representability of any fermionic 1RDM \citep{Coleman1963}, which is readily extended to the bosonic case.

\begin{thm}[Coleman]
\label{thm:Coleman}
Any $\gamma \in \NoneMat_{\pm}$ with $\trace\{\gamma\} = N \in \Nats$ is ensemble integer $N$-representable, i.e.\ there always exists a density-matrix operator $\hat{\rho} \in \NdensMat_{\pm}$ containing only $N$-particle states (exactly $N$ creation operators acting on the empty ket) such that $\gamma = \Trace\{\hat{\rho}\,\hat{\gamma}\}$.
\end{thm}

\begin{figure}[t]\centering
\begin{tikzpicture}[scale=1.3,x={(1.1,0.3)},z={(0.5,-0.7)}]
  \draw[->] (0,0,0) -- (4.5,0,0);
  \node[anchor=west] at (4.5,0,0) {$n_2$};
  \draw[->] (0,0,0) -- (0,4.5,0);
  \node[anchor=west] at (0,4.5,0) {$n_3$};
  \draw[->] (0,0,0) -- (0,0,4.5);
  \node[anchor=west] at (0,0,4.5) {$n_1$};
  \filldraw[black] (0,0,0) circle (1.5pt);
  \node[anchor=east] at (0,0,0) {$\Gamma^0_{\pm}$};
  \filldraw[thick, blue,fill opacity=0.2] (1,0,0) -- (0,1,0) -- (0,0,1) -- cycle;
  \filldraw[blue] (1,0,0) circle (0.65pt);
  \filldraw[blue] (0,1,0) circle (0.65pt);
  \filldraw[blue] (0,0,1) circle (0.65pt);
  \node[anchor=east, blue] at (0,0,1) {$\Gamma^1_{\pm}$};
  \draw[thin] (1,0,0) -- (1,1,0) -- (0,1,0) -- (0,1,1) -- (0,0,1) -- (1,0,1) -- cycle;
  \filldraw[thick, red,fill opacity=0.2] (2,0,0) -- (0,2,0) -- (0,0,2) -- cycle;
  \filldraw[red] (2,0,0) circle (0.65pt);
  \filldraw[red] (0,2,0) circle (0.65pt);
  \filldraw[red] (0,0,2) circle (0.65pt);
  \node[anchor=east, red] at (0,0,2) {$\Gamma^2_+$};
  \filldraw[red, fill opacity=0.8] (1,1,0) -- (0,1,1) -- (1,0,1) -- cycle;
  \filldraw[red] (1,1,0) circle (0.65pt);
  \filldraw[red] (0,1,1) circle (0.65pt);
  \filldraw[red] (1,0,1) circle (0.65pt);
  \draw[thin] (1,1,0) -- (1,1,1);
  \draw[thin] (1,0,1) -- (1,1,1);
  \draw[thin] (0,1,1) -- (1,1,1);
  \filldraw[thick, green,fill opacity=0.2] (3,0,0) -- (0,3,0) -- (0,0,3) -- cycle;
  \filldraw[green] (3,0,0) circle (0.65pt);
  \filldraw[green] (0,3,0) circle (0.65pt);
  \filldraw[green] (0,0,3) circle (0.65pt);
  \node[anchor=east, green] at (0,0,3) {$\Gamma^3_+$};
  \filldraw[green] (1,1,1) circle (1.5pt);
  \filldraw[thick, brown,fill opacity=0.2] (4,0,0) -- (0,4,0) -- (0,0,4) -- cycle;
  \filldraw[brown] (4,0,0) circle (0.65pt);
  \filldraw[brown] (0,4,0) circle (0.65pt);
  \filldraw[brown] (0,0,4) circle (0.65pt);
  \node[anchor=east, brown] at (0,0,4) {$\Gamma^4_+$};
\end{tikzpicture}
\caption{Bosonic and fermionic polytopes with $N_b = 3$. The fermionic polytope, $\Gamma^N_-$ is obtained from the bosonic polytope, $\Gamma^N_+$, by constraining it to the unit cube.}
\label{fig:polytopes}
\end{figure}
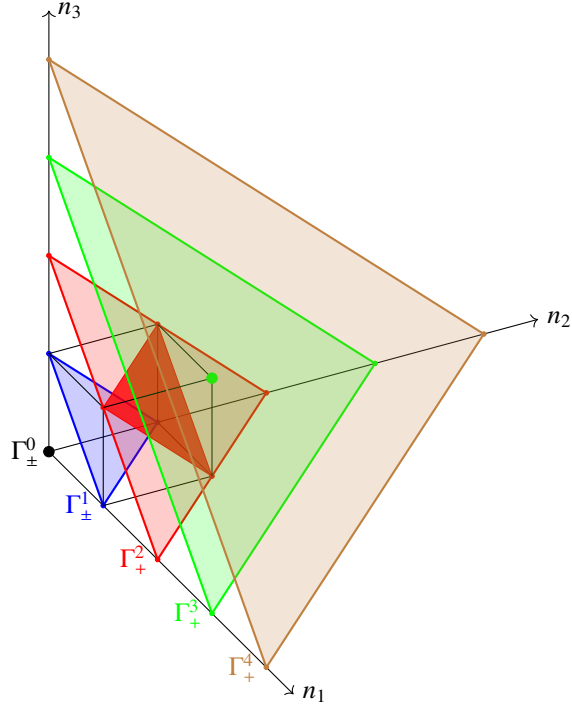

\begin{proof}
Coleman originally considered the fermionic case, but the bosonic case is somewhat simpler, so let us consider that one first. We can always assume that we work in the NO basis. That is, we perform a basis transformation in the single-particle space $\oneH$ such that we diagonalise the $N_b \times N_b$ matrix $\gamma$. In this basis the 1RDM can be expressed as an $N_b$ dimensional vector containing occupation numbers $\mat{n} = \bigl(n_1,n_2,\dotsc,n_{N_b}\bigr)$. Since the sum of the occupation numbers is restricted to be $N$, all the $N$-particle 1RDMs with the same set of NOs ($=\Gamma^N_+$), constitute a convex polytope. This means that each 1RDM can be expressed as a linear combination of its extreme elements. Since these extreme elements are readily identified as the 1RDMs which have one occupation number equal to $N$ and all the other to zero, the extreme elements are scaled unit vectors, $N\mat{e}_i$. So the set of all $N$-boson 1RDMs with a given set of NOs can be expressed as
\begin{equation*}
\Gamma^N_+
= \set*{\sum_{i=1}^{N_b}\lambda_i\,N\mat{e}_i}
{\lambda_i \geq 0, \sum_{i=1}^{N_b}\lambda_i = 1} \, ,
\end{equation*}
which is just a scaled simplex.
Next note that each extreme element is generated by a pure state in which one orbital is occupied $N$-times, $\ket{0\dotsc n_i \dotsc 0}\bra{0\dotsc n_i \dotsc 0} \to N\mat{e}_i$, so the extreme elements are even pure state $N$-representable. Because the map $\hat{\rho} \to \gamma$ is linear, this implies that for each $N$-bosonic 1RDM we can write a density-matrix operator which generates this 1RDM as a linear combination of the pure states generating the extreme points
\begin{equation*}
\hat{\rho}(\mat{n}) = \sum_{i=1}^{N_b}\lambda_i\ket{0\dotsc n_i \dotsc 0}\bra{0\dotsc n_i \dotsc 0} \, .
\end{equation*}
The same strategy works in the fermionic case, except that the polytope has a more complicated shape due to the additional condition $n_i \leq 1$. The extreme points of the fermionic polytope are all possible permutations of $N$ occupation numbers set to one and all others set to zero
\begin{equation*}
\bar{\gamma}_{I} \isDefinedAs
\bar{\gamma}_{i_1\dotsc i_N} \isDefinedAs \mat{e}_{i_1} + \dotsb + \mat{e}_{i_N} \, ,
\end{equation*}
for $1 \leq i_1 < \dotsb < i_N \leq N_b$. The index $I$ is a renumeration of $i_1\dotsc i_N$ and has $K = \binom{N_b}{N}$ elements.
The fermionic polytope can now explicitly be given in terms of these extreme $N$-fermion 1RDMs as
\begin{equation*}
\Gamma^N_-
= \set*{\sum_{I=1}^K\lambda_I\,\bar{\gamma}_I}{\lambda_I \geq 0, \sum_{I=1}^K\lambda_I = 1} \, .
\end{equation*}
The extreme elements can now be identified with all possible $N$-particle determinants,
$\ket{I} \isDefinedAs \ket{i_1\dotsc i_N} \to \bar{\gamma}_{i_1\dotsc i_N}$, so they are also pure state $N$-representable. Using again that the mapping $\hat{\rho} \to \gamma$ is linear any $N$-fermion 1RDM can be generated from a linear combination of the determinants generating the extreme points
\begin{equation*}
\hat{\rho}(\mat{n}) = \sum_{I=1}^K\lambda_I\ket{I}\bra{I} \, . \qedhere
\end{equation*}
\end{proof}

The polytopes used in the proof are illustrated in Fig.~\ref{fig:polytopes} for $N_b = 3$. Note that the fermionic polytopes can be obtained from the bosonic ones by constraining them to the unit (hyper-)\linebreak[0]cube. Since multiple particles are needed to bring out the exchange effects, $\Gamma^0_+ = \Gamma^0_-$ and $\Gamma^1_+ = \Gamma^1_-$. In the $N = 2$ case, the  $\Gamma^2_-$ is the small triangle within the $\Gamma^2_{+}$ polytope. Because there is only one fermionic state with $N = 3$, the $\Gamma^3_-$ is just a single point in the $\Gamma^3_+$ polytope.

Since every 1RDM with a fractional number of $N$ particles can be created as a linear combination between an $\lfloor N \rfloor$- and $\lceil N \rceil$-particle 1RDM, we have immediately the following corollary.

\begin{cor}
\label{cor:ColemanNonInteger}
Any $\gamma \in \NoneMat_{\pm}$ is ensemble $N$-representable, i.e.\ there always exists a density-matrix operator $\hat{\rho} \in \NdensMat_{\pm}$ such that $\gamma = \Trace\{\hat{\rho}\,\hat{\gamma}\}$.
\end{cor}

This corollary is especially useful for the universal function~\eqref{def:universalFunction}, since it implies that we can always find at least one $\hat{\rho} \to \gamma$ for all $\gamma \in \NoneMat$. Since the extreme points in theorem~\ref{thm:Coleman} and corollary~\ref{cor:ColemanNonInteger} span a finite-dimensional space, we have $E[\gamma] < \infty$, $S[\gamma] < \infty$ and $F[\gamma] < \infty$ for all $\gamma \in \NoneMat$. It is therefore natural to consider $F[\gamma]$ for $\gamma \in \NoneMat$. Later we will show that the infimum can be replaced by a minimum if the maximum order of the interactions in the Hamiltonian is finite, $\nmax< \infty$, and strictly positive definite. The existence of a minimum implies that $F[\gamma] > -\infty$ for all $\gamma \in \NoneMat$, and $\NoneMat$ will be the domain of $F[\gamma]$.

At this point it is important to note that the physically relevant 1RDMs are the ones that are associated with a Gibbs state $\hat{\rho}_v$ of a Hamiltonian $\hat{H}_v$. Thus while we have now defined the most general space of (ensemble $N$-representable) 1RDMs $\NoneMat$, it is the set of all $v$-representable 1RDMs
\begin{equation}\label{def:VoneMat}
\VoneMat \isDefinedAs \set{\gamma \in \NoneMat}{\exists\;v \in \FiniteNpot \mapsto \gamma} \, 
\end{equation}
that is central to our considerations. Though there might be many $\hat{\rho}$ that produce a given $\gamma \in \VoneMat$ one of our goals is to show that there is one and only one $\hat{\rho}_v$. A first obvious characterisation is that $\VoneMat \subseteq \NoneMat$. For a more in-depth characterisation of the set of $v$-representable 1RDMs we will employ results from convex analysis in finite dimensions, which we recapitulate now in the following section.

\subsection{Convex and concave functions}
\label{sec:convex}
As mentioned in the introduction, most of the functions we will be dealing with are convex or concave. In the \emph{finite} dimensional case they have some convenient general properties which we can readily exploit, because we work with a finite one-particle basis. These properties are very intuitive and will be illustrated with the help of some figures. Additional mathematical details and proofs of these properties can be found in Appendix~\ref{ap:convex}. For completeness, let us at least give a mathematical definition of a convex (concave function).
\begin{defn}[convex\slash{}concave function]
\label{def:convexFunc}
Consider a set $X$. A function $f \colon X \to \Reals \cup \{+\infty\}$ is called convex if for all $x_1,x_2 \in X$ and $\lambda \in [0,1]$
\begin{equation*}
f(\lambda x_1 + (1-\lambda)x_2) \leq \lambda f(x_1) + (1 - \lambda) f(x_2) \, .
\end{equation*}
The function is called \emph{strictly} convex if for all $x_1 \neq x_2 \in X$ and \( 0 < \lambda < 1 \) there is only an equality if $f(x_1) = -\infty$ or $f(x_2) = -\infty$.
A function $f \colon X \to \Reals \cup \{-\infty\}$ is (strictly) concave if $-f$ is (strictly) convex.
\end{defn}

The definition of a convex function simply means that if we draw a straight line between two points on the graph of a function, the graph of the function needs to be on or below this line. Strict convexity means that the graph of the function is only allowed to lie below the line connecting the points. An example illustrating the concept of convexity is shown in Fig.~\ref{fig:convex}. As (strict) concavity of $f$ simply means that $-f$ is (strictly) convex, we will limit the remainder of the discussion to convex functions.

Note that usually the definition of a convex function $f$ is only given on its domain, $\domain(f) \isDefinedAs \set[\big]{x \in X}{\abs{f(x)} < \infty}$. By allowing a convex function to take on the values $+\infty$, the definition also works over the full set $X$.

Convex functions have a number of convenient properties. The most famous property is that any local minimiser of a convex function is immediately a global minimiser and in the case of strict convex functions, the minimiser is even unique. This property holds even if $X$ is infinite-dimensional, e.g. the bosonic Fock space $\Fock_+$. We will formalise this property into the following theorem, which makes it easier to refer back to it later. The proof can be found in~\ref{ap:unimodal}.

\begin{thm}[Unimodality]\label{thm:unimodal}
Let $f$ be a convex function on a convex subspace $M$, and let $x^* \in M \cap \domain(f)$ be a local minimiser of $f$ on $M$
\begin{equation*}
\exists r > 0  \suchthat f(y) \geq f(x^*) \quad \forall y \in M, \norm{y - x} < r \, .
\end{equation*}
Then $x^*$ is a global minimiser of $f$ on $M$.

If $f$ is strictly convex, then the set of minimisers on $M$ is either empty or contains only one element (singleton).
\end{thm}

\begin{figure}[t]
\centering
\begin{tikzpicture}[scale=0.75]
  \filldraw (-3,5) circle (0.075);
  \draw[dashed] (-3,5) -- (-3,-1.5);
  \draw[thick] plot[domain=-1.5:0] (2*\x,\x*\x) -- plot[domain=0:2] (2*\x+1,\x*\x);
  \filldraw[thick,fill=white] (-3,2.25) circle (0.075);
  \draw[dashed] (5,-0.5) -- (5,6);
  \filldraw[thick,fill=black] (5,4) circle (0.075);
  \node at (3.6,1) {$f(x)$};
  \draw[dashed] (7,-1.5) -- (7,6);
  \draw[thick, red] (-2,1) -- (4,2.25);
  \draw (-3,-0.5) to node[anchor=north] {$\domain(f)$} (5,-0.5);
  \draw (-3,-0.35) -- (-3,-0.65);
  \draw (5,-0.35) -- (5,-0.65);
  \draw (-3,-1.5) to node[anchor=north] {$X$} (7,-1.5);
  \draw (-3,-1.35) -- (-3,-1.65);
  \draw (7,-1.35) -- (7,-1.65);
  \node at (6,6) {$f = +\infty$};
\end{tikzpicture}

\caption{Example of a convex function $f$. Convex means that any straight line segment between two points on the graph of the function lies above or on the graph. An instance is shown by the red straight line segment which lies above all function values. Note that this property only needs to hold on the domain of the function (all $x$ for which $f(x) < \infty$).}
\label{fig:convex}
\end{figure}
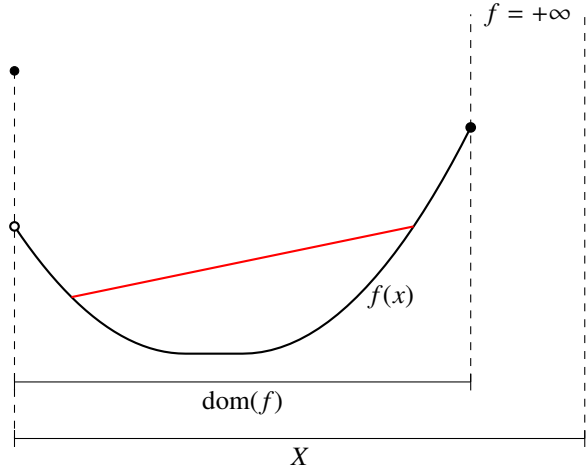

Another convenient property is that convex functions on \emph{finite} dimensional spaces are continuous on the interior of their domain. Interior simply means that we do not include the boundary. From the example in Fig.~\ref{fig:convex} it is clear that a convex function can make jumps at the border of its domain, so clearly convex functions are not necessarily continuous at the edge.

One can even show a somewhat stronger type of continuity: local Lipschitz continuity. Local Lipschitz loosely means that on any closed and finite interval, the function cannot be infinitely steep. Local Lipschitz continuity of a convex function is illustrated in Fig.~\ref{fig:continuity}. Again, we will rephrase this property as a theorem in a more mathematical language, which is explained in more detail in  Appendix~\ref{ap:convexLocLip}.

\begin{thm}\label{thm:convexLocLip}
Let $X$ be a finite-dimensional vector space and $f \colon X \to \Reals \cup \{+\infty\}$ a convex function. The function $f$ is locally Lipschitz continuous on the interior of its domain, $\interiorOf\domain(f)$.
\end{thm}

\begin{figure}[t]
\begin{tabular}{@{}l@{\hspace{0.04\textwidth}}r@{}}
\begin{minipage}[t]{0.48\textwidth}\centering
\begin{tikzpicture}[domain=-1.5:2, scale=0.75]
\filldraw (-3,5) circle (0.075);
\draw[dashed] (-3,5) -- (-3,2.25);
\draw[thick] plot (2*\x,\x*\x);
\filldraw[thick,fill=white] (-3,2.25) circle (0.075);
\draw (-2,2) -- (-2,-1);
\draw (3,3) -- (3,-1);
\draw[<->] (-2,-0.75) to node[anchor=north] {$R$} (3,-0.75);
\draw[thick, red] (2,0) -- (-3,1.25);
\draw[red] (2,0.25) -- (-1.75,0.25);
\draw[red] (1,0.25) -- (2,1);
\draw[red] (0,-0.5) -- (3.25,1.9375);
\draw[red] (2.5,0.625) -- (-1,-0.25);
\draw[thick, red] (0.5,-0.25) -- (3.5,2.75);
\filldraw (1,0.25) circle (0.1);
\node at (0.8,0.7) {$f(x)$};
\end{tikzpicture}
\end{minipage} &
\begin{minipage}[t]{0.48\textwidth}\centering
\begin{tikzpicture}[domain=-1:1]
\draw[dashed] (-2,4) -- (-2,2);
\draw[thick, domain=-2:0] plot (\x,{0.25*((\x-1)*(\x-1) - 1)});
\filldraw[thick, fill=white] (-2,2) circle (0.05);
\filldraw[thick] (-2,4) circle (0.05);
\draw[thick, domain=0:4] plot (\x,{0.1*((\x+1)*(\x+1) - 1)});
\draw[red] (-1.5,-0.3) -- (0,0) -- (1.5,-0.75);
\draw[red] (-1.5,0.375) -- (1.5,-0.375);
\draw[red] (-1.5,0) -- (1.5,0);
\draw[thick, red, <->] (-1.5,0.75) node[anchor=south east] {$f'_{-1}(x)$} -- (0,0) -- (1.5,0.3) node[anchor=west] {$f'_1(x)$};
\filldraw (0,0) circle (0.07);
\node at (0,0.4) {$f(x)$};
\draw[decorate, decoration=brace] (1.6,-1) -- (-1.6,-0.5);
\node[red] at (0,-1.1) {$\du f(x)$};
\end{tikzpicture}
\end{minipage} \\
\begin{minipage}[t]{0.48\textwidth}
\caption{For a convex function we can always draw two straight lines with a non-vertical slope on any region $R$ in the interior of its domain around $x$, so a convex function on a one dimensional space is locally Lipschitz continuous. Note that at the boundary of its domain, the convex function can jump, so a convex function is not necessarily continuous at the boundary of its domain.}
\label{fig:continuity}
\end{minipage} &
\begin{minipage}[t]{0.48\textwidth}
\caption{Directional derivatives and subgradients at $x$ of a 1-dimensional convex function $f$. As there is a kink at $x$, there are multiple subgradients in the subdifferential $\du f$. The directional derivatives are the subgradients with the highest slope in each direction.}
\label{fig:gradients}
\end{minipage}
\end{tabular}
\end{figure}

The fact that we have even local Lipschitz continuity basically infers directly an other convenient property of convex functions on finite-dimensional spaces and that is that the directional derivative exists in each direction.

\begin{defn}[Directional derivative]\label{def:directionalDerivative}
A function $f$ is differentiable at $x$ in the direction $h$ if the following limit exists
\begin{equation*}
f'_h(x) \isDefinedAs \lim_{t \downarrow 0}\frac{f(x + h\,t) - f(x)}{t} \, .
\end{equation*}
\end{defn}

\begin{thm}\label{thm:existenceDirectionalDerivative}
Let $X$ be a finite-dimensional vector space and $f \colon X \to \Reals \cup \{+\infty\}$ a convex function. The function $f$ is differentiable in any direction at any point in the interior of its domain.
\end{thm}

This basically follows directly from the local Lipschitz continuity of finite-dimensional convex functions. A more detailed proof is given in Appendix~\ref{ap:existenceDirectionalDerivative}.
Though a finite-dimensional convex function is differentiable in each direction, these derivatives $f'_h(x)$ are not necessarily linear in $h$, so the gradient (Gâteaux derivative) of $f$ might not exist.
A typical example is the function $f(x) = \abs{x}$, which is not differentiable at $x=0$. Still we have the following directional derivatives: $f'_{-1}(0) = -1$ and $f'_1(0) = 1$. However, it is possible to define a good surrogate for convex functions: the subdifferential. The subdifferential is the set of all possible lines or (hyper)planes one can draw through a point of the graph of a function which do not cross the function at any point (see Fig.~\ref{fig:gradients}). This idea is made mathematically more precise in the following definition.

\begin{defn}[Subgradient and subdifferential of a finite-dimensional convex function]
\label{def:convexSubGrad}
Let $X$ be a finite-dimensional vector space and $f \colon X \to \Reals \cup \{+\infty\}$ be a convex function. Then $h \in X$%
\footnote{The definition also works in an infinite-dimensional setting with the modification that the subgradients reside in the dual of $X$ (definition~\ref{def:duals}), so $h \in X^*$. The (topological) dual is needed to ensure that $\abs{\braket{h}{x}} < \infty$.} 
is called a subgradient of $f$ at $x \in \domain(f)$ if for any $y \in \domain(f)$ we have
\begin{equation*}
f(y) \geq f(x) + \braket{h}{y - x} \, .
\end{equation*}
The set $\du f(x)$ of all subgradients of $f$ at $x$ is called the subdifferential of $f$ at $x$.
\end{defn}

The subdifferential of a finite-dimensional convex function has the following properties, which are proven in Appendix~\ref{ap:subdifferential} 

\begin{thm}\label{thm:subdifferential}
Let $X$ be a finite-dimensional vector space and $f \colon X \to \Reals \cup \{+\infty\}$ be a convex function. Then

\begin{theonum}
\item\label{thm:nonEmptySubDif} the set $\du f(x)$ is nonempty,

\item\label{thm:subDifCompact} the set $\du f(x)$ is compact (closed and bounded) and convex,

\item\label{thm:subDifDirDif} for any $h \in X$, 
$f'_h(x) = \max\set{ \braket{d}{h} }{ d \in \du f(x) }$,

\item\label{thm:uniqueSubDif} $f$ is differentiable if and only if the subdifferential contains only one element. In that case this element equals the usual gradient: $\du f(x) = \{\nabla f(x)\}$.
\end{theonum}
\end{thm}

This theorem is illustrated in Fig.~\ref{fig:gradients}.
At each point at the interior we can clearly draw a tangent line to the graph of $f$ which is completely below the graph (part~\ref{thm:nonEmptySubDif}).
At the kink we can draw multiple subgradients, so the subdifferential contains more than one element. The directional derivatives are the subgradients with maximum slope in each direction (part~\ref{thm:subDifDirDif}).
All subgradients are a convex combination of the directional derivatives, so $\du f(x)$ is convex. Since the directional derivatives are finite and contained in $\du f(x)$, it is also compact (part~\ref{thm:subDifCompact}).
At any other point in the interior we clearly have only one tangent line, so one element in $\du f(x)$, which obviously needs to be equal to the derivative of $f$ (part~\ref{thm:uniqueSubDif}).

Property $iv$ of theorem~\ref{thm:subdifferential} is particular useful for us. Since if we can show that a finite-dimensional function is convex, we only need to show uniqueness of the tangent to proof that the function is differentiable on the interior of its domain. The only other task is then to characterise (the interior of) the domain.

At this point we would like to make a connection to the Gâteaux derivative often encountered in formal DFT. The Gâteaux derivative is the directional derivative (Definition~\ref{def:directionalDerivative}) with the additional requirement that it is linear and continuous in its direction $h$. In the finite-dimensional case linearity automatically implies continuity, but in the infinite-dimensional case, continuity cannot be taken for granted anymore. It is exactly the continuity property which causes most trouble. As pointed out by Lammert \citep{Lammert2006a}, this complication has been overlooked by Englisch \& Englisch in their proof for the differentiability of the universal function in DFT \citep{EnglischEnglisch1984a, EnglischEnglisch1984b} and unfortunately repeated by many others \citep{PhD-Leeuwen1994, Eschrig1996, Farid1998, Farid1999, HolasMarch2002, Leeuwen2003, LindgrenSalomonson2003, Eschrig2003, LindgrenSalomonson2004, ZaharievWang2004, Ayers2006, Eschrig2010}, see also~\cite[p.~50]{DreizlerGross1990}.
Several remedies have been proposed by introducing some regularization. The original problem is then approached by taking the limit to no regularization. Lammert proposed to coarse-grain the density by partitioning the space into cells \citep{Lammert2006b, Lammert2010}. In the limit of small cells, the original system is recovered. Kvaal et al.\ proposed to use the Moreau--Yosida regularization, which adds a smoothening term to the functional whose contribution is adjusted by a constant $\epsilon$. For any $\epsilon > 0$ the functional is now differentiable and in the limit $\epsilon \to 0^+$, the original system is recovered \citep{KvaalEkstromTeale2014}.
Apart from the issue of Gâteaux differentiability, there are several other difficulties which one needs to deal with in the infinite-dimensional case. Together with the differentiability issues, these difficulties are discussed later in Section~\ref{sec:infiniteDiscus}.

\subsection{General properties of the grand potential and implications on the universal functional}
\label{sec:grandpotential}

As we have shown in the introduction, the existence of a density-matrix operator which minimises the grand potential cannot be taken for granted. In this section we discuss important consequences if the canonical density-matrix operator~\eqref{eq:equiRho} does exist, i.e.\ if $Z[v] < \infty$. Later we will show that this is the case for any potential in the fermionic case in corollary~\ref{cor:minimumgrandpotential}. Some additional restrictions on the potential are needed in the bosonic case, as shown in theorem~\ref{thm:bosonicOmegaMinExist}.
With this assumption it is easy to establish the following.
\begin{thm}\label{thm:rhoToH}
For $v \in \FiniteNpot$, the mapping $\hat{H}_v \mapsto \hat{\rho}_v$ is invertible up to a constant in the Hamiltonian, i.e.\ $h^{(0)}$ in~\eqref{eq:ncHam}.
\end{thm}

\begin{proof}
We can use~\eqref{eq:interimSol} to proof the theorem in the same manner as the first Hohenberg--Kohn theorem \citep{HohenbergKohn1964}. Assume that two different Hamiltonians, $\hat{H}_v$ and $\hat{H}_v'$, yield the same density-matrix operator $\hat{\rho}_v$. Since~\eqref{eq:interimSol} holds for both $\hat{H}_v$ and $\hat{H}_v'$ with constants $C$ and $C'$ respectively, we can subtract the two equations, which yields $\hat{H}_v - \hat{H}_v' = C - C' = \text{constant}$.
\end{proof}

\begin{cor}\label{cor:rhoToV}
For $v \in \FiniteNpot$, the mapping $v \mapsto \hat{\rho}_v$ is invertible.
\end{cor}

A significantly more elaborate proof of this corollary can be found in~\cite[p.~28]{PhD-Baldsiefen2012} and~\cite[p.~3]{BaldsiefenCangiGross2015}.

Another important observation is that the canonical density-matrix operator is strictly positive definite, $\hat{\rho}_v > 0$, so these density-matrix operators reside in the following subspace of $\NdensMat_{\pm}$
\begin{equation}\label{def:NTdensMat}
\NTdensMat_{\pm} \isDefinedAs \set*{\hat{\rho} \colon \Fock_{\pm} \to \Fock_{\pm} }{ 
\hat{\rho} = \hat{\rho}^{\dagger}, \hat{\rho} > 0, \Trace\{\hat{\rho}\} = 1 } \, .
\end{equation}
This is consistent with the notion that at finite temperature, all eigenstates of the Hamiltonian $\ket{\Psi_I}$ contribute with a Boltzmann weight $w_i = \e^{-\beta E_I}/Z > 0$. This justifies that we only took the constraint $\Trace\{\hat{\rho}\} = 1$ into account in the minimisation procedure~\eqref{eq:interimSol} and not the positivity of the ensemble weights.
Note that $\NdensMat_{\pm}$ forms the closure of $\NTdensMat_{\pm}$.

This observation also implies that the corresponding 1RDMs have $n_i > 0$ and in the fermionic case $n_i < 1$ additionally.
To show strict positivity, we work in the NO basis.
First note that for any state $\ket{\Psi_I}$, we have $\brakket{\Psi_I}{\crea{a}_i\anni{a}_i}{\Psi_I} = \norm{\anni{a}_i \Psi_I}^2 \geq 0$. As the eigenstates of the Hamiltonian form a complete basis in the Fock space, the NO $i$ contributes to at least one of these eigenstates, so for at least one of these eigenstates $\ket{\Psi_I}$ we have $\norm{\anni{a}_i \Psi_I}^2 > 0$. As $w_i = \e^{-\beta E_I}/Z > 0$, we immediately find the following lower bound
\begin{equation}
n_i = \sum_Iw_I\brakket{\Psi_I}{\crea{a}_i\anni{a}_i}{\Psi_I} > 0 \, .
\end{equation}
In the case of fermions, the anti-commutation implies $\brakket{\Psi_I}{\crea{a}_i\anni{a}_i}{\Psi_I}
= 1 - \brakket{\Psi_I}{\anni{a}_i\crea{a}_i}{\Psi_I} = 1 - \norm{\crea{a}_i \Psi_I}^2 \leq 1$. Similarly, as the eigenstates of the Hamiltonian form a complete basis in the Fock space, the $i$-th NO cannot be omnipresent in all eigenstates. So for at least one of these eigenstates $\ket{\Psi_I}$ we have $\norm{\anni{a}_i \Psi_I}^2 < 1$, which yields the following upper bound for fermions
\begin{equation}
n_i = \sum_Iw_I\brakket{\Psi_I}{\crea{a}_i\anni{a}_i}{\Psi_I} < 1 \, .
\end{equation}
The 1RDMs produced by a potential therefore reside only in the interior of $\NoneMat_{\pm}$
\begin{subequations}\label{defs:ToneMats}
\begin{align}
\ToneMat_+ &\isDefinedAs \interior(\NoneMat_+)
= \set*{\gamma \in \HermitanMat(N_b) }{ \gamma > 0 } \, , \\
\ToneMat_- &\isDefinedAs \interior(\NoneMat_-)
= \set*{\gamma \in \HermitanMat(N_b) }{ 
\gamma > 0, \gamma^2 < \gamma } \, .
\end{align}
\end{subequations}
We will show momentarily that the interior of the $N$-representable 1RDMs can actually be identified with the set of $v$-representable 1RDMs, i.e.\ $\ToneMat = \VoneMat$. But first, we need to show some additional properties of the (canonical) grand potential.

Let us consider the value of the grand potential evaluated at the canonical density-matrix operator, the canonical grand potential
\begin{equation}\label{eq:OmegaVmin}
\Omega[v] \isDefinedAs \min_{\hat{\rho} \in \NdensMat}\Omega_v[\hat{\rho}]
= \min_{\hat{\rho} \in \NTdensMat}\Omega_v[\hat{\rho}]
= -\beta^{-1}\ln\bigl(Z[v]\bigr) \, .
\end{equation}
Since $\Omega[v]$ is obtained by minimisation of $\Omega_v[\hat{\rho}]$, it is readily shown to be concave \citep{Eschrig2010}.

\begin{thm}
$\Omega[v]$ is strictly concave in $v$.
\end{thm}

\begin{proof}
Concavity trivially follows from its expression as a minimisation~\eqref{eq:OmegaVmin}. So for $v_1 \neq v_2$ and $0 < t < 1$ we have
\begin{align}
\Omega[tv_1 + (1-t)v_2]
&= \min_{\hat{\rho} \in \NTdensMat}\Trace\biggl\{\hat{\rho}\Bigl(t \hat{H}_{v_1} + (1-t)\hat{H}_{v_2} + \frac{1}{\beta}\ln(\hat{\rho})\Bigr)\biggr\} \\
&> t\min_{\hat{\rho}_1 \in \NTdensMat}\Trace\biggl\{\hat{\rho}_1\Bigl(\hat{H}_{v_1} + \frac{1}{\beta}\ln(\hat{\rho}_1)\Bigr)\biggr\} +
(1-t) \min_{\hat{\rho}_2 \in \NTdensMat}\Trace\biggl\{\hat{\rho}_2\Bigl(\hat{H}_{v_2} + \frac{1}{\beta}\ln(\hat{\rho}_2)\Bigr)\biggr\}
= t\Omega[v_1] + (1-t)\Omega[v_2] \, , \notag
\end{align}
where the strict inequality follows from the fact that the minimiser of $\Omega_v[\hat{\rho}]$ is unique (corollary~\ref{cor:rhoToV}).
\end{proof}

From corollary~\ref{cor:rhoToV}, Mermin's generalisation of the Hohenberg--Kohn theorem \citep{Mermin1965} follows directly.

\begin{thm}[Mermin]\label{thm:Mermin}
For $v \in \FiniteNpot$, the map $v \mapsto \gamma_v$ is invertible, i.e.\ the potential which generates a particular 1RDM is unique.
\end{thm}

\begin{proof}
The proof goes by reductio ad absurdum, so assume that there are two different potentials, $v_1 \neq v_2 \mapsto \hat{\rho}_1 \neq \hat{\rho}_2$ which both yield the same 1RDM, $\gamma$.
\begin{align*}
\Omega[v_1] &= \Omega_{v_1}[\hat{\rho}_1]
= \Omega_{v_2}[\hat{\rho}_1] + \Trace\{\gamma(v_1 - v_2)\} \\
&> \Omega_{v_2}[\hat{\rho}_2] + \Trace\{\gamma(v_1 - v_2)\}
= \Omega[v_2] + \Trace\{\gamma(v_1 - v_2)\} \, .
\end{align*}
Now turning the roles of $v_1$ and $v_2$ around and adding the two equations to each other we find the inconsistency
\begin{equation*}
\Omega[v_1] + \Omega[v_2] > \Omega[v_2] + \Omega[v_1] \, .
\end{equation*}
Hence, our assumption that there are two one-body potentials which yield the same 1RDM is false.
\end{proof}

Now we would like to show strict convexity of the universal functional, $F[\gamma]$. For this, we need to show strict convexity of $\Omega_v[\hat{\rho}]$. As the energy~\eqref{def:Energy} is linear, we only need to show that the entropy~\eqref{def:entropy} is strictly concave \citep{Ruelle1969, Lieb1975, Wehrl1978}.
\begin{thm}
The entropy is strictly concave. That is, for any \( \hat{\rho}_{\lambda} = \lambda\hat{\rho}_0 + (1-\lambda)\hat{\rho}_1 \in \NdensMat \) with $\lambda \in [0,1]$ we have $S[\hat{\rho}_{\lambda}] \geq \lambda S[\hat{\rho}_0] + (1 - \lambda)S[\hat{\rho}_1]$. For \( \lambda \in (0,1) \) and $\hat{\rho}_0 \neq \hat{\rho}_1$ we only have an equality if $S[\hat{\rho}_0] = \infty$ or $S[\hat{\rho}_1] = \infty$.
\end{thm}

\begin{proof}
If $S[\hat{\rho}_{\lambda}] = \infty$, we immediately find the inequality and only when $S[\hat{\rho}_0] = \infty$ or $S[\hat{\rho}_1] = \infty$, we have an equality.

Now consider the situation when $S[\hat{\rho}_{\lambda}] < \infty$.
Let $\hat{\rho}_{\lambda} = \sum_kw_k\ket{\Psi_k}\bra{\Psi_k}$.
Strict concavity of the entropy now follows directly from the strict concavity of the function $s(x) = -x\ln(x)$.
\begin{align*}
S[\hat{\rho}_{\lambda}] &= -\sum_kw_k\ln(w_k)
= \sum_ks\bigl(\brakket{\Psi_k}{\hat{\rho}_{\lambda}}{\Psi_k}\bigr) 
= \sum_ks\bigl(\lambda\brakket{\Psi_k}{\hat{\rho}_0}{\Psi_k} +
(1-\lambda)\brakket{\Psi_k}{\hat{\rho}_1}{\Psi_k}\bigr) \\*
&> \lambda\sum_ks\bigl(\braket{\Psi_k}{\hat{\rho}_0 \, \Psi_k}\bigr) +
(1 - \lambda)\sum_ks\bigl(\braket{\Psi_k}{\hat{\rho}_1 \, \Psi_k}\bigr) \\*
&\geq \lambda\sum_k\braket{\Psi_k}{s(\hat{\rho}_0) \, \Psi_k} +
(1 - \lambda)\sum_k\braket{\Psi_k}{s(\hat{\rho}_1) \, \Psi_k}
= \lambda S[\hat{\rho}_0] + (1 - \lambda) S[\hat{\rho}_1] \, .
\end{align*}
The last inequality follows from Jensen's inequality (Lemma~\ref{thm:Jensen}), which is simply extending the convexity (concavity) definition over a convex combination of more than two points. 
\end{proof}

\begin{cor}\label{cor:OmegaStrictConvex}
The grand potential $\Omega_v[\hat{\rho}]$ is strictly convex in the density-matrix operator, $\hat{\rho}$.
\end{cor}

Note that the strict convexity of $\Omega_v[\hat{\rho}]$ implies that its minimiser $\hat{\rho}_v$ is unique if it exists (see Theorem~\ref{thm:unimodal}) which is in agreement with Theorem~\ref{thm:rhoToH} and Corollary~\ref{cor:rhoToV} . Indeed, from a minimalists point of view we could have avoided to proof Theorem~\ref{thm:rhoToH} in the usual Hohenberg--Kohn way and just stated it as a corollary at this point. But for the sake of simplicity we kept it separate. From the strict convexity of the grand potential $\Omega_v[\hat{\rho}]$, we can readily establish the desired property of the universal functional

\begin{thm}\label{thm:Fconvex}
The universal functional $F[\gamma]$ is strictly convex on $\NoneMat$.
\end{thm}

\begin{proof}
Let $\gamma_{\lambda} = \lambda\gamma_1 + (1-\lambda)\gamma_2$. 
Using the strict convexity of $\Omega_v[\hat{\rho}]$, we find
\begin{align*}
\lambda F[\gamma_1] + (1-\lambda)F[\gamma_2]
&= \lambda\;\inf_{\crampedclap{\hat{\rho}_1 \to \gamma_1}}\;\Omega_0[\hat{\rho}_1] +
(1-\lambda)\;\inf_{\crampedclap{\hat{\rho}_2 \to \gamma_2}}\;\Omega_0[\hat{\rho}_2] \\
&> \inf_{\hat{\rho}_1 \to \gamma_1}\;
\;\inf_{\crampedclap{\hat{\rho}_2 \to \gamma_2}}
\Omega_0[\lambda\hat{\rho}_1 + (1-\lambda)\hat{\rho}_2]
= \inf_{\crampedclap{\hat{\rho} \to \gamma_{\lambda}}}
\Omega_0[\hat{\rho}] = F[\gamma_{\lambda}] \, . \qedhere
\end{align*}
\end{proof}

Since $F \colon \NoneMat \to \Reals$ is convex, it will have all the nice properties discussed before in~\ref{sec:convex} on the interior of its domain $\ToneMat$~\eqref{defs:ToneMats}.
If we can additionally show that the infimum can be replaced by a minimum, then we have the very nice property of $\gamma \in \ToneMat$ are not only $N$-representable, but that there even exists a canonical density-matrix operator which generates them, $\hat{\rho}_v \mapsto \gamma_v$. So every $\gamma \in \ToneMat$ is even $v$-representable, which implies that the universal function is differentiable. This is the main result of this work and is made more precise in the following theorem.

\begin{thm}\label{thm:vRepr}
If the minimum in~\eqref{def:universalFunction} is attained, then
a) $\ToneMat = \VoneMat$ and b) the universal functional $F[\gamma]$ is differentiable on the interior of its domain $\ToneMat$.
\end{thm}

\begin{proof}
As $F$ is convex, it has at least one subgradient, $h \in \du F[\gamma]$ for any $\gamma \in \ToneMat$. So $F[\tilde{\gamma}] + \braket{-h}{\tilde{\gamma}} \geq F[\gamma] + \braket{-h}{\gamma}$ for all $\tilde{\gamma} \in \NoneMat$. This implies that
\begin{equation*}
F[\gamma] + \braket{-h}{\gamma}
= \min_{\crampedclap{\tilde{\gamma} \in \NoneMat}}\,
\bigl(F[\tilde{\gamma}] + \braket{-h}{\tilde{\gamma}}\bigr)
= \Omega[-h] \, .
\end{equation*}
Hence, the negative of any subgradient, $-h$, yields a potential generating $\gamma$. However, from Mermin's theorem~\ref{thm:Mermin}, we have that only one such potential exist, so any $\gamma \in \ToneMat$ is uniquely $v$-representable. This implies also, that there is only one subgradient, \( \du F[\gamma_v] = \{-v\} \), so $F[\gamma]$ is differentiable for all $\gamma \in \ToneMat$ and equals minus the potential which yielded $\gamma$, i.e.\ $\nabla F[\gamma_v] = -v$.
\end{proof}

\begin{figure}[!t]
\begin{tikzpicture}[descrip/.style={draw=none, color=purple, text width=1.25cm, scale=2},
        text width=3cm, align=flush center, grow cyclic,
	edge from parent/.style={very thick, {Latex[length=10pt]}-, draw},
	every node/.style={thick, draw=black, rounded corners},
	level 1/.style={level distance=2cm, sibling angle=180},
	level 2/.style={level distance=2.5cm, sibling angle=90},
	level 3/.style={level distance=3.3cm, sibling angle=50}]
  \node {$F[\gamma]$ is differentiable on $\ToneMat$ (Thm~\ref{thm:vRepr}b)}
    child { node[name=nodeVunique] {$v \to \gamma$ is unique (Mermin, Thm~\ref{thm:Mermin})} }
    child { node {$\ToneMat = \VoneMat$ (Thm~\ref{thm:vRepr}a)}
      child[level distance=5.7cm] { node[name=nodeinfToMinFferm] {$\inf \to \min$ in $F_-[\gamma]$ (Thm~\ref{thm:extremeValue} \& Cor.~\ref{cor:minimumgrandpotential})}
        child[level distance=2.5cm, grow=90] { node[name=nodeFfermCont] {$\Omega^-_v[\hat{\rho}]$ is continuous (Thm~\ref{thm:convexLocLip})} [level distance=2.75cm]
          child[grow=60] { node[name=nodeNfinite, text width=2cm] {$\NdensMat_-$ is finite-dimensional} }
          child[grow=120] { node[text width=2.25cm] {$\Omega_v[\hat{\rho}]$ is convex (Cor.~\ref{cor:OmegaStrictConvex})} }
        }
      }
      child { node {existence of subgradient (Thm~\ref{thm:subdifferential})}
        child { node[name=nodeFconvex, text width=2cm] {$F[\gamma]$ convex (Thm~\ref{thm:Fconvex})} }
        child { node[name=nodeFbounded, text width=2.6cm] {$F[\gamma] < \infty$ on $\NoneMat$ (Thm~\ref{thm:Coleman} \& Cor.~\ref{cor:ColemanNonInteger})} }
      }
      child[level distance=5.7cm] { node {$\inf \to \min$ in $F_+[\gamma]$ (Thm~\ref{thm:FinfToMinBosons})}
        child[grow=-90, level distance=2cm] { node[text width=4cm] {$\Omega_0^+[\hat{\rho}]$ is weak-* lower semicontinuous (Cor.~\ref{cor:OmegaWeakStarlsc})}
          child[level distance=1.75cm] { node {$Z[0] < \infty$ (Thm~\ref{thm:finiteZ})}
            child { node {$h^{(\nmax)} > 0$ (Sec.~\ref{sec:potentials})} }
          }
        }
        child[grow=90, level distance=2.5cm] { node[text width=4cm] {$\NdensMat_+$ is compact in the weak-* topology (Banach--Alaoglu, Thm~\ref{thm:BanachAlaoglu})} }
      }
    };
  \path (nodeNfinite.south) -- +(0.75,0) coordinate(nodeNfiniteR);
  \draw[very thick, -{Latex[length=10pt]}] (nodeNfiniteR) |- (nodeinfToMinFferm);
  
  \node[below=0.2 of nodeVunique, descrip] (nodeGen) {general};
  \node[base right=3.4 of nodeGen, descrip] {fermions};
  \node[base left=3.3 of nodeGen, descrip] {bosons};
  
   \begin{pgfonlayer}{background}
     \clip[xshift=0.2cm] (-0.5\textwidth,10cm) rectangle ++(\textwidth,-13.9cm);
     \colorlet{bosonColor}{blue!25};
     \colorlet{generalColor}{orange!25};
     \colorlet{fermionColor}{green!50!black!25};
     
     \path (nodeFconvex.east) -- ++(0.6,1) coordinate(coordFconvex);
     \path (nodeFbounded.west) -- ++(-0.5,1) coordinate(coordFbound);
     \fill [bosonColor] (coordFbound) rectangle ++(-10,-20);
     \fill [generalColor] (coordFbound) -- (coordFconvex) -- ++(0,-20) -| cycle;
     \fill [fermionColor]  (coordFconvex) rectangle ++(10,-20);
     
     \shade [left color=bosonColor, right color=generalColor]
       ([xshift=-0.3cm]coordFbound) rectangle ++(0.8,-20);
     \shade [left color=generalColor, right color=fermionColor]
       ([xshift=-0.6cm]coordFconvex) rectangle ++(0.8,-20);
   \end{pgfonlayer}
\end{tikzpicture}
\caption{Overview how the most important theorems lead to the differentiability of the universal functional. Note the simplicity in the fermionic case due to the finite dimension of the fermionic Fock space $\Fock_-$ and hence of $\NdensMat_+$. To prove the existence of the minimum in the bosonic case additional assumptions are needed ($h^{(\nmax)} > 0$) and more powerful mathematics.}
\label{fig:overview}
\end{figure}
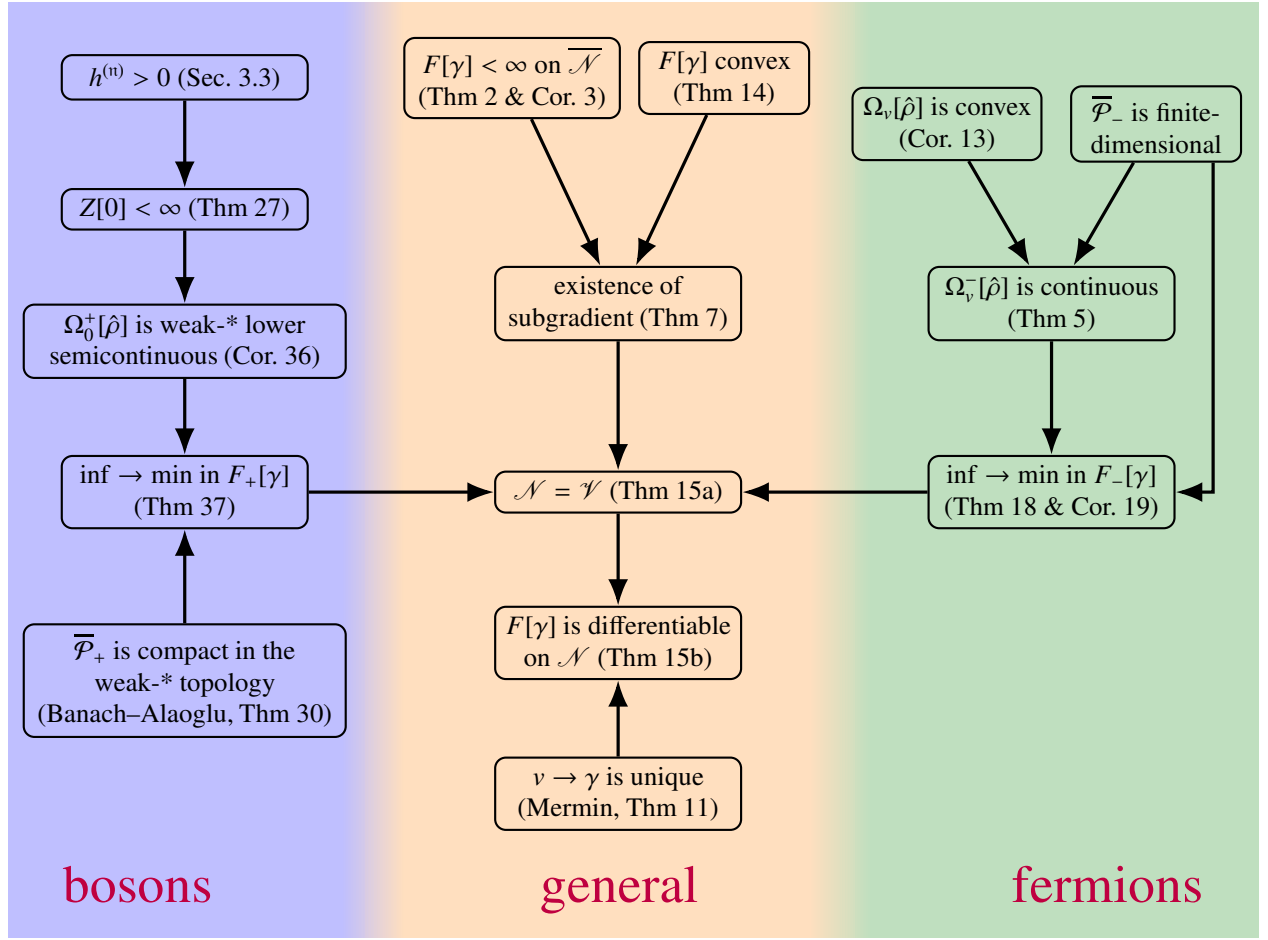

At this point it becomes useful to take a look at the scheme in Fig.~\ref{fig:overview} which presents an overview of the most important theorems and how they are connected. So far we have been dealing with the general part of the theory. From the scheme it is clear, that to make the theory fly, we still need to show that the infima can be replaced by minima. The simplicity of the scheme for the fermions stresses again that the fermionic case will be relatively straightforward, whereas the bosonic will be more complicated.

\section{The fermionic case}
\label{sec:fermions}

In this section we specialise to the fermionic case. In this finite-dimensional setting we will show that the infima can be replaced by minima. These results are needed to substantiate the previous section.
Since the fermionic Fock space is finite-dimensional, we immediately have from Theorem~\ref{thm:convexLocLip}

\begin{cor}\label{thm:fermScont}
The fermionic energy, entropy and grand potential are locally Lipschitz continuous.
\end{cor}

For the energy, we even have a somewhat stronger continuity property.

\begin{prop}
The fermionic energy $E_v[\hat{\rho}]$ is (globally) Lipschitz continuous.
\end{prop}

\begin{proof}
Since the Hamiltonian acts on a finite Hilbert space, it has a largest singular value, $\norm{\hat{H}_v}_{\infty} < \infty$. So for any sequence of density-matrix operators $\hat{\rho}_n \to \hat{\rho}$ for $n \to \infty$, we have
\begin{equation*}
\abs{E_v[\hat{\rho}_n] - E_v[\hat{\rho}]}
= \abs[\big]{\Trace\{\hat{H}_v(\hat{\rho}_n - \hat{\rho})\}}
\leq \norm{\hat{H}_v}_{\infty}\abs[\big]{\Trace\{(\hat{\rho}_n - \hat{\rho})\}}
\leq \norm{\hat{H}_v}_{\infty}\norm{\hat{\rho}_n - \hat{\rho}}_1 \, .
\end{equation*}
So for $\norm{\hat{\rho}_n - \hat{\rho}}_1 \to 0$ we find that $\abs{E_v[\hat{\rho}_n] - E_v[\hat{\rho}]} \to 0$. Since the convergence is linear with respect to $\norm{\hat{\rho}_n - \hat{\rho}}_1$ with a global constant, $\norm{\hat{H}_v}_{\infty}$, the fermionic energy is even Lipschitz continuous with the Lipschitz constant $\norm{\hat{H}_v}_{\infty}$.
\end{proof}

Since the grand potential is strictly convex with respect to $\hat{\rho}$ and the universal functional with respect to $\gamma$, we know that the respective minimisers will be unique if they exist. The existence of the minimisers in the fermionic case is guaranteed by the following theorem, because $\NdensMat_-$ is compact (closed and bounded).

\begin{thm}[Extreme value]\label{thm:extremeValue}
Let $f \colon X \to \Reals$ be a continuous function and $M \subseteq X$ a compact (closed and bounded) set. Then $f$ must attain a minimum and maximum at least once.
\end{thm}

\begin{proof}
First we proof that a continuous function $f \colon X \to \Reals$ is bounded on a compact space $M \subseteq X$. We do this by reductio ad absurdum for the upper bound, so suppose that $f$ is not bounded above on $M$. Then for every natural number $n$ there exists an $x_n \in M$ such that $f(x_n) > n$, so we have a sequence $\{x_n\}$. Since $M$ is compact, this sequence has a convergent subsequence $\{x_{n_l}\}$ with a limit $x \in M$, cf.\ definition~\ref{def:compactSet}.
Because $f$ is continuous, $f(x_l)$ converges to $f(x) \in \Reals$. But $f(x_{n_l}) > n_l$ for all $l$ implies that $f(x_{n_l})$ diverges to $+\infty$, so a contradiction. Therefore the initial assumption that $f$ would be unbounded is incorrect. We can repeat the proof in a similar manner to proof that $f$ has a lower bound on $M$.

Since any converging sequence converges in a compact space, also any converging sequence $\{x_n\}$ to the maximiser (minimiser) converges to some $x \in M$. So the maximum (minimum) is attained for some $x \in M$.
\end{proof}

As $\NdensMat_-$ is finite-dimensional it is compact. Thus, from the extreme value theorem we immediately have the following corollary.

\begin{cor}
\label{cor:minimumgrandpotential}
The minimum in the fermionic grand potential $\Omega_-[v]$ and fermionic universal functional $F_-[\gamma]$ are achieved, so the infima in~\eqref{eq:OmegaVpart} and~\eqref{def:universalFunction} can be replaced by minima. Additionally, the replacing the infima by minima in~\eqref{eq:OmegaVmin} is justified in the fermionic case.
\end{cor}

\section{The bosonic case}
\label{sec:bosons}

The bosonic case is more complicated, because we need to deal with an infinite-dimensional Fock space. However, the infinity is only caused by an unbounded number of particles, so we can keep everything relatively well under control. This section is split in three parts. First we will show that the bosonic grand potential has a minimum if and only if $Z[v] < \infty$. In the second part we will show that if the Hamiltonian has a maximum-order interaction $0 < \nmax < \infty$ which is strictly positive definite, that $Z[v] < \infty$ and $\av{N^k}_v = \Trace\{\hat{\rho}_v\hat{N}^k\} < \infty$. The case $k=1$ is important for our theoretical setting, as it guarantees that all $v \in \FiniteNpot$ yield a proper 1RDM with finite entries, $\gamma \in \NoneMat$.
The last part of this section is attributed to showing that the minimum in the bosonic universal functional is achieved.

\subsection{When does the minimum of the grand potential exist?}
\label{sec:minimumgrandpotential}

Let us again consider the entropy first.
Though we have shown concavity of the bosonic entropy, it does not imply continuity. Because the domain $\NdensMat_+$ is infinite-dimensional, theorem~\ref{thm:convexLocLip} does not apply anymore. As a matter of fact, the bosonic entropy is even not continuous (in trace norm), because it is unbounded in every neighbourhood \citep{Wehrl1978}.

\begin{thm}
\label{thm:domainSbosonic}
Let $\rho \in \NdensMat_+$ and $\epsilon > 0$. Then there always exist another density-matrix operator $\hat{\rho}'$ such that $\norm{\hat{\rho} - \hat{\rho'}} < \epsilon$ and $S[\hat{\rho}'] = \infty$.
\end{thm}

\begin{proof}
The proof goes by construction. If $\hat{\rho}$ has an infinite number of weights $w_k \neq 0$, we can always find a sufficiently large index $l$, such that $L = 1 - \sum_{k=1}^lw_k < \epsilon/2$. So if we set $w_1' = w_1, \dotsc, w_l' = w_l$, then we already have that $\norm{\hat{\rho} - \hat{\rho}'} < \epsilon$. Now set the remaining weights as
\begin{equation}\label{eq:weightChoice}
w'_k = \frac{A}{k(\ln k)^2} \quad \text{for $k > l$},
\end{equation}
where $A > 0$ is a normalisation constant such that $\norm{\hat{\rho}'} = 1$. That such a constant exists, i.e.\ that $\norm{\hat{\rho}'} < \infty$ follows from
\begin{equation*}
\sum_{\crampedclap{k=r}}^{\infty}\;\frac{A}{k(\ln k)^2}
\leq \frac{A}{r(\ln r)^2} + \binteg{k}{r}{\infty}\frac{A}{k(\ln k)^2}
= \frac{A}{r(\ln r)^2} + \binteg{u}{\ln r}{\infty}\frac{A}{u^2}
= \frac{A}{r(\ln r)^2} + \frac{A}{\ln r} < \infty \, ,
\end{equation*}
for any $r > 1$.
On the other hand we can partition the contribution to the entropy as
\begin{equation*}
\sum_{k=r}^{\infty}\frac{-A}{k(\ln k)^2}\ln\biggl(\frac{A}{k(\ln k)^2}\biggr)
= \sum_{k=r}^{\infty}\left(\frac{-A\ln(A)}{k(\ln k)^2} +
\frac{A}{k\ln k} + \frac{2A\ln(\ln k)}{k (\ln k)^2}\right) \, .
\end{equation*}
We have already seen that the first sum converges. The third sum also converges, since
\begin{equation*}
\sum_{k=r}^{\infty}\frac{\ln(\ln k)}{k (\ln k)^2}
\leq C + \binteg{u}{\ln r}{\infty}\frac{\ln u }{u^2}
= C + \binteg{u}{\ln r}{\infty}\frac{\ln u }{u^2}
= C + \binteg{x}{\ln(\ln r)}{\infty} x \e^{-x}
= C + \frac{1 + \ln(\ln r)}{\ln r} < \infty \, ,
\end{equation*}
for $r > \e$ and $C = \ln(\ln r) / \bigl(r (\ln r)^2\bigr)$. However, the second sum diverges, because
\begin{equation*}
\sum_{k=r}^{\infty}\frac{A}{k\ln k} \geq \binteg{k}{r}{\infty}\frac{A}{k\ln k}
= \binteg{u}{\ln r}{\infty}\frac{A}{u} = \infty .
\end{equation*}
Hence, we have $S[\hat{\rho}'] = \infty$.

In the case $\hat{\rho}$ has $m$ non-zero weights, i.e.\ $w_k = 0$ for $k > m$, we can set $w_1' = w_1, \dotsc, w_{m-1}' = w_{m-1}$ and $w'_m = \max(0,w_m - \epsilon/2)$, such that again $\norm{\hat{\rho} - \hat{\rho}'} < \epsilon$. By choosing the other weights as before~\eqref{eq:weightChoice}, we can now repeat the same argument.
\end{proof}

Since the bosonic entropy can jump to $+\infty$ for an arbitrary small variation in the density-matrix operator, it cannot be continuous. However, we actually do not need to use any continuity property to show when the infimum in~\eqref{eq:OmegaVpart} and~\eqref{eq:OmegaVmin} can be replaced by a minimum. Instead, we will use a different route via the relative entropy. To this end, consider Klein's inequality.

\begin{thm}[Klein's inequality]
\label{thm:Klein}
Let $f$ be a convex (concave) function and $A,B \in \TraceClass$. Then
\begin{equation*}
\Trace\{f(B) - f(A)\} \substack{\geq \\ (\leq)} \Trace\{(B - A)f'(A)\} \, .
\end{equation*}
\end{thm}

A proof for Klein's inequality is given in Appendix~\ref{ap:KleinProof}.
If we take $f(x) = -x\ln(x)$, Klein's inequality yields
\begin{equation}\label{eq:Klein}
\Trace\bigl\{A\bigl(\ln(A) -\ln(B)\bigr)\bigr\} \geq \Trace\{A - B\} \, .
\end{equation}
For density-matrix operators, the right hand side vanishes, so the expression on the left is positive. The left hand side is called the relative entropy, which is defined for density-matrix operators $\hat{\rho}$ and $\hat{\sigma}$ as
\begin{equation}\label{def:relEntropy}
S[\hat{\rho} | \hat{\sigma}]
\isDefinedAs \Trace\bigl\{\hat{\rho}\bigl(\ln(\hat{\rho}) - \ln(\hat{\sigma})\bigr)\bigr\}  \, .
\end{equation}
A necessary condition for $S[\hat{\rho} | \hat{\sigma}] < \infty$ is that $\ker(\hat{\sigma}) \subseteq \ker(\hat{\rho})$ \citep{Lindblad1973}.

The relative entropy is particularly useful in our investigation of the grand potential. If we take the Gibbs state
$\hat{\rho}_v \isDefinedAs \e^{-\beta \hat{H}_v} / \Trace\bigl\{\e^{-\beta \hat{H}_v}\bigr\}$
for $\hat{\sigma}$ (thus $\ker(\hat{\sigma}) = \ker(\hat{\rho}_v) = \emptyset$), we recover the grand potential
\begin{equation}
S[\hat{\rho}| \hat{\rho}_v]
= \beta\bigl(\Omega_v[\hat{\rho}] -\Omega_v[\hat{\rho}_v]\bigr) \, .
\end{equation}
Because of Klein's inequality~\eqref{eq:Klein}, we have $S[\hat{\rho} | \hat{\sigma}] \geq 0$, which implies immediately that
\begin{equation}\label{eq:OmegaInequality}
\Omega_v[\hat{\rho}] \geq \Omega_v[\hat{\rho}_v] = -\beta^{-1}\ln\bigl(Z[v]\bigr) = \Omega[v] \, .
\end{equation}
Via Klein's inequality we therefore find that the Gibbs state $\hat{\rho}_v$ is a minimum.

\begin{thm}\label{thm:bosonicOmegaMinExist}
The grand potential has a minimum if and only if $Z[v] < \infty$. If $Z[v] < \infty$, then $\hat{\rho}_v$ is the unique minimiser.
\end{thm}

\begin{proof}
Obviously, when $Z[v] < \infty$, $\hat{\rho}_v \in \NdensMat \subset \TraceClass$ and~\eqref{eq:OmegaInequality} shows that it yields a minimum. Since by corollary~\ref{cor:OmegaStrictConvex}, $\Omega_v[\hat{\rho}]$ is strictly convex, the minimum is unique (theorem~\ref{thm:unimodal}).

If $Z[v] = \infty$, then we can construct a sequence of density-matrix operators \(\hat{\rho}_n = \hat{P}_{\leq n}\e^{-\beta\hat{H}} / Z_n \in \NdensMat\), where \(Z_n = \Trace\bigl\{\hat{P}_{\leq n}\e^{-\beta\hat{H}}\bigr\}\) and $\hat{P}_{\leq n}$ are finite-dimensional projectors on the part of the Fock space with $n$ or less particles. In the limit $n \to \infty$ we have $Z_n \nearrow Z[v] = \infty$. Hence, we make $\Omega[\hat{\rho}_n]$ arbitrarily low by taking $n$ large enough, i.e.\ the grand potential is unbounded from below.
\end{proof}

Theorem~\ref{thm:bosonicOmegaMinExist} is the first main result of this section on the bosonic grand potential. It tells us that bosonic 1RDM functional theory at finite temperature is only sensible if we choose a potential (Hamiltonian) such that $Z[v] < \infty$.

\subsection{Boundedness of the partition function}
\label{sec:bosonpartition}

In this part we will show that if the Hamiltonian is bounded from below and has a maximum order of interaction $\nmax < \infty$ that its partition function is finite. For this purpose we will first consider the following lemma.

\begin{lem}\label{lem:domHNn}
If the highest-order interaction in the Hamiltonian $\hat{H}$ is of order $0 < \nmax < \infty$, i.e.\ the maximum number of creation\slash{}annihilation operators is $2\nmax$, then $\domain(\hat{H}) \supseteq \domain(\hat{N}^{\nmax})$.
\end{lem}

\begin{proof}
First observe that from the inequality for operators, $0 \leq \abs{\hat{A} - \hat{B}}^2$ it follows that on their common domain
\begin{equation*}
\hat{A}^{\dagger}\hat{B} + \hat{B}^{\dagger}\hat{A} \leq \abs{\hat{A}}^2 + \abs{\hat{B}}^2 \, .
\end{equation*}
By using different strings of creation and annihilation operators for $\hat{A}$ and $\hat{B}$ (see Tab.~\ref{tab:inequalities}), one finds that all terms in the Hamiltonian can be bounded by
\begin{align*}
\hat{H} &\leq C +
\sum_{i_1=1}^{N_b}C_{i_1}\abs{\anni{a}_{i_1}}^{\mathrlap{2}} + \dotsb + 
\sum_{\crampedclap{i_1,\dotsc,i_{\nmax}=1}}^{N_b}
C_{i_1\dotsc i_{\nmax}}\abs{\anni{a}_{i_1}\dotsb\anni{a}_{i_{\nmax}}}^2 \\
&\leq M^{(0)} + M^{(1)}\hat{N} + M^{(2)}\hat{N}^2 + \dotsb + M^{(\nmax)}\hat{N}^{\nmax} \, ,
\end{align*}
where all $\abs{M^{(i)}} < \infty$, because the number of parameters in the Hamiltonian is finite. Repeating the same for $-\hat{H}$ leads to a similar lower bound on $\hat{H}$. Therefore, if $\hat{N}^{\nmax}$ has a finite value then also $\hat{H}$ has and thus $\domain(\hat{H}) \supseteq \domain(\hat{N}^{\nmax})$
\end{proof}

In the following it becomes advantageous to define a splitting of the Hamiltonian. We do so by first defining $\hat{P}_{\leq n}$ as the projection operator which projects on all states with maximum $n$ particles. The Hamiltonian is then split as $\hat{H} = \hat{H}\hat{P}_{\leq n} + \hat{H}(1 - \hat{P}_{\leq n})$. The first part is bounded, so $\domain(\hat{H}\hat{P}_{\leq n}) = \Fock_+$ for any finite $n$. By choosing $n$ large enough, the $\nmax^{\text{th}}$-order interaction becomes the dominant part. Hence, there exist constants $\abs{K_l},\abs{K_h} < \infty$ such that $K_l\hat{N}^{\nmax}(1 - \hat{P}_{\leq n}) \leq \hat{H}(1 - \hat{P}_{\leq n}) \leq K_h\hat{N}^{\nmax}(1 - \hat{P}_{\leq n})$ for large enough $n$.

\begin{table}[b]\centering
\caption{Creation\slash{}annihilation operator inequalities.}
\label{tab:inequalities}
\begin{tabular}{c c c}
\toprule
$\hat{A}$			&$\hat{B}$	&inequality \\
\otoprule
$\alpha\anni{a}_i$	&$1$
	&$\alpha^*\crea{a}_i  + \alpha\,\anni{a}_i
	 \leq \abs{\alpha}^2\crea{a}_i\anni{a}_i + 1$ \\
$\alpha\anni{a}_i$	&$\anni{a}_k$
	&$\alpha^*\crea{a}_i\anni{a}_k + \alpha\,\crea{a}_k\anni{a}_i
	\leq \abs{\alpha}^2\crea{a}_i\anni{a}_i + \crea{a}_k\anni{a}_k$ \\
$\alpha\anni{a}_i$	&$\crea{a}_k$
	&$\alpha^*\crea{a}_i\crea{a}_k + \alpha\,\anni{a}_k\anni{a}_i
	\leq \abs{\alpha}^2\crea{a}_i\anni{a}_i + \anni{a}_k\crea{a}_k$ \\
$\alpha\anni{a}_i\anni{a}_j$	&$1$
	&$\alpha^*\crea{a}_j\crea{a}_i + \alpha\,\anni{a}_i\anni{a}_j
	\leq \abs{\alpha}^2\crea{a}_j\crea{a}_i\anni{a}_i\anni{a}_j + 1$ \\
$\alpha\crea{a}_i\anni{a}_j$	&$1$
	&$\alpha^*\crea{a}_j\anni{a}_i + \alpha\,\crea{a}_i\anni{a}_j
	\leq \abs{\alpha}^2\crea{a}_j\anni{a}_i\crea{a}_i\anni{a}_j + 1$ \\
$\alpha\anni{a}_i\anni{a}_j$	&$\anni{a}_k$
	&$\alpha^*\crea{a}_j\crea{a}_i\anni{a}_k + \alpha\,\crea{a}_k\anni{a}_i\anni{a}_j
	\leq \abs{\alpha}^2\crea{a}_j\crea{a}_i\anni{a}_i\anni{a}_j + \crea{a}_k\anni{a}_k$ \\
$\alpha\anni{a}_i\anni{a}_j$	&$\crea{a}_k$
	&$\alpha^*\crea{a}_j\crea{a}_i\crea{a}_k + \alpha\,\anni{a}_k\anni{a}_i\anni{a}_j
	 \leq \abs{\alpha}^2\crea{a}_j\crea{a}_i\anni{a}_i\anni{a}_j + \anni{a}_k\crea{a}_k$ \\
\vdots	&\vdots	&{} \vdots \\
\bottomrule
\end{tabular}
\end{table}

\begin{thm}\label{thm:eqDomHNn}
If the highest-order interaction in the Hamiltonian is positive definite, i.e.\ if there exists a $K_l > 0$, then $\domain(\hat{H}) = \domain(\hat{N}^{\nmax})$.
\end{thm}

\begin{proof}
As there exists a constant $K_l > 0$ such that $K_l\hat{N}^{\nmax} (1 - \hat{P}_{\leq n}) \leq \hat{H}(1 - \hat{P}_{\leq n})$, we have for any $\ket{\Psi} \in \Fock$
\begin{equation*}
\brakket{\Psi}{\hat{N}^{\nmax}(1 - \hat{P}_{\leq n})}{\Psi}
\leq K_l^{-1} \brakket{\Psi}{\hat{H}(1 - \hat{P}_{\leq n})}{\Psi} \, ,
\end{equation*}
so $\domain(\hat{N}^{\nmax}) \supseteq \domain(\hat{H})$. Combined with Lemma~\ref{lem:domHNn} with have $\domain(\hat{H}) = \domain(\hat{N}^{\nmax})$.
\end{proof}

\begin{thm}\label{thm:boundedfrombelow}
If the highest-order interaction in the Hamiltonian is positive semidefinite, i.e.\ if there exists a $K_l \geq 0$, then $\hat{H}$ is bounded from below on $\domain(\hat{N}^{\nmax})$.
\end{thm}

\begin{proof}
If the highest-order interaction is positive semidefinite, we have for large enough $n$ that $K_l \geq 0$ (defined after the proof of lemma~\ref{lem:domHNn}). Hence, for large enough $n$ we have
\begin{equation*}
\mathcal{E}_0 =
\inf_{\Psi}\frac{\brakket{\Psi}{\hat{H}}{\Psi}}{\braket{\Psi}{\Psi}}
\geq \inf_{\Psi}\frac{\brakket{\Psi}{\hat{H}\hat{P}_{\leq n}}{\Psi}}{\braket{\Psi}{\Psi}} +
\inf_{\Psi}\frac{\brakket{\Psi}{\hat{H}(1-\hat{P}_{\leq n})}{\Psi}}{\braket{\Psi}{\Psi}}
\geq E_n + K_l \geq E_n \, ,
\end{equation*}
where $E_n$ is the lowest eigenvalue of $\hat{H}\hat{P}_{\leq n}$. As $\hat{H}\hat{P}_{\leq n}$ only acts within a finite-dimensional part of the Fock space, we have $E_n > -\infty$ and hence, $\mathcal{E}_0 > - \infty$.
\end{proof}

For the partition function to be bounded, we do not only need the Hamiltonian to be bounded from below, but also the absence of accumulation points. So we will consider Hamiltonians with a strictly positive definite highest-order interaction.

\begin{prop}
A Hamiltonian with a highest-order interaction $0 < \nmax < \infty$ has no accumulation point in its spectrum if the highest-order interaction is strictly positive definite, i.e.\ if there exists a $K_l > 0$.
\end{prop}

\begin{proof}
As the highest-order interaction is strictly positive definite, we have $K_l > 0$ for large enough $n$. This implies that the energy difference for $n \to \infty$ behaves asymptotically as $\sim K_ln^{\nmax}$ which does not converge and has no converging subsequence.
\end{proof}

That a strictly positive definite highest-order interaction is sufficient to have a bounded partition function is formulated in the following theorem.

\begin{thm}\label{thm:finiteZ}
The partition function, $Z[v]$ is finite if the Hamiltonian $\hat{H}_v$ has a maximum order of interaction $0 < \nmax < \infty$ which is strictly positive definite, i.e.\ if there exists a $K_l >0$.
\end{thm}

\begin{proof}
First consider a non-interacting Hamiltonian, $\hat{H}_v = \sum_{ij}v_{ij}\crea{a}_i\anni{a}_j$, so $\nmax = 1$. Without loss of generality, we can assume that we work in the one-particle basis which diagonalises $v$ and can identify $K_l$ with the lowest eigenvalue of $v$. Now we can put the following bound on the bosonic partition function
\begin{equation*}
Z[v] = \Trace\bigl\{\e^{-\beta\hat{H}_v}\bigr\}
\leq \Trace\bigl\{\e^{-\beta K_l\hat{N}}\bigr\}
= \sum_{n=0}^{\infty}\binom{n+N_b-1}{n}\e^{-\beta K_l n} = \frac{1}{(1 - \e^{-\beta K_l})^{N_b}} \, ,
\end{equation*}
where the second equality follows from counting the number of states in the $n$-particle sector and
last equality follows from working out the $N_b$-th order derivative of the geometric series. Hence, if the potential is strictly positive definite, $K_l > 0$, the bosonic partition function is finite for all temperatures.

This result for the non-interacting case also applies to non-interacting Hamiltonians with additional non-conserving terms. The source term only adds a shift to the creation and annihilation operators and leaves the spectrum invariant (see the first part of Appendix~\ref{ap:non-interatingSolution}). In contrast, the pairing field affects the positive definiteness of the interaction (see the last part of Appendix~\ref{ap:non-interatingSolution}) and care should be taken not to spoil the positive definiteness of the highest-order interaction. One could use the third inequality in Table~\ref{tab:inequalities} to put some sufficient bounds on the variation in the potential $v$, though they can typically be weakened.

Now consider a Hamiltonian with some proper interaction, i.e.\ $\nmax > 1$. In that case the partition function can be bounded as
\begin{align*}
Z[v] &\leq \Trace\bigl\{\e^{-\beta K_l\hat{N}^{\nmax}}\bigr\}
= C + \sum_{n=0}^{\infty}\binom{n+N_b-1}{n}\e^{-\beta K_l n^{\nmax}} \\
&\leq C + \sum_{n=0}^{\infty}\binom{n+N_b-1}{n}\e^{-\beta K_l n}
= C + \frac{1}{(1 - \e^{-\beta K_l})^{N_b}} \, ,
\end{align*}
where $C$ is some finite positive constant and $K_l$ is the constant defined after the proof of Lemma~\ref{lem:domHNn}. Hence, the partition function is finite if $K_l > 0$, i.e.\ if the highest-order interaction is strictly positive definite. Note that $K_l$ does not depend on the one-body potential $v$, so if the partition function is finite for an interacting system ($\nmax > 1$), it is finite for any one-body potential $v$.

Again, this result also applies to general Hamiltonians that mix the number of particles. We only need to take care that the non-conserving terms of order $2\nmax$ do not spoil the positive definiteness of the highest-order interaction.
\end{proof}

Theorem~\ref{thm:finiteZ} is the second important result of this section. As we typically model a physical system with some finite order of positive interaction between the particles, we have $Z[v] < \infty$ without much difficulty. Thus $\FiniteNpot_{+}$ is either the full space $\HermitanMat(N_b)$ for the interacting situation, or in the non-interacting case as discussed in Sec.~\ref{sec:potentials}, we have $\FiniteNpot_{+}^{\text{nonint}} = \set*{v \in \HermitanMat(N_b) }{h^{(1)} + v > 0} $.

We can proceed somewhat further along these lines to show that also most expectation values will be finite.

\begin{thm}\label{thm:finiteNk}
The expectation value of any finite power of the number operator is finite, i.e.\ $\av{N^k}_v = \Trace\{\hat{N}^k\e^{-\beta\hat{H}_v}\} / Z[v] < \infty$ for $k < \infty$, if the Hamiltonian $\hat{H}_v$ has a maximum order of the interaction which is strictly positive definite, i.e.\ if there exists a $K_l > 0$.
\end{thm}

\begin{proof}
Basically we need to repeat the previous proof. So first for the non-interacting case, we have
\begin{align*}
Z[v]\,\av{\hat{N}^k}_v &= \Trace\bigl\{\hat{N}^k\e^{-\beta\hat{H}_v}\bigr\}
\leq \Trace\bigl\{\hat{N}^k\e^{-\beta K_l\hat{N}}\bigr\}
= \sum_{n=0}^{\infty}\binom{n+N_b-1}{n}n^k\e^{-\beta K_l n} \\
&= \du_{-\beta K_l}^k\frac{1}{(1 - \e^{-\beta K_l})^{N_b}}
= \frac{\bigl(N_b\e^{-\beta K_l}\bigr)_k}{(1 - \e^{-\beta K_l})^{N_b+k}} < \infty \, ,
\end{align*}
where $(x)_k = \Gamma(x + k)/\Gamma(x)$ denotes the Pochhammer symbol. Hence, if a non-interacting Hamiltonian is strictly positive definite, $K_l > 0$, also $\av{\hat{N}^k}_v < \infty$ for any finite $k$. For an interacting system ($\nmax > 1$) we have analogous to the previous proof
\begin{equation*}
Z[v]\,\av{\hat{N}^k}_v \leq \Trace\bigl\{\hat{N}^k\e^{-\beta K_l\hat{N}^{\nmax}}\bigr\}
\leq C_k + \frac{\bigl(N_b\e^{-\beta K_l}\bigr)_k}{(1 - \e^{-\beta K_l})^{N_b+k}} < \infty \, ,
\end{equation*}
where $C_k$ is some finite positive constant and $K_l > 0$ is again the constant defined after the proof of Lemma~\ref{lem:domHNn}.
\end{proof}

Theorem~\ref{thm:finiteNk} effectively means that any reasonable expectation value is finite. In particular, the 1RDM is finite, since $\gamma_{ij} \leq \Trace\{\gamma\} = \av[\big]{\hat{N}} < \infty$, so any $\gamma[v] \in \NoneMat$.
The same argument applies to any higher order reduced density-matrix
\begin{equation}
\Gamma^{(n)}_{i_1\dotsc i_n,j_n\dotsc j_1} \leq \Trace\bigl\{\Gamma^{(n)}\bigr\}
= \av[\big]{\hat{N}\bigl(\hat{N} - 1\bigr)\dotsm\bigl(\hat{N} - n + 1\bigr)}
\leq \av[\big]{\hat{N}^n} < \infty \, .
\end{equation}
Additionally, due to the bounds on the energy used in the proof of lemma~\ref{lem:domHNn} and the fact that $\Omega[\hat{\rho}_v] = -\beta^{-1}\ln\bigl(Z[v]\bigr) < \infty$, we immediately have the following corollary.

\begin{cor}
The entropy $S[v]$ and energy $E[v]$ are finite, if the Hamiltonian has a maximum order of the interaction which is strictly positive definite.
\end{cor}

\subsection{Existence of minimum in the bosonic universal function}
\label{sec:bosonicfunctional}

Now we will turn our attention to the question whether the infimum is attained in the bosonic universal function~\eqref{def:universalFunction}. In the fermionic case we used the extreme value theorem~\ref{thm:extremeValue}, but it is not applicable for two reasons.
\begin{enumerate*}[label=\arabic*)]
\item As we have already seen, the bosonic grand potential is not continuous.
\item The set of bosonic $\hat{\rho}$ is not compact in the trace norm, as the unit ball in an infinite-dimensional space is not compact in the usual norm, i.e.\ the trace norm in our case.
\end{enumerate*}

To resolve these issues, we follow the same strategy as used by Lieb to show the existence of the minimum in the universal functional in DFT \citep{Lieb1983} and repeated by others \citep{Eschrig1996, Eschrig2003, Leeuwen2003, Lammert2006a, Lammert2006b, Lammert2010} and in the first attempt for a rigorous finite temperature 1RDM functional theory by Baldsiefen et al. \citep{PhD-Baldsiefen2012, BaldsiefenCangiGross2015}.
We first focus on the latter problem, i.e.\ that the space $\NdensMat_+$ is not compact and neither \( \set{\hat{\rho} \in \NdensMat_+}{\hat{\rho} \to \gamma } \). To resolve this issue, we will introduce a weaker norm under which the unit ball will be compact. For this we first need to properly introduce the notion of the dual space, as was already briefly exemplified in Sec.~\ref{sec:convex}.

\begin{defn}[Dual space]
\label{def:duals}
The space of all continuous linear functionals $f \colon X \to \Complex$ is called the dual space of $X$ and denoted as $X^*$. These functionals typically denoted with a bracket in quantum mechanics as $f(x) = \braket{f}{x}$.
\end{defn}

\begin{defn}[Weak-* convergence]
We say that a sequence $f_n \in X^*$ converges in the weak-* topology if for all $x \in X$, $f_n(x) \to f(x)$. We will denote this as $f_n \weakstarly f$
\end{defn}

Weak-* convergence is useful for our purposes, since we know that the space of trace-class operators is the dual of the space of compact operators, $\CompactOperators^* = \TraceClass$, so the \mbox{weak-*} topology is well defined for the trace-class operators (and density-matrix operators). A trace-class operator $T \in \TraceClass$ is identified with a continuous linear functional on the compact operators as $\Trace\{T\,\cdot\}$.
Additionally, the weak-* topology is weak enough to make the unit ball compact as stated by the following famous theorem.

\begin{thm}[Banach--Alaoglu]\label{thm:BanachAlaoglu}
Let $X$ be a normed space and $X^*$ its dual. The closed unit ball in $X^*$ is compact with respect to the weak-* topology.
\end{thm}

Now we have a new topology such that the closed unit ball is compact, so that every sequence has a convergent subsequence in the closed unit ball (see definition~\ref{def:compactSet}). We then need to turn our attention to the lack of continuity. It turns out that it is sufficient to use a weaker property: lower semi-continuity.

\begin{defn}[Lower/Upper semi-continuity]
Consider a topological space $X$ and a function $f : X \to \Reals \cup \{-\infty,+\infty\}$. A function $f$ is called lower semi-continuous at $x_0$ if for every $\epsilon > 0$ there exists a neighborhood $U$ of $x_0$ such that $f(x) \geq f(x_0) - \epsilon$ for all $x \in U$.

Upper semi-continuity at $f(x_0)$ is simply lower semi-continuity of $-f(x_0)$.

If a function $f$ is lower (upper) semi-continuous at every point in its domain, $f$ is called lower (upper) semi-continuous.
\end{defn}

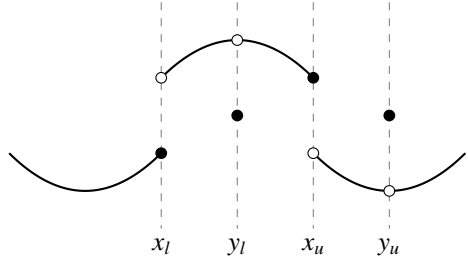
\begin{figure}\centering
\begin{tikzpicture}
  \draw[thick, smooth, domain=-3:-1] plot (\x,{0.5*(\x + 1)*(\x + 3)});
  \draw[thick, smooth, domain=-1:1] plot (\x,{-0.5*(\x*\x - 1) + 1});
  \draw[thick, smooth, domain=1:3] plot (\x,{0.5*(\x - 1)*(\x - 3)});
  \draw[dashed, gray] (-1,2) -- (-1,-1) node[black, anchor=north] {$x_l$};
  \filldraw (-1,0) circle (2pt);
  \filldraw[fill=white] (-1,1) circle (2pt);
  \draw[dashed, gray] (0,2) -- (0,-1) node[black, anchor=north] {$y_l$};
  \filldraw (0,0.5) circle (2pt);
  \filldraw[fill=white] (0,1.5) circle (2pt);
  \draw[dashed, gray] (1,2) -- (1,-1) node[black, anchor=north] {$x_u$};
  \filldraw (1,1) circle (2pt);
  \filldraw[fill=white] (1,0) circle (2pt);
  \draw[dashed, gray] (2,2) -- (2,-1) node[black, anchor=north] {$y_u$};
  \filldraw (2,0.5) circle (2pt);
  \filldraw[fill=white] (2,-0.5) circle (2pt);
\end{tikzpicture}
\caption{Plot of a function which is lower semicontinuous at $x_l$ and $y_l$. It is upper semi-continuous at $x_u$ and $y_u$. The solid dots indicate where the function attains it value and the open dots where it does not.}
\label{fig:semicontinuity}
\end{figure}

The concepts of lower and upper semi-continuity are illustrated in Fig.~\ref{fig:semicontinuity}.
Obviously, a function is continuous if and only if it is lower and upper semi-continuous. An other convenient property is the following.

\begin{prop}\label{prop:supCombi}
The supremum over a collection of lower semi-continuous functions, $f(x) = \sup_if_i(x)$, is also lower semi-continuous.
\end{prop}

\begin{proof}
By definition of the supremum we have in a neighborhood around $x_0$
\begin{equation*}
f(x) = \sup_if_i(x) \geq \sup_if_i(x_0) - \epsilon = f(x_0) - \epsilon \, . \qedhere
\end{equation*}
\end{proof}

The relevance of lower (upper) semi-continuity in this context is that the extreme value theorem is readily adapted to these weaker forms of continuity.

\begin{thm}\label{thm:extremeValueSemi}
Let $f \colon X \to \Reals$ be a lower (upper) semi-continuous function and $M \subseteq X$ a compact set. Then $f$ must attain a minimum (maximum) at least once.
\end{thm}

\begin{proof}
As we only used lower (upper) semi-continuity in the existence of a lower (upper) bound in the proof of the extreme value theorem, it is basically the same.
\end{proof}

The next step is to demonstrate lower semi-continuity of the relevant functionals. As we would like to use compactness of the unit ball, we should show lower semi-continuity with respect to the weak-* topology. To differentiate this from lower semi-continuity with respect to the usual norm, we will call this weak-* lower semi-continuity.

\begin{thm}\label{thm:bosSlsc}
The entropy is weak-* lower semi-continuous.
\end{thm}

\begin{proof}\citep{Wehrl1978}
To show this, we will use the fact that any finite-rank operator is compact, so in particular every finite-dimensional projection operator $\hat{P}$. Therefore, weak-* convergence of $\hat{\rho}_n$ to $\hat{\rho}$ means that \( \Trace\bigl\{\hat{P}\hat{\rho}_n\bigr\} \to \Trace\bigl\{\hat{P}\hat{\rho}\bigr\} \) for any $\hat{P}$. Since the function $s(x) = -x\ln(x)$ is continuous, we have $\Trace\bigl\{\hat{P}\bigl(s(\hat{\rho}_n) - s(\hat{\rho})\bigr)\bigr\} \to 0$. Further, $\Trace\{\hat{P}\hat{A}\} \leq \Trace\{\hat{A}\}$ for $\hat{A} \geq 0$, so $ \Trace\{\hat{A}\} = \sup_{\hat{P}}\Trace\{\hat{P}\hat{A}\}$. We therefore have that
\begin{equation*}
S[\hat{\rho}] = \sup_{\hat{P}}\Trace\{\hat{P}s(\hat{\rho})\}
\end{equation*}
is lower semi-continuous by proposition~\ref{prop:supCombi}.
\end{proof}

\begin{thm}
The energy is weak-* lower semi-continuous.
\end{thm}

\begin{proof}
The proof is the same as for the entropy.
\end{proof}

The grand potential is a linear combination of a weak-* lower semi-continuous functional (the energy) and a upper semi-continuous functional (minus the entropy), so we can not say anything directly. Again, the relative entropy~\eqref{def:relEntropy} comes to our aid. There is an alternative expression for the relative entropy, which avoids the product of non-commuting operators~\cite[Lemma 4]{Lindblad1973}
\begin{equation}\label{eq:altDefRelS}
S[\hat{\rho} | \hat{\sigma}]
= \sup_{\mathclap{0 < \lambda < 1}} \lambda^{-1}S_{\lambda}[\hat{\rho} | \hat{\sigma}] \, ,
\end{equation}
where $S_{\lambda}[\hat{\rho} | \hat{\sigma}] \isDefinedAs S[\lambda\hat{\rho} + (1 - \lambda)\hat{\sigma}] - \lambda S[\hat{\rho}] - (1-\lambda)S[\hat{\sigma}]$.
The proof for the equality in~\eqref{eq:altDefRelS} has been deferred to Appendix~\ref{ap:ProofRelSasLim}.
With this alternative expression, we can readily establish the following theorem.

\begin{thm}\label{thm:relSlsc}
The relative entropy is weak-* lower semi-continuous.
\end{thm}

\begin{proof}
Using the alternative expression for the relative entropy~\eqref{eq:altDefRelS}, we can repeat the proof for weak-* lower semi-continuity of the entropy (theorem~\ref{thm:bosSlsc}). The only change is that the supremum is now also taken over $\lambda$.
\end{proof}

\begin{cor}\label{cor:OmegaWeakStarlsc}
The grand potential $\Omega_v[\hat{\rho}]$ is weak-* lower semi-continuous if and only if $Z[v] < \infty$.
\end{cor}

Now we are almost done. Since the grand potential is weak-* lower semi-continuous and the infimum in $F[\gamma] = \inf_{\hat{\rho} \to \gamma}\Omega_0[\hat{\rho}]$ is taken over the shell of a unit ball in $\TraceClass$, a minimising sequence $\hat{\rho}_n \weakstarly \hat{\rho}$ exists and converges within the unit ball, so $\Trace\{\hat{\rho}\} \leq 1$. If we can show that this implies that $\hat{\rho}_n \to \hat{\rho}$ strongly, i.e.\ $\Trace\bigl\{\hat{\rho}\bigr\} = 1$, we are done. We need this to be sure that we do not end up with a $\hat{\rho} \notin \NdensMat$ and it also ensures that the resulting density-matrix operator yields the requested 1RDM.

\begin{thm}\label{thm:FinfToMinBosons}
The minimum in the bosonic universal functional $F_+[\gamma]$ is achieved if $Z[0] < \infty$, so the infimum can be replaced by a minimum.
\end{thm}

\begin{proof}
As $\hat{\rho}_n \weakstarly \hat{\rho}$, we have $\hat{P}\hat{\rho}_n \to \hat{P}\hat{\rho}$ strongly (in the 1-norm) for any finite-dimensional projection $\hat{P}$. 
For any $0 < \epsilon \leq 1$ we can find a finite-dimensional projection operator such that
\begin{equation*}
\epsilon > \Trace\bigl\{\hat{\rho}_n\bigl(1 - \hat{P}\bigr)\bigr\} \qquad
\Rightarrow \qquad
\Trace\bigl\{\hat{\rho}_n\hat{P}\bigr\} \geq 1 - \epsilon \, .
\end{equation*}
Since $\hat{P}\hat{\rho}_n \to \hat{P}\hat{\rho}$, this implies that $\Trace\bigl\{\hat{\rho}\bigr\} \geq 1 - \epsilon$, and therefore $\Trace\{ \rho \} = 1$.
\end{proof}

\section{Discussion on the extension to an infinite-dimensional one-particle space}
\label{sec:infiniteDiscus}

Within the setting of a finite one-particle basis to generate the Fock space, we have provided a rigorous framework for 1RDM functional theory at finite temperatures. The main advantage is that the `interface' quantities of the theory ($v$ and $\gamma$) are always finite-dimensional, i.e.\ the potentials and 1RDMs are finite-dimensional matrices. As all relevant functionals are convex, this immediately implies convenient properties such as Lipschitz continuity, directional differentiability and the existence of a subgradient. The quantity under the hood ($\hat{\rho}$) is still allowed to be infinite-dimensional, as we have no desire to prove any differentiability properties. We only need that the minimiser of $\Omega_v[\hat{\rho}]$ exists to establish the connection between subgradients and potentials.

The situation becomes much more involved if one desires to provide a rigorous framework for 1RDM functional theory at elevated temperatures allowing for an infinite-dimensional one-particle space. One could always argue that the infinite-dimensional one-particle case could in principle be obtained by taking the infinite-basis limit at the end of the calculation (provided it exists~\citep{thirring2013}). This is exactly the procedure followed in practice when simulating the quantum system on a computer, as we always only have a finite amount of memory at our disposal. By considering a sequence of calculations in increasing basis sets, one could extrapolate the results to the infinite-basis limit.

Nevertheless, it would be desirable to establish 1RDM functional theory for infinite (complete) basis sets, as this would guarantee that the infinite-basis limit exists (converges) and would also guarantee its uniqueness. For realistic finite systems with $\av{N} < \infty$, such a generalisation seems possible in principle. On the other hand, if one wants to investigate bulk properties of an extended system (infinite solid), one also lets the average particle number go to infinity such that $\av{N}/V$ remains constant, i.e.\ the thermodynamic limit. It is well known that the thermodynamic limit is often not unique, corresponding to different phases of the material~\citep{Emch1972}. Apart from non-uniqueness, if one would like to formulate a 1RDM functional theory for the thermodynamic-limit case directly, we would either need to employ a renormalised 1RDM that stays finite, e.g., $\bar{\gamma} = \gamma/V$ with $\trace\{\bar{\gamma}\} = \av{N}/V < \infty$, or change the mathematical setting and go beyond trace-class density matrices.

From these considerations it is clear that one should first deal with the extension to infinite one-particle basis sets and afterwards one could consider the thermodynamic limit. Let us therefore go into more detail about the complications one would need to deal with, if one intends to extend the present approach to the infinite-dimensional one-particle basis case.

\subsection{Characterisation of the proper set of potentials}
To prove the boundedness of the partition function for a Hamiltonian with a strictly positive definite interaction in the bosonic case, we used that each particle sector is finite-dimensional. In the infinite-dimensional case this is not true anymore, even for fermions. The situation is even somewhat more involved, as $Z[v] < \infty$ does not imply that $\av{N}_v < \infty$, so that the potential would yield a proper 1RDM with finite trace. From the discussion on the bosonic case in the finite-dimensional one-particle space in Sections~\ref{sec:potentials} and~\ref{sec:bosonpartition}, one expects that the set of potential $\av{N}_v < \infty$ also depends on the other terms in the Hamiltonian, e.g.\ repulsive interactions. As discussed in Section~\ref{sec:grandCanonIntro}, one should at least put the system in some box or other confining potential to avoid $Z = \infty$ already in the one-particle sector.

An alternative route might be to abandon the direct use of the partition function altogether via an algebraic approach to quantum field theory. The algebraic approach gives a generalisation of the Gibbs state which avoids the use of the partition function: the Kubo--Martin--Swinger (KMS) state \citep{HaagWinninkHugenholtz1967, Emch1972}. The KMS states are defined to be the density-matrix operators satisfying the KMS boundary conditions \citep{Kubo1957, MartinSchwinger1959}. Nevertheless, as there is no finite-valued grand potential to work with in general, we lack a global value which is being minimised.
One would therefore need a completely different route to construct a 1RDM functional theory, as all quantities should be formulated directly in terms of the KMS states.

\subsection{The domain of the universal functional}
In the finite-dimensional setting we could easily show that any ensemble $N$-representable 1RDMs, $\gamma \in \NoneMat$, yields a finite value for the universal functional, $F[\gamma] < \infty$. The argument relied on the fact that there always exists a compact density-matrix operator, $\hat{\rho}_{\gamma}$ which generates this 1RDM. As this density-matrix operator is compact, one immediately has $E[\hat{\rho}_{\gamma}] < \infty$, $S[\hat{\rho}_{\gamma}] < \infty$ and $\Omega_v[\hat{\rho}_{\gamma}] < \infty$. In the infinite-dimensional case we cannot use this argument anymore and $F\bigl[\gamma \in \NoneMat\bigr] < \infty$ is not expected to hold. This set is expected to depend crucially on the other terms in the Hamiltonian like in DFT, where finiteness of the kinetic-energy operator requires the density not only to be integrable, but also $\nabla\sqrt{n} \in L^2$~\cite[Thms~3.8 and~3.9]{Lieb1983}

\subsection{Lack of smootheness of the universal functional}
We loose all convenient properties implied by the convexity of the functionals in the infinite-dimensional case. So convexity of $F[\gamma]$ does not imply Lipschitz continuity, the existence of a directional derivative in any direction or the existence of a subgradient anymore.
Nevertheless, the Hahn--Banach theorem does guarantee the existence of a tangent functional~\cite[p.~259]{Lieb1983}. This means that for each $\gamma_0 \in \ToneMat$ there exists a linear functional $L_{\gamma_0}$ such that $F[\gamma_1] \geq F[\gamma_0] + L_{\gamma_0}[\gamma_1 - \gamma_0]$. For the tangent functional to be a subgradient, it also needs to be continuous~\cite[p.~1945]{Lammert2006a}. Analogous to DFT \citep{Lieb1983, EnglischEnglisch1984b}, one would expect that exactly at $v$-representable 1RDMs the tangent functional would be continuous, unique and of the form $-\Trace\{v\,\cdot\}$ where $v \mapsto \gamma$. So a unique subgradient would exist at those 1RDMs and equal $-v$. Unfortunately, the existence of a unique subgradient does not imply differentiability anymore. 
From an optimizational perspective, the existence of a subgradient is sufficient to formulate first-order optimality conditions.
To proof actual (Gâteaux) differentiability, one would probably need a more constructive approach as Lammert did in the $T=0$ DFT setting \citep{Lammert2006a} or use the Moreau--Yosida regularisation as proposed by Kvaal et al.~\citep{KvaalEkstromTeale2014}.

\section{Conclusion}
\label{sec:conclusion}

In this review we have provided a self-contained, rigorous formulation of 1RDM functional theory for a finite one-particle space and arbitrary number of particles at finite temperatures. 

For the fermionic case, as the Fock space $\Fock_{-}$~\eqref{def:FockSpace} is finite-dimensional, any hermitian Hamiltonian $\hat{H}$ is allowed and the universal functional $F_-[\gamma]$~\eqref{def:universalFunction} has a unique minimum. For any hermitian matrix $v \in \HermitanMat(N_b)$~\eqref{def:hermitianMatrix} defining a non-local external potential that is added to the universal part of the Hamiltonian, the grand potential $\Omega_-[v]$~\eqref{eq:OmegaVpart} has a unique minimum as well. The mapping $v \mapsto \gamma$ from $\HermitanMat(N_b)$ to $\VoneMat_{-} = \ToneMat_{-}$~\eqref{defs:ToneMats} is bijective and $F_-[\gamma]$ is differentiable in $\VoneMat_{-}$ such that $\partial F[\gamma]/ \partial \gamma = -v$ holds.

For the bosonic case, where the Fock space $\Fock_{+}$~\eqref{def:FockSpace} is infinite-dimensional, a difference between simple hermiticity and self-adjointness arises, partition functions $Z[v]$~\eqref{eq:equiRho} and other observables might become undefined and the existence of Gibbs states is no longer guaranteed. Therefore, we have to impose restrictions. First of all, only those Hamiltonians are allowed that have a highest-order number-conserving interaction~\eqref{eq:ncHam}  that is strictly positive or in the case of a non-interacting Hamiltonian the single-particle Hamiltonian has to be strictly positive. Non-conserving parts in the Hamiltonian can be allowed, provided that they do not effectively destroy the strict positivity of the highest-order interaction, cf.\ Table~\ref{tab:inequalities}. This is always guaranteed if the non-conserving part is of lower order in the annihilation and creation operators than the highest-order interaction.
Under these conditions, the universal functional $F_+[\gamma]$~\eqref{def:universalFunction} is guaranteed to have a unique minimum. For the grand potential $\Omega_+[v]$ we need to make a distinction between the interacting and non-interacting case.
In the interacting case, any hermitian matrix $v \in \HermitanMat(N_b)$~\eqref{def:hermitianMatrix}
yields to a unique minimum and the mapping $v \mapsto \gamma$ from $\HermitanMat(N_b)$ to $\VoneMat_{+}  = \ToneMat_{+}$~\eqref{defs:ToneMats} is bijective and $F_+[\gamma]$ is differentiable in $\VoneMat_{+}$ such that $\partial F_+[\gamma]/ \partial \gamma = -v$ holds. Similar results hold in the non-interacting case, however, only those $v \in \HermitanMat(N_b)$ are allowed that leave the total single-particle Hamiltonian strictly positive~\eqref{eq:FiniteNPot}. This set of allowed external non-local potentials we denoted by $\FiniteNpot_{+}^{\mathrm{nonint}}$.

In both the fermionic as well as the bosonic case we can therefore also set up a KS-type scheme where instead of solving for the interacting problem, a non-interacting auxiliary problem with a Hartree-exchange-correlation functional that contains interaction and entropic terms is solved. No non-interacting $v$-representability issues known from the $T=0$ situation arise, which makes the minimisation for approximate functionals straightforward. Also, the non-interacting universal functional $F_s^{\pm}[\gamma]$ is known explicitly in contrast to DFT. Therefore, 1RDM functional theory does not only put DFT for a finite basis set on rigorous grounds but is also an appealing alternative, at least in the grand-canonical setting investigated here.

As we have pointed out in detail, the case of an infinite-dimensional single-particle space or the case of $T=0$ is not so easy to handle in terms of the 1RDM. The mathematical issues that arise have so far hampered the development of more accurate and reliable approximate functionals for 1RDM functional theory. In this respect the current review poses a clear and comprehensive starting point to also investigate this other setting and learn more about the fundamental issues of $v$-representability and properties of the functionals involved. In the end, an explicit and at the same time simple characterisation of the involved spaces of non-local potentials and 1RDMs is a necessary prerequisite to make such a functional approach work also in these other cases.

In this work, we mainly gave a theoretical motivation for the finite temperature formalism. However, in many experiments, temperature effects play an important role, so the proposed theoretical framework provides an important extension of the zero temperature formalism. Important examples are metal-insulator transitions in transition metal oxides \citep{YooMaddoxKlepeis2005, RueffMattilaBadro2005, PattersonAracneJackson2004, MitaIzakiKobayashi2005, MitaSakaiIzaki2001, NoguchiKusabaFukuoka1996}, (high $T_c$) superconductors \citep{NagamatsuNakagawaMuranaka2001, BednorzMuller1986} and protein folding \citep{Anfinsen1972, TakaiNakamuraToki2008, NichollsSharpHonig1991}.
More extreme examples are rapid heating of solids via strong laser fields \citep{GavnholtRubioOlsen2009}, dynamo effect in giant planets \citep{RedmerMattssonNettelmann2011}, shock waves \citep{RootMagyarCarpenter2010, Militzer2006}, warm dense matter \citep{KietzmannRedmerDesjarlais2008} and hot plasmas \citep{Dharma-wardanaPerrot1982, PerrotDharma-wardana2000, Dharma-wardanaMurillo2008}.
Therefore the finite temperature framework is clearly of physical importance beyond our rather technical requirements.

While so far 1RDM functional theory was mainly concerned with fermionic problems, the extension to include the bosonic case is particularly timely. In the recent years investigations at the interface between quantum chemistry, solid-state physics and quantum optics uncovered interesting situation where a strong coupling between photons and matter, for instance, when molecules are put in a high-Q optical cavity or on metallic nano structures, dramatically change the chemical and physical properties of matter \citep{ebbesen2016,sukharev2017}. The emergent hybrid light-matter states, so called polaritons, can lead to, e.g., a change of chemical reactions \citep{hutchison2012} or lead to exciton-polariton condensates \citep{byrnes2014}. Since a detailed description of all constituents is necessary \citep{flick2017a, flick2017b}, first-principles approaches extended to fermion-boson systems become important \citep{ruggenthaler2011, tokatly2013, ruggenthaler2014, ruggenthaler2015, melo2015}. The current work only deals with single-component systems. Nevertheless, the presented formalism should be general enough to be extendable to multi-component systems. We do not see any reason a priori that any significant problem could arise. The current work therefore lays the foundation of an extension of 1RDM to matter-photon systems.

The theory presented here deals with the thermodynamic equilibrium in the grand canonical ensemble. As experiments typically involve some time-dependent perturbation of the system of which the response is measured, a time-dependent extension would be very desirable. Typically, time-dependent theories are developed at the level of the time-dependent Schrödinger equation, e.g., time-dependent DFT \citep{RungeGross1984, Leeuwen1999, Leeuwen2001, RuggenthalerLeeuwen2011,FarzanehpourTokatly2012} and time-dependent current DFT \citep{GhoshDhara1988, Vignale2004, Tokatly2011a}. Naively, one could hope that the Runge--Gross proof \citep{RungeGross1984} would be generalisable, but as the 1RDM does not commute with the Coulomb interaction, this strategy does not work. However, the proof for invertibility in the linear response regime by Van Leeuwen \citep{Leeuwen2001} is actually generalisable. This strategy works for both degenerate ground states \citep{Giesbertz2015} as well as for statistical ensembles \citep{Giesbertz2016} as initial states. Nevertheless, a proof going beyond the perturbative regime would be preferred. We have actually tried to generalise the global fixed-point proof by Ruggenthaler and Van Leeuwen \citep{RuggenthalerLeeuwen2011}, but one quickly realises that one needs to deal with the a time-dependent version of the pure-state $N$-representability conditions, which are already for the ground state very intricate \citep{Klyachko2006, AltunbulakKlyachko2008}. A possible way to remedy these issues is to couple the system to a bath which allows one to work with ensembles rather than pure states. A rudimentary version has already been proposed \citep{GiesbertzGritsenkoBaerends2010a}, but more investigations are needed. Although a rigorous foundation of time-dependent 1RDM functional theory is lacking, this has not deterred scientists to develop approximations \citep{GiesbertzBaerendsGritsenko2008, GiesbertzPernalGritsenko2009, AppelGross2010, GiesbertzGritsenkoBaerends2010b, GiesbertzGritsenkoBaerends2012b, MeerGritsenkoGiesbertz2013, GiesbertzGritsenkoBaerends2014, MeerGritsenkoBaerends2014, BricsBauer2013, RappBricsBauer2014, BricsRappBauer2014, BricsRappBauer2016}. As time-dependent 1RDM functional theory can be seen as the first tier of the Bogoliubov--Born--Green--Kirkwood--Yvon (BBGKY) hierarchy \citep{Yvon1935, Bogoliubov1946b, Bogoliubov1946a, Kirkwood1946, BornGreen1946, BogoliubovGurov1947, Kirkwood1947}, one often opts for treating the 2RDM explicitly \citep{SchmittReinhardToepffer1990, GheregaKriegReinhard1993, Schafer-BungNest2008, AkbariMohammad-HashemiNieminen2012, LacknerBrezinovaSato2015a, LacknerBrezinovaSato2017, KronkeSchmelcher2018}.

\section*{Funding}
This work was supported by the Academy of Finland (Grant No.~127739), a VENI grant by the Netherlands Foundation for Research NWO (722.012.013) and European Research Council under H2020/ERC Consolidator Grant ‘‘corr-DFT’’ (Grant No.~648932).

\appendix
\renewcommand*{\thesection}{\Alph{section}} 

\section{Convex and concave functions}
\label{ap:convex}
In this part of the appendix we give some additional details and a more precise treatment of convex functions. Most of these proofs have been taken from \citep{convexIntro}.

\subsection{Proof of theorem~\ref{thm:unimodal}}
\label{ap:unimodal}

\begin{proof}
We should proof that $f(y) \geq f(x^*)$ for all $y \in M$. If $f(y) = +\infty$, there is nothing to proof, so assume $y \in \domain(f)$. Since $x^*$ is a local minimiser we have by definition $x^* \in \domain(f)$. Because $f$ is convex, we have for all $t \in (0,1)$ and $x_t = t y + (1-t) x^*$
\begin{equation*}
f(x_t) - f(x^*) \leq \bigl(f(y) - f(x^*)\bigr) \, .
\end{equation*}
Because $x^*$ is a local minimiser, the left-hand side is nonnegative for small enough $t > 0$, so the right-hand side needs to be nonnegative for any $y \in M \cap \domain(f)$.

Now we proof that if the minimiser exists that it is unique if $f$ is strictly convex. We do this by reductio ad absurdum. Suppose that two distinct minimisers exist: $x^* \neq x^{\star}$. Then from strict convexity we have
\begin{equation*}
f\biggl(\half x^* + \half x^{\star}\biggr) < \half\biggl(f(x^*) + f(x^{\star})\biggr) = \min_{x \in M}f(x) \, .
\end{equation*}
Thus the point between $x^*$ and $x^{\star}$ would yield a lower value than the two minima at $x^*$ and $x^{\star}$. This is clearly a contradiction, so our initial assumption that there can be multiple minima is incorrect.
\end{proof}

\subsection{Proof of theorem~\ref{thm:convexLocLip}}
\label{ap:convexLocLip}

For a more rigorous understanding of theorem~\ref{thm:convexLocLip} we need a few definitions. Though we only need them in a finite-dimensional setting for the proof of theorem~\ref{thm:convexLocLip}, we discuss them more generally, as these definitions are needed in an infinite-dimensional setting in the bosonic case. In particular, the notion of a compact set is also needed in the infinite-dimensional setting.

\begin{defn}[Compact set]\label{def:compactSet}
Let $X$ be a normed space and $A \subseteq X$. Then the following are equivalent
\begin{itemize}
\item $A$ is compact
\item $A$ is complete (every Cauchy sequence converges) and can be covered by finitely many subsets with a finite size.
\item Every sequence in $A$ has a convergent subsequence whose limit is in $A$.
\end{itemize}
\end{defn}

Further, due to the Heine--Borel theorem a subset of an Euclidean space $S \subseteq \Reals^n$ is compact if and only if it is closed and bounded. Compact can therefore be regarded as a generalisation of closed and bounded (sub)sets to normed spaces (or when put in an even more general setting, to topological spaces). Most of the time we will work in Euclidean spaces, so compactness can simply be read as closed-and-boundedness. Only when dealing with the bosonic ensemble, we actually need the more general notion as described in its definition~\ref{def:compactSet}.

\begin{defn}[(local) Lipschitz continuity]
Let $X$ be a normed space. A function $f \colon X \to \Reals$ is called (globally) Lipschitz continuous, if there exists a constant $K$ such that for all $x,y \in X$
\begin{equation*}
\abs{f(x) - f(y)} \leq K\norm{x - y} \, .
\end{equation*}
If such a constant can only be found for any compact subspace of $X$, the function is called \emph{locally} Lipschitz continuous.
\end{defn}

To show local Lipschitz continuity of convex functions on finite-dimensional spaces, it is convenient to use Jensen's inequality.

\begin{prop}[Jensen's inequality]
\label{thm:Jensen}
Consider a convex (concave) function $f$. Then for any convex combination
\begin{equation*}
x \in \set*{ \sum_{i=1}^N\lambda_ix_i }{ \lambda_i \geq 0, \sum_{i=1}^N\lambda_i = 1 } \, ,
\end{equation*}
where $N \in \Nats \cup \{\infty\}$,
one has
\begin{equation*}
f(x) \substack{ \leq \\ (\geq) } \sum_{i=1}^N\lambda_if(x_i) \, .
\end{equation*}
\end{prop}

\begin{proof}
Simply apply the definition of a convex (concave) function in definition~\ref{def:convexFunc} repeatedly (induction).
\end{proof}

Further, we will use the following lemma.

\begin{lem}\label{lem:convexBounded}
Let $X$ be a finite-dimensional vector space and $f \colon X \to \Reals \cup \{+\infty\}$ a convex function. The function $f$ is bounded on any compact (closed and bounded) set contained in the interior of its domain, $\interiorOf\domain(f)$.
\end{lem}

\begin{proof}
Consider a simplex, $\Delta \subseteq \interiorOf\domain(f)$
\begin{equation*}
\Delta = \set*{\sum_{i=0}^d\lambda_ix_i }{ \lambda_i \geq 0, \sum_{i=0}^d\lambda_i = 1 } \, ,
\end{equation*}
where $d$ is the dimension of the domain of $f$. By Jensen's inequality we have
\begin{equation*}
f\left(\sum_{i=0}^d\lambda_ix_i\right) \leq \sum_{i=0}^d\lambda_if(x_i) \leq \max_i f(x_i) \, ,
\end{equation*}
where we the maximum exists, because we work in a finite-dimensional space. Since any compact set in the interior of the domain of $f$ can be covered by a finite number of simplexes, $f$ has an upper bound on any compact set in the interior of its domain.

Now we need to show that the upper bound implies also a lower bound. Consider a closed ball around $\bar{x}$, $\closedBall_r(\bar{x}) = \set*{x }{ \norm{x - \bar{x}} \leq r }$, with its radius $r$ sufficiently small such that $\closedBall_r(\bar{x}) \subseteq \interiorOf\domain(f)$. Let $x \in \closedBall_r(\bar{x})$, so also $x' = 2\bar{x} - x \in \closedBall_r(\bar{x})$. Since $\bar{x} = (x + x') / 2$, so by convexity of $f$ we have
\begin{equation*}
f(x) \geq 2f(\bar{x}) - f(x') \geq 2f(\bar{x}) - \max_{\crampedclap{y \in \closedBall_r(\bar{x})}}\;f(y)
\end{equation*}
for all $x \in \closedBall_r(\bar{x})$. Since any compact set can be covered by a finite number of balls, this implies that $f$ is bounded on any compact set in the interior of its domain.
\end{proof}

Now we are ready to proof theorem~\ref{thm:convexLocLip}

\begin{proof}
Consider $\closedBall_r(\bar{x}) \in \interiorOf\domain(f)$. By lemma~\ref{lem:convexBounded} we have that $f$ is bounded on $\closedBall_r(\bar{x})$ by some constant, $\abs{f} \leq C_r$. For any $x \neq x' \in \closedBall_{r/2}(\bar{x})$ extend the line segment from $x$ to $x'$ to the boundary of $\closedBall_r(\bar{x})$ and call this point $x''$, so $\norm{x - x''} = r$ and $\lambda = \norm{x' - x} / \norm{x'' - x} \in (0,1)$. Convexity of $f$ now implies
\begin{equation*}
f(x') - f(x) \leq \lambda\bigl(f(x'') - f(x)\bigr)
\leq \frac{f(x'') - f(x)}{\norm{x'' - x}}\norm{x' - x}
\leq \frac{4C_r}{r}\norm{x' - x} \, .
\end{equation*}
Interchanging the roles of $x'$ and $x$ we find the desired inequality
\begin{equation*}
\abs{f(x') - f(x)} \leq \frac{4C_r}{r}\norm{x' - x} \, . \qedhere
\end{equation*}
\end{proof}

\subsection{Proof of theorem~\ref{thm:existenceDirectionalDerivative}}
\label{ap:existenceDirectionalDerivative}

\begin{proof}
Let $x \in \interiorOf\domain(f)$ and consider the function
\begin{equation*}
\phi(t) = \frac{f(x + h t) - f(x)}{t}, 0 < t \leq \epsilon \, ,
\end{equation*}
where $\epsilon$ is small enough such that $x + \epsilon h \in \interiorOf\domain(f)$. For $0 <\lambda \leq 1$ we have by convexity of $f$ that
$f(x + \lambda h t) \leq (1-\lambda)f(x) + \lambda f(x + ht)$. Hence
\begin{equation*}
\phi(\lambda t) = \frac{f(x + \lambda ht) - f(x)}{\lambda t}
\leq \frac{f(x + ht) - f(x)}{t} = \phi(t)
\end{equation*}
for any $0 <\lambda \leq 1$, so $\phi(t)$ is decreasing as $t \downarrow 0$. Due to the local Lipschitz property of finite-dimensional convex functions (theorem~\ref{thm:convexLocLip}), $\phi(t)$ is bounded from below, so the limit exists.
\end{proof}

\subsection{Proof of theorem~\ref{thm:subdifferential}}
\label{ap:subdifferential}

To proof theorem~\ref{thm:subdifferential} it is convenient to work with an alternative definition of a convex function in terms of its epigraph.

\begin{defn}[Epigraph]
The epigraph of a function $f \colon X \to \Reals \cup \{+\infty\}$ is defined as the set of points lying above its graph
\begin{equation*}
\epigraph(f) \isDefinedAs
\set*{(x,\mu) }{ x \in X, \mu \in \Reals, \mu \geq f(x) } \subseteq X \times \Reals \, .
\end{equation*}
\end{defn}

A function $f$ is convex if and only if its epigraph is a convex set.

Additionally we need the following very intuitive theorem from geometry for which we do not supply a proof.
\begin{thm}[Hyperplane separation]
Let $X$ and $Y$ two nonempty convex sets of $\Reals^n$ such that $\interior(X) \cap \interior(Y) = \emptyset$. Then there exists a nonzero vector $v$ and a real number $c$ such that
\begin{equation*}
\braket{v}{x} \geq c	\qquad \text{and} \qquad
\braket{v}{y} \leq c
\end{equation*}
for all $x \in X$ and $y \in Y$. In other words, the hyperplane $\braket{v}{\cdot} = c$ with normal vector $v$ separates (the interiors of) $X$ and $Y$.
\end{thm}

Further we need the following properties of the directional derivative to be able to establish the part iii) of theorem~\ref{thm:subdifferential}

\begin{prop}
Let $f$ be a convex function over a finite-dimensional space and $x \in \interiorOf\domain(f)$. Then $f'_h(x)$ is a convex positive homogeneous (of degree 1) function of $h$ and for any $y \in \domain(f)$
\begin{equation*}
f(y) \geq f(x) + f'_{y-x}(x) \, .
\end{equation*}
\end{prop}

\begin{proof}
Homogeneity in $h$ is trivially shown by working out for $\tau > 0$
\begin{equation*}
f'_{\tau h} = \lim_{t \downarrow 0}\frac{f(x + \tau h t) - f(x)}{t}
= \tau \lim_{\alpha \downarrow 0}\frac{f(x + h \alpha) - f(x)}{\alpha}
= \tau f'_h(x) \, .
\end{equation*}
Convexity in $h$ follows directly from the convexity of $f$. Indeed, for any $h_1, h_2 \in \Reals^n$ and $\lambda \in [0,1]$ we have
\begin{align*}
f'_{\lambda h_1 + (1-\lambda h_2)}(x)
&= \lim_{t \downarrow} t^{-1}\bigl[f(x + (\lambda h_1 + (1-\lambda h_2))t) - f(x)\bigr] \\
&\leq \lim_{t \downarrow}t^{-1}\bigl[\lambda\bigl(f(x + th_1) - f(x)\bigr) +
(1-\lambda)\bigl(f(x + th_2) - f(x)\bigr)\bigr]
= \lambda f'_{h_1}(x) + (1 - \lambda) f'_{h_2}(x) \, .
\end{align*}
To proof the last part let $t \in (0,1]$, $x,y \in \domain(f)$ and $y_t = (1-t)x + ty$. Hence, by convexity of $f$ we have $f(y_t) \leq t f(y) + (1-t)f(x)$ which can be rearranged as
\begin{equation*}
f(y) \geq f(y_t) + \frac{1-t}{t}\bigl(f(y_t) - f(x)\bigr) \, .
\end{equation*}
By taking the limit $t \downarrow 0$, we find the desired equality.
\end{proof}

Now we are ready to proof theorem~\ref{thm:subdifferential}.

\begin{proof}
i) That $\du f(x)$ is nonempty follows directly from the separating hyperplane theorem for convex sets in finite-dimensional spaces applied to $\epigraph(f)$ and the point $\bigl(x,f(x)\bigr)$.

ii) Closedness and convexity are obvious from its definition~\ref{def:convexSubGrad}. Now we show that $\du f(x)$ is bounded. Since $\du f(x)$ is nonempty, there exists a $(d,-\alpha) \in X \times \Reals$ such that
\begin{equation*}
\braket{d}{y} - \alpha \tau \leq \braket{d}{x} - \alpha f(x)
\end{equation*}
for any $(y,\tau) \in \epigraph(f)$. Since $(x,\tau) \in \epigraph(f)$, we find $\alpha \geq 0$.

We even have that $\alpha > 0$, since $f$ is locally Lipschitz (theorem~\ref{thm:convexLocLip}). If we confine $y \in \closedBall_{\epsilon}(x) \subseteq \interiorOf\domain(f)$ with $\epsilon > 0$, there exists a finite constant $M_{\epsilon}$ such that
\begin{equation*}
\braket{d}{y - x} \leq \alpha\bigl(f(y) - f(x)\bigr) \leq  \alpha M_{\epsilon} \norm{y - x} \, .
\end{equation*}
If we now set $y = x + \epsilon d$ we get $\norm{d}^2 \leq \alpha M_{\epsilon}\norm{d}$. If $\norm{d} \neq 0$, we have $\alpha \geq \norm{d} / M_{\epsilon} > 0$ and otherwise if $\norm{d} = 0$, we have $\alpha > 0$, because $(d, -\alpha) \neq 0$. Thus, we can normalise the normal vector such that $\alpha = 1$ to obtain
\begin{equation*}
\braket{d}{y - x} \leq f(y) - f(x) \, .
\end{equation*}
Without loss, assume $d \neq 0$ and choose $y = x + \epsilon d/\norm{d}$. Then
\begin{equation*}
\epsilon\norm{d} = \braket{d}{y - x} \leq M_{\epsilon} \norm{y - x} = M_{\epsilon}\epsilon \, ,
\end{equation*}
so $\norm{d} \leq M_{\epsilon}$ for any $\epsilon > 0$. Since this inequality applies to any $d \in \du f(x)$, this implies boundedness of $\du f(x)$.

iii) First note that since $f'_0(x) = 0$ identically, we have
\begin{equation*}
f'_h(x) - f'_0(x) = f'_h(x) = \lim_{t \downarrow 0}\frac{f(x + ht) - f(x)}{t} \geq \braket{h}{d} \, ,
\end{equation*}
for any $d \in \du f(x)$. The subdifferential of $f'_h(x)$ therefore exists at $h=0$ and $\du f(x) \subseteq \du_h f'_0(x)$.

Because $f'_h(x)$ is convex in $h$, we have for any $d \in \du_h f'_0(x)$
\begin{equation*}
f'_{y-x}(x) = f'_{y-x}(x) - f'_0(x) \geq \braket{d}{y-x} \, .
\end{equation*}
Hence, for any $y \in \domain(f)$ and $d \in \du_h f'_0(x)$ we can establish the following inequality
\begin{equation*}
f(y) \geq f(x) + f'_{y-x}(x) \geq f(x) + \braket{d}{y-x} \, .
\end{equation*}
Thus, $\du_h f'_0(x) \subseteq \du f(x)$, so by the previous result we have $\du_h f'_0(x) = \du f(x)$.

Let now $d_h \in \du_h f'_h(x)$, so for any $v \in X^*$ and $\tau > 0$
\begin{equation*}
\tau f'_v(x) = f'_{\tau v}(x) \geq f'_h(x) + \braket{d_h}{\tau v - h} \, .
\end{equation*}
Then for $\tau \to \infty$ we find $f'_v(x) \geq \braket{d_h}{v}$, so $d_h \in \du_hf'_0(x) = \du f(x)$. Taking the limit $\tau \to 0$ we obtain $0 \geq f'_h(x) - \braket{d_h}{h}$, so $\braket{d_h}{h} = f'_h(x)$. Hence the directional derivative is attained as the maximum over the subdifferential as stated in part iii) of theorem~\ref{thm:subdifferential}.

iv) First suppose that $\du f(x)$ only contains one element. By part iii) we have $f'_h(x) = \braket{d}{h}$, which is linear in $h$. Hence $f$ is differentiable at $x$ and $\nabla f(x) = d$.

To show the converse, if $d \in \du f(x)$, then by definition
\begin{equation*}
f(y) - f(x) \geq \braket{d}{y - x} \, .
\end{equation*}
Now set $y = x + t h$ with $t > 0$ and divide both sides of the inequality by $t$. Taking the limit $t \downarrow 0$ we obtain
\begin{equation*}
\braket{\nabla f}{h} \geq \braket{d}{h} \, .
\end{equation*}
Since this inequality should be valid for all $h$, we find $d = \nabla f(x)$.
\end{proof}

\section{Non-interacting systems}
\label{ap:nonInteracting}

The partition function of a non-interacting fermionic system is readily calculated by expressing the determinants in the eigenbasis of $\brakket{i}{\hat{h} }{ j}$
\begin{equation}
Z^-_s = \sum_{\crampedclap{n_1,\dotsc,n_{N_b} = 0}}^1
\brakket{n_1,\dotsc,n_{N_b}}{\prod_{i=1}^{N_b}\e^{-\beta\epsilon_i\hat{n}_i}}{n_1,\dotsc,n_{N_b}} 
= \prod_{i=1}^{N_b}\bigl(1 + \e^{-\beta \epsilon_i}\bigr) \, .
\end{equation}
For a bosonic system we get the following result
\begin{equation}
Z^+_s = \sum_{\crampedclap{n_1,\dotsc,n_{N_b} = 0}}^{\infty}
\brakket{n_1,\dotsc,n_{N_b}}{\prod_{i=1}^{N_b}\e^{-\beta\epsilon_i\hat{n}_i}}{n_1,\dotsc,n_{N_b}}
= \prod_{i=1}^{N_b}\sum_{n=0}^{\infty}\bigl(\e^{-\beta \epsilon_i}\bigr)^n 
= \prod_{i=1}^{N_b}\frac{1}{1 - \e^{-\beta\epsilon_i}} \, .
\end{equation}
The grand potential can be worked out as
\begin{equation}
\Omega_s^{\pm} = \pm\frac{1}{\beta}\sum_{i=1}^{N_b}\ln\bigl(1 \mp \e^{-\beta\epsilon_i}\bigr) \, .
\end{equation}
The occupation numbers are readily found as the 1RDM is diagonal in this basis
\begin{equation}
n^{\pm}_i = \frac{\du\Omega_s}{\du\epsilon_i}
= \frac{\e^{-\beta\epsilon_i}}{1 \mp \epsilon^{-\beta\epsilon_i}}
= \frac{1}{\epsilon^{\beta\epsilon_i} \mp 1} \, ,
\end{equation}
which can be inverted to yield the NO energies as functions of the occupation numbers
\begin{equation}
\epsilon^{\pm}_i = \frac{1}{\beta}\ln\biggl(\frac{1 \pm n_i}{n_i}\biggr) \, .
\end{equation}
We can insert this expression back into the grand potential to obtain it as a function of the occupation numbers
\begin{equation}
\Omega^{\pm}_s = \mp\frac{1}{\beta}\sum_{i=1}^{N_b}\ln(1 \pm n_i) \, .
\end{equation}
The energy is can be calculated as
\begin{equation}
E^{\pm}_s = \sum_{i=1}^{N_b}n_i\epsilon_i
= \frac{1}{\beta}\sum_{i=1}^{N_b}n_i\ln\biggl(\frac{1 \pm n_i}{n_i}\biggr) \, .
\end{equation}
The entropy is readily obtained by subtracting the grand potential from the energy
\begin{equation}
S^{\pm}_s = \beta\bigl(E^{\pm}_s - \Omega^{\pm}_s\bigr)
= \sum_{i=1}^{N_b}\bigl[(n_i \pm 1)\ln(1 \pm n_i) - n_i\ln(n_i)\bigr] \, .
\end{equation}
As we have now also the energy and entropy explicitly, the non-interacting universal function is readily constructed to be
\begin{equation}
F^{\pm}_s 
= E^{\pm}_{s,0} - \frac{1}{\beta}S_s^{\pm}
= \sum_{i=1}^{N_b}\bigl[n_i\bigl(\epsilon^{s,0}_i + \frac{1}{\beta}\ln(n_i)\bigr) -
\frac{1}{\beta}(n_i \pm 1)\ln(1 \pm n_i)\bigr] \, , 
\end{equation}
where $\epsilon^{s,0}_i$ are the eigenvalues of the reference one-body hamiltonian $h^{(1)}_{s,0}$. Since now the trace of the 1RDM can be given in terms of the natural occupation numbers $n_i$ (see Sec.~\ref{thm:Coleman}) we find the expressions presented in Sec.~\ref{sec:general}.

\section{Solving a general non-interacting system}
\label{ap:non-interatingSolution}
In this appendix we solve a general non-interacting system including both a source term and a pairing field. Hence, the Hamiltonian under consideration is of the general form
\begin{equation}
\hat{H} = \sum_{ij}\omega_{ij}\crea{a}_i\anni{a}_j +
\sum_i\bigl(h_i^*\crea{a}_i + h_i^{\vphantom{*}}\anni{a}_i\bigr) +
\sum_{ij}\bigl(D^{\dagger}_{ij}\crea{a}_i\crea{a}_j +
D^{\vphantom{\dagger}}_{ij}\anni{a}_i\anni{a}_j\bigr) \, ,
\end{equation}
where $D^T = \pm D$ for bosons (upper sign) and fermions (lower sign).
The first step is to transform the source term away. This step is identical for the both the bosonic and the fermionic case. This is readily done by adding a constant to the annihilation and creation operators
\begin{equation}
\anni{b}_i = \anni{a}_i + \tilde{h}_i^* \qquad
\Rightarrow \qquad
\crea{b}_i = \crea{a}_i + \tilde{h}_i^{\vphantom{*}} \, .
\end{equation}
The vector $\tilde{h}$ should be chosen such that the source term disappears, so
\begin{equation}
\hat{H} + C_h = \sum_{ij}\omega_{ij}\crea{b}_i\anni{b}_j +
\sum_{ij}\bigl(D^{\dagger}_{ij}\crea{b}_i\crea{b}_j +
D^{\vphantom{\dagger}}_{ij}\anni{b}_i\anni{b}_j\bigr) \, .
\end{equation}
One readily finds that the vector $\tilde{h}$ needs to satisfy the following linear equation
\begin{equation}
\sum_j\Bigl(\omega_{ij}\tilde{h}^{\vphantom{*}}_j + \bigl(D \pm D^T\bigr)_{ij}\tilde{h}^*_j\Bigr)
= h_i \, .
\end{equation}
This system is guaranteed to be solvable as the effective matrices will be normal (symmetric for the real part of $\tilde{h}$ and anti-symmetric for the imaginary part) and the assumed positivity of the spectrum, as we want the system to have a ground state.
The corresponding constant shift in the Hamiltonian will be
\begin{equation}
C_h
= \tilde{h}^{\dagger}\omega\tilde{h} + \tilde{h}^{\dagger}D\tilde{h}^* + \tilde{h}D^{\dagger}\tilde{h}
= \thalf\bigl(\tilde{h}^{\dagger}h + h^{\dagger}\tilde{h}\bigr) \, .
\end{equation}
To transform the pairing field away, we need a general Bogoliubov transform \citep{Bogoliubov1947, Valatin1958, Bogoljubov1958}.
The Bogoliubov transform is a generalisation of a unitary transformation between the one-particle states to linear combinations of creation and annihilation operators
\begin{subequations}
\begin{align}
\anni{c}_k &= \sum_{r=1}^{N_b}\bigl(U_{kr}\anni{b}_r + V_{kr}\crea{b}_r\bigr) \, , \\
\crea{c}_k &= \sum_{r=1}^{N_b}\bigl(U^*_{kr}\crea{b}_r + V^*_{kr}\anni{b}_r\bigr) \, .
\end{align}
\end{subequations}
The transformation between the annihilation and creation operators can be written in a more compact manner by collecting the creation and annihilation operators in a column vector. This allows us write the Bogoliubov transformation as
\begin{equation}
\begin{pmatrix} \anni{c} \\ \crea{c} \end{pmatrix}
= \begin{pmatrix} U &V \\ V^* & U^* \end{pmatrix}
\begin{pmatrix} \anni{b} \\ \crea{b} \end{pmatrix} \, ,
\end{equation}
where $U$ and $V$ are $N_b \times N_b$ matrices.
By working out the (anti-)commutation relations for bosons (fermions), one finds that the Bogoliubov transformation needs to satisfy
\begin{equation}\label{eq:unitaryCond}
\begin{pmatrix} U &V \\ V^* & U^* \end{pmatrix}
\begin{pmatrix} \unitMat &0 \\ 0 & \mp\unitMat \end{pmatrix}
\begin{pmatrix} U^{\dagger} &V^T \\ V^{\dagger} & U^T \end{pmatrix}
= \begin{pmatrix} \unitMat &0 \\ 0 & \mp\unitMat \end{pmatrix} \, ,
\end{equation}
where $\unitMat$ is the unit matrix and the upper (lower) sign refers to bosons (fermions) respectively. Thus we find that the Bogoliubov transform is unitary for fermions with respect to the Euclidian metric, so corresponds to an element of the definite unitary group: U$(2N_b)$. On the other hand, the bosonic transformation is unitary with respect to an indefinite metric and corresponds to an element of the indefinite unitary group: U$(N_b,N_b)$.

With the help of the commutation relation \([\anni{a}_k, \crea{a}_l]_{\mp} = \delta_{kl}\), the Hamiltonian can be rewritten in the following form
\begin{equation}
\hat{H} + C_h = \begin{pmatrix} \crea{b} & \anni{b} \end{pmatrix}
\begin{pmatrix} \omega/2 & D^{\dagger} \\ D & \pm\omega/2\end{pmatrix}
\begin{pmatrix} \anni{b} \\ \crea{b} \end{pmatrix} \mp \half\Trace\{\omega\} \, .
\end{equation}
Inserting the unit matrix on both sides of the matrix and using~\eqref{eq:unitaryCond}, we find 
\begin{equation}
\hat{H} + C_h = \begin{pmatrix} \crea{c} & \anni{c} \end{pmatrix}
\begin{pmatrix} U &\mp V \\ \mp V^* & U^* \end{pmatrix}
\begin{pmatrix} \omega/2 & D^{\dagger} \\ D & \pm\omega/2\end{pmatrix}
\begin{pmatrix} U^{\dagger} & \mp V^T \\ \mp V^{\dagger} & U^T\end{pmatrix}
\begin{pmatrix} \anni{c} \\ \crea{c} \end{pmatrix} \mp \half\Trace\{\omega\} \, .
\end{equation}
So to bring the Hamiltonian to diagonal form, we simply need to diagonalise it with respect to the appropriate metric
\begin{equation}\label{eq:eigenEQ}
\begin{pmatrix} \omega/2 & D^{\dagger} \\ D & \pm\omega/2\end{pmatrix}
\begin{pmatrix} U^{\dagger} & \mp V^T \\ \mp V^{\dagger} & U^T\end{pmatrix}
= \begin{pmatrix} \unitMat &0 \\ 0 & \mp\unitMat \end{pmatrix}
\begin{pmatrix} U^{\dagger} & \mp V^T \\ \mp V^{\dagger} & U^T\end{pmatrix}
\begin{pmatrix}\mathcal{E}/2 & 0 \\ 0 & \tilde{\mathcal{E}}/2 \end{pmatrix} \, .
\end{equation}
where $\mathcal{E}$ and $\tilde{\mathcal{E}}$ are diagonal.

Now let us derive some properties of the eigenvalues and eigenvectors. For the $k$-th eigenvector, we can work out the eigenvalue equation as
\begin{equation}
\begin{pmatrix} \omega/2 & D^{\dagger} \\ D & \pm\omega/2\end{pmatrix}
\begin{pmatrix} u_k \\ v_k\end{pmatrix}
= \epsilon_k\begin{pmatrix} u_k \\ \mp v_k\end{pmatrix}
\end{equation}
Since the matrix is hermitian, we have
\begin{equation}
\epsilon_l\begin{pmatrix} u^{\dagger}_k & v^{\dagger}_k\end{pmatrix}
\begin{pmatrix} u_l \\ \mp v_l\end{pmatrix}
= \begin{pmatrix} u^{\dagger}_k & v^{\dagger}_k\end{pmatrix}
\begin{pmatrix} \omega/2 & D^{\dagger} \\ D & \pm\omega/2\end{pmatrix}
\begin{pmatrix} u_l \\ v_l\end{pmatrix}
= \epsilon_k^*\begin{pmatrix} u^{\dagger}_k & \mp v^{\dagger}_k\end{pmatrix}
\begin{pmatrix} u_l \\ v_l\end{pmatrix} \,.
\end{equation}
This expression can be rearranged to yield
\begin{equation}\label{eq:hermEigvalCond}
(\epsilon_k^* - \epsilon_l)\braket{k}{l}_{\mp} = 0 \, ,
\end{equation}
where $\braket{k}{l}_{\mp}$ denotes the inner product between the two eigenvectors with respect to the indefinite metric for bosons ($-$) and with respect to the usual Euclidean metric for fermions ($+$). Now let us first consider the case $k=l$. In the fermionic case we have a proper metric, so $\braket{k}{k}_{+} = 0$ only for the zero vector, which is no eigenvector. Hence, we find that the eigenvalues need to be real in the fermionic case.

In the bosonic case, however, we have an indefinite metric, so $\braket{k}{k}_{-} = 0$ is also possible for a non-zero vector. so we need to distinguish two cases
\begin{subequations}
\begin{align}
\braket{k}{k}_{-} \neq 0	\qquad&\Rightarrow\qquad	\epsilon_k \in \Reals \, , \\
\braket{k}{k}_{-} = 0		\qquad&\Leftarrow\qquad		\epsilon_k \notin \Reals \, .
\end{align}
\end{subequations}
Now let us consider the case $k \neq l$. In the fermionic case condition~\eqref{eq:hermEigvalCond} implies that non-degenerate eigenvectors are orthogonal, as the eigenvalues are real. As only degenerate eigenvectors may be non-orthogonal, we can always orthogonalise them, as any linear combination degenerate eigenvectors is again an eigenvector.

The situation is again more complicated in the bosonic situation. For the eigenvectors with real eigenvalues and finite norm we get the same result as in the fermionic case: non-degenerate eigenvectors are orthogonal and degenerate eigenvectors can be chosen to be orthogonal.

As the eigenvectors are related in pairs, one expects the eigenvalues $\mathcal{E}/2$ and $\tilde{\mathcal{E}}/2$ to be related. This is indeed the case.
To establish this relationship, will assume $\omega$ to be a real diagonal matrix. If it is not diagonal, it can always be brought to diagonal form by a simply unitary transformation and its eigenvalues will be real, as the matrix is hermitian.
Now we work out the eigenvalue equation for the first set of eigenvectors
\begin{subequations}
\begin{equation}\label{eq:eigVal1}
\begin{split}
\thalf\omega\,U^{\dagger} - D^*V^{\dagger} &= U^{\dagger}\mathcal{E}/2 \, , \\
D\,U^{\dagger} - \thalf\omega\,V^{\dagger} &= V^{\dagger}\mathcal{E}/2 \, .
\end{split}
\end{equation}
By taking the complex conjugate of the second set of eigenvectors, we find
\begin{equation}
\begin{split}
\thalf\omega\,U^{\dagger} - D^* V^{\dagger} &= -U^{\dagger}\tilde{\mathcal{E}}^*/2 \, , \\
D\,U^{\dagger} - \thalf\omega\,V^{\dagger} &= - V^{\dagger}\tilde{\mathcal{E}}^*/2 \, ,
\end{split}
\end{equation}
\end{subequations}
where we used that $\omega$ can be assumed to be real (and diagonal).
Hence, we find that $\tilde{\mathcal{E}} = -\mathcal{E}^*$.

After solving the (generalised) eigenvalue equation~\eqref{eq:eigenEQ}, by using the commutation relations again, the Hamiltonian can be written as
\begin{equation}
\hat{H} + C_h
= \begin{pmatrix} \crea{c} & \anni{c} \end{pmatrix}
\begin{pmatrix} \mathcal{E}/2 & 0 \\ 0 & \pm\mathcal{E}^*/2 \end{pmatrix}
\begin{pmatrix} \anni{c} \\ \crea{c} \end{pmatrix} \mp \half\Trace\{\omega\}
= \crea{c}\,\Re\mathcal{E} \,\anni{c} \pm \half\Trace\{\mathcal{E}^* - \omega\} \, .
\end{equation}
As the spectrum of $\hat{H}$ should be real, we see that complex $\mathcal{E}$ is not permissible in the bosonic case. It simply means that the Hamiltonian under consideration is not self-adjoint. Using the inequalities in Table~\ref{tab:inequalities}, we can put some sufficient inequalities on the matrix elements $\omega$ and $D$ for $\mathcal{E}$ to be real.

It would be desirable to only calculate one set of the eigenvectors, so we need to cut the dimension of the eigenvalue problem down by a factor two.
If the matrix $D$ only contains real entries, this is readily achieved by adding and subtracting the equations in~\eqref{eq:eigVal1}, which yields
\begin{equation}
(\omega/2 \pm D)\bigl(U^{\dagger} \mp V^{\dagger}\bigr)
= \bigl(U^{\dagger} \pm V^{\dagger}\bigr)\mathcal{E}/2 \, .
\end{equation}
We can now eliminate the even or odd combination by multiplying from the right by $\mathcal{E}/2$, and substituting for the unwanted combination, which yields
\begin{equation}
(\omega/2 \pm D)(\omega/2 \mp D)\bigl(U^{\dagger} \pm V^{\dagger}\bigr)
= \bigl(U^{\dagger} \pm V^{\dagger}\bigr)\mathcal{E}^2/4 \, .
\end{equation}
In the case that the pairing matrix $D$ has complex entries, we can always find a unitary matrix to make it real. As the matrix $D$ is symmetric for bosons and anti-symmetric for fermions, we need to show this for both cases separately. Let us first consider the bosonic case.

\begin{thm}(bosonic case)\label{thm:bosonDiag}
Given a symmetric matrix $D \in \Complex^n \times \Complex^n$, it can be brought to diagonal form by the transformation $U^TCU$, where $U$ is a unitary matrix which diagonalises $C^{\dagger}C$. The diagonal entries can be chosen to be the square root of the eigenvalues of $C^{\dagger}C$.
\end{thm}

\begin{proof}
The matrix product $C^{\dagger}C$ is obviously hermitian and also positive semidefinite. Therefore, it is diagonalizable by a unitary matrix $U$ and has $a_i \in \Reals_+$ as its eigenvalues (spectral theorem)
\begin{equation*}
a_i\delta_{ij} = \bigl(U^{\dagger}C^{\dagger}CU\bigr)_{ij}
= \bigl(U^{\dagger}C^{\dagger}U^*\,U^TCU\bigr)_{ij}
= \bigl(\tilde{C}^{\dagger}\tilde{C}\bigr)_{ij} = \sum_k\tilde{C}^*_{ik}\tilde{C}_{jk} \, ,
\end{equation*}
where $\tilde{C} = U^TCU = \tilde{C}^T$.
Now multiplying by $\tilde{C}^*_{jl}$ and summing over $j$ we find
\begin{equation*}
a_i\tilde{C}^*_{il} = \sum_ja_i\delta_{ij}\tilde{C}^*_{jl}
= \sum_{kj}\tilde{C}^*_{ik}\tilde{C}_{jk}\tilde{C}^*_{jl}
= \sum_k\tilde{C}^*_{ik}a_k\delta_{kl}
= \tilde{C}^*_{il}a_l \, ,
\end{equation*}
which can be rearranged to
\begin{equation*}
(a_i - a_j)\tilde{C}_{ij} = 0 \qquad \text{for all $i,j$.}
\end{equation*}
So if $C^{\dagger}C$ only has non-degenerate eigenvalues, we find that $\tilde{C}$ needs to be diagonal with diagonal entries $\sqrt{a_i}\e^{\im\phi_i}$, where $\phi_i$ is complex phase factor which is undetermined in the diagonalization of $C^{\dagger}C$. So we can choose $\phi_i = 0$ to make the matrix $\tilde{C}$ real and positive semidefinite.

In the case some of the eigenvalues of $C^{\dagger}C$ are degenerate, $\tilde{C}$ is only block diagonal. So we need to show that we can each of these blocks can be brought to diagonal form by $Q^T\tilde{C}Q$. Let $B=B^T$ denote one of these degenerate blocks. Such a degenerate block has the special property that $B^{\dagger}B = a \unitMat$, where $\unitMat$ denotes the unit matrix. This implies that $B$ is a normal matrix so $B^{\dagger}B = BB^{\dagger}$. Splitting the real and imaginary parts as $B = R + \im I$, we can work this out as
\begin{equation*}
0 = B^{\dagger}B - BB^{\dagger}
= (R - \im I)(R + \im I) - (R + \im I)(R - \im I)
= 2 \im(RI - IR) = 2\im[R,I] \, .
\end{equation*}
As the real and imaginary parts of $B$ commute they can be brought to diagonal form by the same orthogonal transformation $Q$. Hence, also $B$ will be brought to diagonal form by the same orthogonal transformation $Q$
\begin{equation*}
Q^TBQ = \sqrt{a}\diag\bigl(\e^{\im\phi_i}\bigr) \, .
\end{equation*}
The phase factors can be transformed away by the remaining freedom, i.e.\ $Q \to Q \diag\bigl(\e^{-\im\phi_i/2}\bigr)$.
\end{proof}

We see that Theorem~\ref{thm:bosonDiag} even shows that $D$ can be brought to a diagonal and real form, so the Hamiltonian can be assumed to be of the following simple form
\begin{equation}
\hat{H} + C_h = \sum_{ij}\omega_{ij}\crea{b}_i\anni{b}_j +
\sum_{i}d_{ii}\bigl(\crea{b}_i\crea{b}_i + \anni{b}_i\anni{b}_i\bigr) \, ,
\end{equation}
where $d_i \in \Reals_+$ are the eigenvalues of $D$. It is therefore tempting to perform the Bogoliubov transform for each one-particle state
\begin{equation}
\anni{c}_i = \cosh(\theta_i)\anni{b}_i + \sinh(\theta_i)\crea{b}_i \, ,
\end{equation}
where $2\theta_i = \arctanh(2d_i/\omega_{ii})$. Unfortunately, the resulting cross-terms give rise to a new pairing field, so this method does not work. Now let us prove a similar theorem for the fermionic case.

\begin{thm}(fermionic case)\label{thm:fermionDiag}
Given an anti-symmetric matrix $D \in \Complex^n \times \Complex^n$, it can be brought to $2\times2$ block-diagonal form by the transformation $U^TCU$, where $U$ is a unitary matrix which diagonalises $C^{\dagger}C$. The off-diagonal entries can be chosen to be the square root of the eigenvalues of $C^{\dagger}C$. If the dimension is odd, a 0 column and row should be added to the $2\times2$ block-diagonal form.
\end{thm}

\begin{proof}
The proof is basically the same as the proof of Theorem~\ref{thm:bosonDiag}, though $\tilde{C} = U^TCU = -\tilde{C}^T$, as $C$ is now anti-symmetric and we also have
\begin{equation*}
a_i\tilde{C}^*_{il} = \tilde{C}^*_{il} a_l \, .
\end{equation*}
However, as $\tilde{C}$ is anti-symmetric, $\tilde{C}_{ii} = 0$, so it cannot be diagonal. Hence we need at least a two-fold degeneracy in all the eigenvalues of $C^{\dagger}C$. For the moment, assume that we only have a two-fold degeneracy. The off-diagonal elements of $\tilde{C}$ are $\sqrt{a_i}\e^{\im\phi_i}$, where $\phi_i$ is complex phase factor which is undetermined in the diagonalization of $C^{\dagger}C$. So we can choose $\phi_i = 0$ to make the matrix $\tilde{C}$ real.

In the case of higher order degeneracies, we can use the same argument as in the symmetric case. Let $B = -B^T$ denote one of these degenerate blocks. Again such a block is normal, so the real and imaginary parts commute, so can be brought to block diagonal form by the same orthogonal transformation. Further, the eigenvalues of anti-symmetric matrices come in pairs, so the degeneracy of the eigenvalues $a_i > 0$ can only be even. If the degeneracy would be odd, $B$ would have at least one zero eigenvalue, which would correspond to $a_i = 0$. So only the block corresponding to $a_i = 0$ can have odd dimensionality.
\end{proof}

\section{Additional proofs for the bosonic case}

\subsection{Proof of Klein's inequality (Thm.~\ref{thm:Klein})}
\label{ap:KleinProof}
\begin{proof}
Given a bounded hermitian operator $\hat{B}$ and a convex (concave) function $f \colon \Reals \to \Reals$, we have by Jensen's inequality (Proposition~\ref{thm:Jensen})
\begin{equation*}
\brakket{\phi}{f(\hat{B})}{\phi}
= \sum_i\braket{\phi}{\psi_i}f(b_i)\braket{\psi_i}{\phi}
\;\substack{\geq \\ (\leq)}\; f\biggl(\sum_i\abs{\braket{\phi}{\psi_i}}^2b_i\biggr)
= f\bigl(\brakket{\phi}{\hat{B}}{\phi}\bigr) \, .
\end{equation*}
Because convex (concave) function always has a subgradient (see Definition~\ref{def:convexSubGrad} and Theorem~\ref{thm:nonEmptySubDif}), we have
\begin{equation*}
f(y) - f(x) \;\substack{\geq \\ (\leq)}\; (y - x)f'(x) \, ,
\end{equation*}
where $f'(x) \in \du f(x)$. So for all eigenvectors $\phi_i$ of the operator $\hat{A}$ with eigenvalues $\alpha_i$, we have
\begin{align*}
\Trace\bigl\{f(\hat{B}) - f(\hat{A})\bigr\}
&= \sum_i\bigl(\brakket{\phi_i}{f(\hat{B})}{\phi_i} - \brakket{\phi_i}{f(\hat{A})}{\phi_i}\bigr) \\
&\substack{\geq \\ (\leq)}\;
\sum_i\bigl(f(\brakket{\phi_i}{\hat{B}}{\phi_i} - \brakket{\phi_i}{f(\alpha_i)}{\phi_i}\bigr) \\
&\substack{\geq \\ (\leq)}\; \sum_i(\brakket{\phi_i}{\hat{B}}{\phi_i} - \alpha_i)f'(\alpha_i)
= \Trace\bigl\{(\hat{B} - \hat{A})f'(\hat{A})\bigr\} \, . \qedhere
\end{align*}
\end{proof}

\subsection{The relative entropy as a limit (Eq.~\eqref{eq:altDefRelS})}
\label{ap:ProofRelSasLim}
Here we will present the proof that the relative entropy can be expressed as a commutator as in Ref.~\cite[Lemma 4]{Lindblad1973}. A very brief sketch can also be found in~\cite[Eq.~(3.8)]{Wehrl1978}. We will change the notation slightly and aim to show that
\begin{equation}\label{eq:altLimDefRelS}
\lim_{\mathclap{\lambda \to 0}}\lambda^{-1}S_{\lambda}[\hat{\rho}_1 | \hat{\rho}_0]
= S[\hat{\rho}_1 | \hat{\rho}_0] \, .
\end{equation}

\begin{proof}
First we rewrite $S_{\lambda}$ as
\begin{equation*}
S_{\lambda}[\hat{\rho}_1 | \hat{\rho}_0]
= \lambda S_{\lambda}[\hat{\rho}_1 | \hat{\rho}_{\lambda}] +
(1-\lambda)S_{\lambda}[\hat{\rho}_0 | \hat{\rho}_{\lambda}] \, ,
\end{equation*}
where $\hat{\rho}_{\lambda} = \lambda\hat{\rho}_1 + (1-\lambda)\hat{\rho}_0$.

First we note that
\begin{equation*}
\lim_{\mathclap{\lambda \to 0}}S[\hat{\rho}_1 | \hat{\rho}_{\lambda}]
= S[\hat{\rho}_1 | \hat{\rho}_0] \, .
\end{equation*}
This can be seen by considering the partial sums
\begin{equation*}
g_n(\lambda) = \sum_{i=1}^n\bigl(w_i\ln(w_i) - w_i\brakket{\Psi_i}{\ln(\hat{\rho}_{\lambda})}{\Psi_i} +
\brakket{\Psi_i}{\hat{\rho}_{\lambda}}{\Psi_i} - w_i\bigr) \, ,
\end{equation*}
where $w_i$ and $\ket{\Psi_i}$ are the eigenvalues and eigenstates respectively of $\hat{\rho}_1$.
The functions $g_n(\lambda)$ are continuous in $\lambda = 0$: \(g_n(0) = \lim_{\lambda \to 0}g_n(\lambda)\), because \(\ker(\hat{\rho}_1) \subseteq \ker(\hat{\rho}_0)\). As $\ln(x)$ is concave, the functions $g_n(\lambda)$ are convex in $\lambda$ and $\{g_n(\lambda)\}$ form a monotonic non-decreasing sequence, \(g_n(\lambda) \to g(\lambda) = S[\hat{\rho}_1 | \hat{\rho}_{\lambda}]\), since each term is non-negative due the Klein's inequality, cf.~\eqref{eq:Klein}. Hence, $\lim_{\lambda \to 0}g(\lambda) = g(0)$ is unique.
Using the same arguments, we also have
\begin{equation*}
\lim_{\mathclap{\lambda \to 0}}S[\hat{\rho}_0 | \hat{\rho}_{\lambda}]
= S[\hat{\rho}_0 | \hat{\rho}_0] = 0 \, .
\end{equation*}
From convexity of $x\ln(x)$ it follows that $S_{\lambda}[\hat{\rho}_1 | \hat{\rho}_0]$ is concave in $\lambda$, so $\lambda^{-1}S_{\lambda}[\hat{\rho}_1 | \hat{\rho}_0]$ is monotonically increasing when $\lambda \to 0$ (see Sec.~\ref{ap:existenceDirectionalDerivative}), so \(\lim_{\lambda \to 0}\lambda^{-1}S_{\lambda}[\hat{\rho}_1 | \hat{\rho}_0]\) is uniquely defined. This implies that also \(\lim_{\lambda \to 0}(\lambda^{-1} - 1)S[\hat{\rho}_0 | \hat{\rho}_{\lambda}] \geq 0\) exists and obviously
\begin{equation}\label{eq:lowerBoundRelEntropy}
\lim_{\mathclap{\lambda \to 0}}\lambda^{-1}S_{\lambda}[\hat{\rho}_1 | \hat{\rho}_0]
\geq S[\hat{\rho}_1 | \hat{\rho}_0] \, .
\end{equation}
If $S[\hat{\rho}_1 | \hat{\rho}_0] = \infty$, Eq.~\eqref{eq:altLimDefRelS} is correct. If not, we can write
\begin{equation*}
S_{\lambda}[\hat{\rho}_1 | \hat{\rho}_0]
= \lambda S[\hat{\rho}_1 | \hat{\rho}_0] - S[\hat{\rho}_{\lambda} | \hat{\rho}_0] \, ,
\end{equation*}
so we find
\begin{equation*}
\lim_{\mathclap{\lambda \to 0}}\lambda^{-1}S_{\lambda}[\hat{\rho}_1 | \hat{\rho}_0]
= S[\hat{\rho}_1 | \hat{\rho}_0] - 
\lim_{\mathclap{\lambda \to 0}}\lambda^{-1}S[\hat{\rho}_{\lambda} | \hat{\rho}_0]
\leq S[\hat{\rho}_1 | \hat{\rho}_0] \, ,
\end{equation*}
so combined with its lower bound~\eqref{eq:lowerBoundRelEntropy} we find the required equality.
As we have shown that $\lambda^{-1}S_{\lambda}[\hat{\rho}_1 | \hat{\rho}_0]$ is monotonically increasing for $\lambda \to 0$,we can replace the limit in~\eqref{eq:altLimDefRelS} by the supremum in~\eqref{eq:altDefRelS}.
\end{proof}

\addcontentsline{toc}{section}{References}
\section*{References}
\bibliographystyle{elsarticle-harv}
\bibliography{bible}

\end{document}